\documentclass{article}

\usepackage[english]{babel}
\usepackage[utf8]{inputenc}
\usepackage[T1]{fontenc}
\usepackage[a4paper, margin=1in]{geometry}
\usepackage{graphicx}
\graphicspath{{figures/}}
\usepackage{framed}
\usepackage[framemethod=tikz]{mdframed}
\usepackage{todonotes}
\usepackage{color}

\newcommand{\newtext}[1]{{\color{black} #1}}

\definecolor{darkgreen}{rgb}{0,0.5,0}
\usepackage{hyperref}
\hypersetup{
    unicode=false,          
    colorlinks=true,        
    linkcolor=red,          
    citecolor=darkgreen,        
    filecolor=magenta,      
    urlcolor=cyan           
}

\usepackage{amsthm}
\usepackage{amsmath}
\usepackage{amssymb}
\usepackage{amsfonts}
\usepackage{mathrsfs}
\usepackage{mathtools}
\usepackage{verbatim}
\usepackage{footnote}
\usepackage[linesnumbered,boxed, vlined]{algorithm2e}
	\DontPrintSemicolon
  \SetKwInOut{Input}{Input}
  \SetKwInOut{Output}{Output}
\usepackage{lineno}
\usepackage{caption}
\usepackage{framed}
\usepackage{enumerate}
\usepackage{slashbox}

\usepackage{tikzsymbols}
\usepackage{thmtools,thm-restate}
\usepackage{nicefrac}

\usepackage{subcaption}
\usepackage{makecell}

\usepackage[capitalize, nameinlink]{cleveref}
\usepackage[section]{placeins}
\crefname{theorem}{Theorem}{Theorems}
\Crefname{lemma}{Lemma}{Lemmas}
\Crefname{invariant}{Invariant}{Invariants}
\Crefname{claim}{Claim}{Claims}
\Crefname{observation}{Observation}{Observations}
\Crefname{algorithm}{Algorithm}{Algorithms}
\Crefname{figure}{Figure}{Figures}

\newtheorem{theorem}{Theorem}[section]
\newtheorem{lemma}[theorem]{Lemma}

\newtheorem{definition}[theorem]{Definition}
\newtheorem{invariant}[theorem]{Invariant}
\newtheorem{observation}[theorem]{Observation}

\newtheorem*{remark*}{Remark}

\newcommand{\eps}{\varepsilon}
\newcommand{\eqdef}{\stackrel{\text{\tiny\rm def}}{=}}

\newcommand{\rb}[1]{\left( #1 \right)}

\newcommand{\astar}{{a^{\star}}}

\newcommand{\afirst}{a_{\rm{first}}}

\newcommand{\htau}{\hat{\tau}}

\newcommand{\taustar}{\tau^{\star}}

\newcommand{\jbig}{j_{\rm{big}}}
\newcommand{\jstar}{j^\star}
\newcommand{\Path}[1]{P_{\rm{#1}}}

\newcommand{\Vbase}{V_{\rm{base}}}
\newcommand{\Vnew}{V_{\rm{new}}}

\DeclareRobustCommand*{\leftarc}[1]{\overleftarrow{#1}}

\newcommand{\structure}{\ensuremath{\mathcal{S}}\xspace}
\newcommand{\maxlen}{\ensuremath{\ell_{\textrm{max}}}\xspace}
\newcommand{\maxtau}{\tau_{\textrm{max}}}
\newcommand{\hH}{{\hat{H}}}
\newcommand{\tO}{\tilde{O}}

\newcommand{\cA}{\mathcal{A}}

\newcommand{\cY}{\mathcal{Y}}
\newcommand{\cYstar}{\mathcal{Y}^{\star}}

\newcommand{\true}{\textsc{true}\xspace}
\newcommand{\false}{\textsc{false}\xspace}

\newcommand{\size}{\textsc{size}\xspace}

\newcommand{\AlgPhase}{\textsc{Alg-Phase}\xspace}
\newcommand{\AlgExtendStructures}{\textsc{Extend-Active-Paths}\xspace}
\newcommand{\AlgEdgeMerge}{\textsc{Check-for-Edge-Augmentation}\xspace}
\newcommand{\AlgAugmentStructures}{\textsc{Include-Unmatched-Edges}\xspace}

\newcommand{\AlgMerge}{\textsc{Augment-and-Clean}\xspace}
\newcommand{\AlgBacktrack}{\textsc{Backtrack-Stuck-Structures}\xspace}
\newcommand{\AlgOvertake}{\ensuremath{\textsc{Reduce-Label-and-Overtake}}\xspace}

\newcommand{\PassBundle}{\ensuremath{\textsc{Pass-Bundle}}\xspace}
\newcommand{\storage}{{\textsc{Storage}}\xspace}

\DeclareMathOperator{\RD}{\textsc{RD}\xspace}
\DeclareMathOperator{\poly}{poly}
\newcommand{\sizelimit}{\textsc{limit}\xspace}
\newcommand{\phases}{T}
\newcommand{\Ssize}{\Delta}
\newcommand{\eat}[1]{}

\newcommand{\CONGEST}{\textsf{CONGEST}}

\begin{document}

\title{Deterministic $(1+\eps)$-Approximate Maximum Matching with $\poly(1 / \eps)$ Passes in the Semi-Streaming Model and Beyond}

\author{
	 Manuela Fischer\\
   ETH Zurich \\
   manuela.fischer@inf.ethz.ch\and
	 Slobodan Mitrovi\' c
	\thanks{This work was supported by the Swiss NSF grant No.~P400P2\_191122/1,  MIT-IBM Watson AI Lab and research collaboration agreement No. W1771646, NSF award CCF-1733808, and FinTech@CSAIL. Most of this work was done while the author was affiliated with MIT.}
	\\
   UC Davis \\
	 smitrovic@ucdavis.edu
	\and
	 Jara Uitto \thanks{This work was supported in part by the Academy of Finland, Grant 334238.}\\
   Aalto University \\
	 jara.uitto@aalto.fi
 }
 
\date{}

\maketitle

\begin{abstract}
We present a deterministic $(1+\varepsilon)$-approximate maximum matching algorithm in $\poly 1/\varepsilon$ passes in the semi-streaming model, solving the long-standing open problem of breaking the exponential barrier in the dependence on $1/\varepsilon$. Our algorithm exponentially improves on the well-known randomized $(1/\varepsilon)^{O(1/\varepsilon)}$-pass algorithm from the seminal work by McGregor~[APPROX05], the recent deterministic algorithm by Tirodkar with the same pass complexity~[FSTTCS18]. Up to polynomial factors in $1/\varepsilon$, our work matches the state-of-the-art deterministic $(\log n / \log \log n) \cdot (1/\varepsilon)$-pass algorithm by Ahn and Guha~[TOPC18], that is allowed a dependence on the number of nodes $n$. Our result also makes progress on the Open Problem~60 at sublinear.info\footnote{For details, please visit \href{https://sublinear.info/index.php?title=Open_Problems:60}{https://sublinear.info/index.php?title=Open\_Problems:60}.}.

Moreover, we design a general framework that simulates our approach for the streaming setting in other models of computation. This framework requires access to an algorithm computing an $O(1)$-approximate maximum matching and an algorithm for processing disjoint $(\poly 1 / \varepsilon)$-size connected components.
Instantiating our framework in \textsf{CONGEST} yields a $\poly(\log{n}, 1/\varepsilon)$ round algorithm for computing $(1+\varepsilon$)-approximate maximum matching. In terms of the dependence on $1/\varepsilon$, this result improves exponentially state-of-the-art result by 
Lotker, Patt-Shamir, and Pettie~[LPSP15]. Our framework leads to the same quality of improvement in the context of the Massively Parallel Computation model as well.

\end{abstract}

\newpage

\section{Introduction}
\emph{Maximum Matching} is one of the most fundamental problems in combinatorial optimization and has been extensively studied in the classic centralized model of computation for almost half a century. We refer to \cite{schrijver2003combinatorial} for an overview. In particular, several exact polynomial-time deterministic maximum matching algorithms are known \cite{edmonds1965maximum,hopcroft1973n,micali1980v,gabow1990data}. 
Due to the quickly growing data sets naturally arising in many real-world applications (see \cite{drake2003improved} for an overview), there has been an increasing interest in algorithm design for huge inputs. 
For massive graphs the classical matching algorithms are not only prohibitively slow, but also space complexity becomes a concern. If a graph is too large to fit into the memory of a single machine, all the classical algorithms---which assume random access to the input---are not applicable. 
This demand for a more realistic model for processing modern data sets has led to the proposal of several different computing models that address this shortcoming. 
One that has attracted a lot of attention, especially in the past decade, is the \emph{graph stream model}, which was introduced by Feigenbaum et al.~\cite{feigenbaum2005graph,feigenbaum2005graphDistances,muthukrishnan2005data} in 2005.   
In this model, the edges of the graph are not stored in the memory but appear in an arbitrary (that is, adversarially determined) sequential order, a so-called \emph{stream}, in which they must be processed. The goal is to design algorithms that require little space and ideally only a small constant number of \emph{passes} over the stream.
In particular, it is desirable that the number of passes is independent of the input graph size.
We call an algorithm a $k$-pass algorithm if the algorithm makes $k$ passes over the edge stream, possibly each time in a different order \cite{munro1980selection,feigenbaum2005graphDistances}. 
This model is not only interesting for massive data sets but also whenever there is no random access to the input, for instance, if the input is only defined implicitly. 
Moreover, many insights and techniques from this model naturally carry over to a variety of areas in theoretical computer science, including communication complexity and approximation algorithms.

In the \emph{semi-streaming model}, which is the most commonly established variant of the graph stream model, the algorithm is given $\widetilde{O}(n)$\footnote{The $\widetilde{O}$ hides poly-logarithmic terms, thus $\widetilde{O}(n)=n \cdot \poly( \log n)$.} space for input graphs with $n$ nodes. This has turned out to be the sweet spot since even basic graph problems such as connectivity become intractable with less space \cite{feigenbaum2005graphDistances}. Moreover, note that often even just storing a solution requires $\Omega(n \log n)$ memory.

It is known that finding an exact matching requires linear space in the size of the graph and hence it is not possible to find an exact maximum matching in the semi-streaming model~\cite{feigenbaum2005graph}, at least for sufficiently dense graphs. Nevertheless, this result does not apply to computing a good approximation to the maximum matching in this model. We call an algorithm an \emph{$\alpha$-approximation} if the matching has a size at least $1/\alpha$ times the optimum matching. 

More than 15 years ago, in his pioneering work, McGregor \cite{mcgregor2005finding} initiated the study of arbitrarily good matching approximation algorithms --- that is, $(1+\eps)$-approximate for an arbitrarily small $\eps>0$ --- in the semi-streaming model. He presented a randomized algorithm that needs $\left(1/\eps\right)^{O(1/\eps)}$ passes. The same asymptotic number of passes was recently achieved by a deterministic algorithm by Tirodkar~\cite{tirodkar2018deterministic}. Improving the dependance on $1/\eps$ is left as an open problem in \cite{mcgregor2005finding}:
\\\\
		\begin{minipage}{1\linewidth}
			\begin{mdframed}[backgroundcolor=white, linecolor=red!40!black]
	\emph{
	 ... a weaker dependence [on $1/\eps$] would be
desirable.}
			\end{mdframed}
		\end{minipage}
\\

Summarized, all previously known $(1+\eps)$-approximation algorithms, whether deterministic or randomized, need exponentially in $1/\eps$ many passes. Moreover, while for the special case of bipartite graphs $\poly (1/\eps)$-pass algorithms are known, see \cite{ahn2011linear,Eggert2012}, it is not clear at all how to extend their applicability to general graphs without exponentially increasing the number of passes.

However, to be considered an efficient approximation algorithm in theory, ideally the dependence on all relevant parameters should be polynomial. Indeed, this has been a key property in the qualification of efficiency in parametrized complexity. The question whether there is a $(1+\eps)$-approximate matching algorithm for general graphs with $\poly(1/\eps)$ passes, even if randomness is allowed, thus is arguably one of the most central ones in the area. \cite{Eggert2012} leave this as the most intriguing open problem:
\\\\
		\begin{minipage}{1\linewidth}
			\begin{mdframed}[backgroundcolor=white, linecolor=red!40!black]
	\emph{
	The most intriguing question is ... to find $(1 + \eps)$
approximations of a maximum matching in general graphs, while keeping the bound
on the number of passes independent of $n$ and maintaining a substantially lower number
of passes than McGregor's algorithm needs.}
			\end{mdframed}
		\end{minipage}

\subsection{Our Results}
\begin{table}[h]
\centering
\begin{tabular}{|l||l|l|}
\hline
\backslashbox{Model}{Running Time}  & Our Work & Prior Work \\ \hline
Semi-Streaming & $\poly 1/\eps$ & \makecell{$\exp 1/\eps$ \\ \cite{tirodkar2018deterministic,mcgregor2005finding}} \\ \hline
Linear-Memory MPC & $O(\log \log n \cdot \poly 1/\eps)$ & \makecell{$O(\log \log n \cdot \exp 1/\eps)$,\\ \cite{czumaj2019round,ghaffari2018improved,assadi2019coresets,behnezhad2019exponentially}} \\  \hline
Sublinear-Memory MPC & $\tO(\sqrt{\log n} \cdot \poly 1/\eps)$ & \makecell{$\tO(\sqrt{\log n} \cdot \exp 1/\eps)$ \\ \cite{ghaffari2019sparsifying,onak2018round}} \\ \hline
\textsf{CONGEST} & $O(\log n \cdot \poly 1/\eps)$ & \makecell{$O(\log n \cdot \exp 1/\eps)$ \\ \cite{lotker2015improved}} \\ \hline
\end{tabular}

\caption{\label{table:runningtimes} A summary of the running times in several different models, compared to the previous state-of-the-art, for computing a $(1 + \eps)$-approximate maximum matching. In the distributed setting, ``running time'' refers to the round complexity, while in the streaming setting it refers to the number of passes.}
\end{table}
We break this exponential barrier on the number of passes, and thus solve this long-standing open problem, by devising an algorithm whose pass complexity has polynomial dependence on $1/\eps$. 
In the theorem statement, we assume that each edge can be stored with one \emph{word} of memory.

\begin{restatable}{theorem}{main}\label{thm: main}
	Given a graph on $n$ vertices, there is a deterministic $(1+\eps)$-approximation algorithm for maximum matching that runs in $\poly(1/\eps)$ passes in the semi-streaming model. 
	Furthermore, the algorithm requires $n \cdot \poly(1/\eps)$ words of memory.
\end{restatable}

This is the first (deterministic or randomized) algorithm with polynomially in $1/\eps$ many passes. 
It not only improves exponentially on the randomized $\left(1/\eps\right)^{O(1/\eps)}$-pass algorithm in the seminal work of McGregor \cite{mcgregor2005finding}, but also removes the need for randomness. 
On the deterministic side, an algorithm with the same pass complexity as McGregor's, was only recently found \cite{tirodkar2018deterministic}.
Our result is also an exponential improvement on the pass complexity of this deterministic result.
Moreover, a result by Ahn and Guha showed that by allowing a slightly superlinear memory of $n^{1 + 1/p}$, one can achieve a $(1 + \eps)$-approximation with $p/\eps$ passes~\cite{Ahn2018}.
This also implies a linear memory result with a slightly sublogarithmic number of passes by plugging $p \approx \log n / \log \log n$.

In the special case of bipartite graphs, the deterministic algorithms by Ahn and Guha \cite{ahn2011linear}, Eggert et al.~\cite{Eggert2012}, as well as Assadi et al.~\cite{assadi2022semi} obtain a runtime of $\poly(1/\eps)$ passes.
The first algorithm can also be adapted to the case of general graphs, but the required number of passes loses the independence on the graph size and becomes $\log n \cdot \poly(1/\eps)$.
It is not known whether the latter two approaches can be modified to work in case of general graphs.
Our results apply to general graphs and achieve a pass complexity of $\poly(1/\eps)$ without any dependency on $n$.

In very recent developments, lower bounds for the problem have been studied. 
One line of work focused on showing that for a certain class of deterministic algorithms $\Omega(n)$ passes are needed to improve upon the approximation ratio $1/2$~\cite{khalilK20}; this class of algorithms in each pass runs the simple greedy algorithm on a vertex-induced subgraph. 
More generally, Kapralov showed that any single-pass streaming algorithm for matching with an approximation ratio better than $0.59$  requires $n^{1+\Omega(1/\log \log n)}$ space~\cite{kapralov2021space}. 
This work subsumed earlier works by Kapralov and Goel et al.~\cite{goel2012communication, kapralov2013better}.
In another work~\cite{assadi2021GraphStreamingLowerBounds}, space-pass tradeoffs in graph streaming are analyzed, proving that an $(1+\eps)$-approximation needs either $n^{\Omega(1)}$ space or $\Omega(1/\eps)$ passes, even for restricted graph families. In \cite{assadi2022two}, Assadi provides a bound on the approximation ratio for two-pass semi-streaming algorithms.  

Furthermore, we show how to extend our main algorithm to a general framework, proving the following statement.
\newcommand{\Tmatching}{T_{\rm{matching}}}
\newcommand{\Amatching}{A_{\rm{matching}}}
\newcommand{\Texplore}{T_{\rm{explore}}}
\newcommand{\Aexplore}{A_{\rm{explore}}}
\begin{theorem}[Restatement of \cref{theorem:framework}]
\label{theorem:framework-restated}
Let $G$ be a graph on $n$ vertices and let $\eps \in (0, 1/2)$ be a parameter. Let $\Amatching$ be an algorithm that finds an $O(1)$-approximate maximum matching in time $\Tmatching$. Let $\Aexplore$ be an algorithm that in time $\Texplore$ processes any number of disjoint components of $G$ each of size $\poly 1/\eps$. Then, there is an algorithm that computes a $(1 + \eps)$-approximate maximum matching in $G$ in time $O((\Tmatching + \Texplore) \cdot \poly 1/\eps)$.
Furthermore, the algorithm requires access to $\poly(1/\eps)$ words of memory per each vertex.
\end{theorem}

Instantiating our framework with state-of-the-art results for computing an $O(1)$-approximate maximum matching in \textsf{CONGEST} and MPC, we obtain the results outlined in \cref{table:runningtimes}. In particular, our framework exponentially improves the dependence on $1/\eps$ in these models, hence resolving an open problem posed in \cite{lotker2015improved}:
\\\\
		\begin{minipage}{1\linewidth}
			\begin{mdframed}[backgroundcolor=white, linecolor=red!40!black]
			\emph{
	For unweighted graphs, it is interesting to see whether there exists a $(1+\eps)$-approximation
for general graphs, using small messages, with time complexity polynomial in $1/\eps$
and $\log n$.}
			\end{mdframed}
		\end{minipage}

\subsection{Other Related Work}

In this section, we give a brief overview of related work with slightly different problem formulations or underlying models. For a more thorough overview, we refer to \cite{mcgregor2014graph}. 
One line of research is centered around finding a constant-approximate, as opposed to arbitrarily good, maximum matching in as few passes as possible in the semi-streaming model \cite{feigenbaum2005graph, mcgregor2005finding, ahn2011linear, epstein2011improved, Eggert2012, konrad2012maximum, zelke2012weighted,esfandiari2016finding,kale2017maximum,paz20172+,GhaffariW19,feldman2021maximum}. 
Note that while a trivial greedy maximal matching algorithm yields a $2$-approximation in a single pass, it is a major open question whether there is a one-pass algorithm with approximation ratio smaller than $2$.

Another sequence of results is dedicated to devising streaming algorithms that estimate the size of the maximum matching \cite{kapralov2014approximating,bury2015sublinear,assadi2017estimating,esfandiari2018streaming,kapralov2020space}.
The problem of finding an arbitrarily good approximation has been studied in the streaming model~\cite{bertsekas1988auction,assadi2021auction,konrad2021two} on bipartite graphs as well as various related models that deal with non-random access to the input.
For instance, there are works in the setting of dynamic streams where edges can be added and removed \cite{konrad2015maximum,assadi2016maximum,chitnis2016kernelization}, in the random streaming model where edges or vertices arrive in a random order \cite{Assadi2021, Bernstein2020, gamlath2019weighted, konrad2012maximum, mahdian2011online}, and in models with vertex (instead of edge) arrival \cite{karp1990optimal,epstein2013,chiplunkar2015randomized,buchbinder2019online,GamlathKMSW19}. 
This problem has also been studied in the dynamic \cite{bernstein2016faster,solomon2016fully, Bhattacharya2019, Behnezhad2019, Behnezhad2022} (see~\cite{Hanauer2021} for a survey) as well as the classical centralized model \cite{Kalantari1995,preis1999linear,drake2003improved,duan2014linear}.

\subsection{Roadmap}In \cref{sec:overview,subsec:algorithm}, we introduce the terminology and present our algorithm. 
Furthermore, we make some important observations about invariants that are preserved by operations of our algorithm which we will use later. 
In \cref{sec: correctness}, we prove the correctness of our algorithm. The approximation analysis as well as the proof of the pass complexity can be found in \cref{sec:pass}. In \cref{section:general-framework} we provide details about our general framework for finding approximate maximum matching.

\subsection{Informal Outline and Challenges}

In this section, we give a brief outline of our approach and discuss the challenges we overcome.
As the basic building block, we follow the classic approach by Hopcroft and Karp~\cite{hopcroft1973n} of iteratively finding short augmenting paths to improve a $2$-approximate matching that can easily be found by a greedy algorithm.
To find these augmenting paths, we perform a depth first search (DFS) style truncated search from each free vertex in parallel and, once a sufficient amount of disjoint augmenting paths have been found, we augment the current matching over these augmenting paths.
The search scans over alternating paths of length roughly $1/\eps$ starting from the free (i.e., unmatched) vertices.
Our truncated DFS search method is inspired by~\cite{Eggert2012}, who designed a similar search for bipartite graphs.

As the context to one of our challenges, for $k \leq 1/\eps$, consider an alternating path $P = (b_1, a_1, b_2, a_2, \ldots,$ $a_{k}, b_{k + 1})$ between two free vertices $\alpha$ and $\beta$, where $a$ denotes matched and $b$ denotes unmatched edges. The free vertex $\alpha$ is an endpoint of $b_1$ and the free vertex $\beta$ is an endpoint of $b_{k+1}$.
In the case of bipartite graphs, the following key properties holds.
First, for a bipartite graph $G = (U \cup V, E)$ if $\alpha \in U$, then $\beta \in V$.
Second, if there is an alternating path $P'$ of length $h$ to the edge $a_i$ from some free vertex $\gamma \in U$, then edges $b_{i + 1}$ and $a_{i + 1}$ can be appended to $P'$ to create an alternating path of length $h + 2$ to $a_{i + 1}$ from $\gamma$.
This is true unless $P'$ contains $a_{i + 1}$, but in this case, $P'$ already is an even better (shorter) alternating path to $a_{i + 1}$.

The second property can be leveraged as follows:
Suppose that the DFS search is only performed by the nodes in $U$ and that the DFS search by $\alpha$ is scanning over path $P$.
If no DFS search other than by $\alpha$ is performed over $P$, $\alpha$ will eventually find $\beta$ and we have found a short augmenting path as desired.
However, it can be the case that the DFS search by another free vertex $\gamma$ has already scanned over an edge $a_i$.
If the alternating path $P_\gamma$ starting from $\gamma$ was of length $i' > i$, then it could be that $\gamma$ did not find $\beta$ since we truncate the DFS at length $1/\eps$.
In this case, $\alpha$ continues its search over $a_i$ since there is still hope to find free vertices that were not found by $\gamma$.

If it was the case that $P_\gamma$ was of length $i' \leq i$, then the second property mentioned above guarantees that $\gamma$ can find $\beta$ along the suffix path $a_i, \ldots, a_{k - 1}, b_{k}$.
Importantly, this allows $\alpha$ to ignore $a_i$ with its current search and continue to the next branch of the DFS.
This second property is leveraged in the algorithm by~\cite{Eggert2012}.

\paragraph{The first challenge.}
In the case of general graphs, the second property mentioned above that allows us to complete alternating paths does not hold.
In particular, consider an edge $a_j = (u, v), j > i$ that is ``after'' an edge $a_i$ in the path $P = (b_1, a_1, \ldots)$.
It can be the case that $a_j \in P'$ and the DFS search scans path $P'$ along the edge $a_j$ in the opposite direction as the DFS along path $P$.
See the red path in \Cref{fig:incompletable} for an illustration.

\begin{figure}
	\centering
	\includegraphics[width=0.85\linewidth]{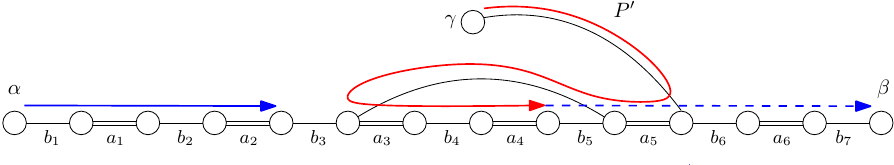}
	\caption{
		Nodes $\alpha$, $\beta$, and $\gamma$ are free. The black single-segments are unmatched and black (full) double-segments are matched edges. The path $P'$  corresponding to a DFS branch of $\gamma$ is shown by the red solid spline. Since the edge $a_5$ is part of the path, the current DFS branch of $\gamma$ cannot be extended up to the free node $\beta$ along the dashed blue line. Furthermore, the path from $\gamma$ to the edge $a_3$ can potentially ``block" a longer DFS search path of $\alpha$ illustrated with a solid blue line. However, the edges along the DFS searches of $\alpha$ and $\gamma$ can be combined to find an augmenting path between $\alpha$ and $\gamma$.
	}
	\label{fig:incompletable}
\end{figure}

Here, we make the observation that by combining the prefixes of $P$ and $P'$ until  the edge $a_j$, we obtain an augmenting path.
On a high level, our approach is to show that $\alpha$ can stop exploring the path $P$ further and our algorithm will, eventually, either find an augmentation between $\alpha$ and $\gamma$, or it will find some other ``good'' augmentation that \emph{intersects} $P$.

\paragraph{The second challenge.}
In the algorithm of~\cite{Eggert2012}, the algorithm only needs to remember the path that is currently active in the DFS search.
For the rest of the graph, \cite{Eggert2012} show that it is enough to store the length of the shortest alternating path that has reached each matched edge. This length is called \emph{label}.
In the first challenge, we considered the possibility that a vertex $\gamma$ ``blocks'' the DFS exploration of $\alpha$ and discussed how this implies an augmenting path between $\gamma$ and $\alpha$.
Unfortunately, it can also be the case that $\alpha$ blocks itself.
Now, if we simply had these distance labels on the edges, $\alpha$ might stop its DFS search due to distance labels it has set itself.
See the red path in \Cref{fig:jumping-example} for an illustration.
This, in turn, brings us to trouble since we cannot use the observation from the first challenge (in which $\alpha$ and $\gamma$ could augment), as there might not be any other free vertex to find an augmentation to.

To overcome this, the free vertices keep their DFS search tree history in memory.
We show that whenever $\alpha$ blocks itself, it must be the case that the current head of the DFS search has discovered an odd cycle.
In a sense and informally, we then implement a variant of the blossom contraction idea by Edmonds~\cite{edmonds_blossom}, and show that if another DFS search by another free vertex ever finds this cycle, then it also found an augmentation.
While we do not explicitly contract the blossom, we continue the search over the blocked edge as if it was not yet searched by a shorter path.

Since the tree history of a free vertex $\alpha$ is tied to $\alpha$, when our algorithm finds an augmentation containing $\alpha$ it removes all the vertices currently kept in the $\alpha$'s search tree, even those not being part of the augmentation. We do so as it is not clear how to deal with the state of vertices in the $\alpha$'s search tree after $\alpha$ is removed. This however poses additional challenges, with
the main difficulty lying in finding a consistent implementation of the DFS search and keeping the sizes of the search trees small. The latter is particularly important as our algorithm removes the entire search tree of a free vertex that gets augmented.

Below, we elaborate on two of the (perhaps) most crucial steps to overcome this challenge.
First, we show that for each found augmentation between free vertices $\alpha$ and $\beta$ the algorithm can afford to \emph{entirely remove} search trees of $\alpha$ and $\beta$ from the graph for the rest of the current DFS search. This removal affects the approximation guarantee only by $\poly 1/\eps$ factor, and by repeating the DFS search $\poly 1/\eps$ times from scratch leads to the desired $(1+\eps)$-approximation.
In contrast to the prior works, our algorithm is allowed to augment over paths much longer than $1/\eps$ (nevertheless, lengths are still polynomial in $1/\eps$).

Second, consider one iteration of finding short augmenting paths in the Hopcroft-Karp framework.
Our algorithms ``puts on hold'' (or pauses) DFS over search trees that become too large. Note that pausing DFS execution of some search trees increases the time required to explore the entire graph. Nevertheless, we show how to set parameters so that putting on hold DFS over large trees increases the number of passes only by a $\poly 1/\eps$ factor.

\paragraph{A Comparison with the Algorithm of Tirodkar~\cite{tirodkar2018deterministic}.}
Tirodkar's algorithm finds a $(1 + \eps)$-approximation to maximum matching in $(1/\eps)^{O(1/\eps)}$ passes.
Similar to our algorithm, it starts by finding a maximal matching and iteratively finding augmenting paths and thereby improving on the approximation guarantee.
The basic building block in the search for augmenting paths is to find \emph{semi-matchings} between the vertices and their matched neighbors such that each vertex has a small amount of neighbors in the semi-matching.
In the case of bipartite graphs, they show that their method of searching for augmenting paths in a graph defined by the semi-matchings finds a significant, but only an exponentially small in $\eps$, fraction of all possible augmenting paths.
Hence, after an exponentially many applications of augmentations, they obtain a $(1 + \eps)$-approximation.
To account for odd cycles in the input graph, they study \emph{directed} semi-matchings that, in a sense, corresponds to our idea of a directed DFS search.
In both works, a crucial ingredient is to store the right edges of the input graph to make sure that the augmenting paths yielded by the search do not contain the same node twice.
One of our main contributions is to show that a \emph{structure} (defined in \Cref{sec:overview} and maintained by each free vertex) contains sufficient information to perform augmentations, while at the same time fitting in only $\poly 1/\eps$ memory.
Our DFS search approach guarantees that we find a $\poly \eps$ fraction of all possible augmentations, giving rise to an algorithm that in $\poly 1/\eps$ passes finds a $(1+\eps)$-approximate maximum matching.

\section{Definitions, Notations, and Preliminaries}\label{sec:overview}
In the following, we will first introduce all the terminology, definitions, and notations.
Throughout the paper, if not stated otherwise, all the notation is implicitly referring to a currently given matching $M$ which we aim to improve.

\begin{definition}[An Unmatched Edge and a Free Node]
We say that an edge $\{u,v\}$ is \emph{matched} iff  $\{u,v\} \in M$, and \emph{unmatched} otherwise. We call a node $v$ \emph{free} if it has no incident matched edge, i.e., if $\{u,v\}$ are unmatched for all edges $\{u,v\}$. Unless stated otherwise, $\alpha, \beta, \gamma$ are used to denote free nodes. 
\end{definition}

\begin{definition}[Alternating Path, Alternating Length, Alternating Distance]
\label{definition:alternating-paths}
An \emph{alternating path} $P = (u_1, \ldots, u_k)$ is a path that consists of a sequence of alternatingly matched and unmatched arcs that is disjoint. 
A disjoint sequence means that it must not contain a vertex twice.
The \emph{alternating length} $|P|$ of an alternating path $P$ is defined as the number of \emph{matched} arcs involved in $P$. 
\end{definition}

\paragraph{Notation of alternating paths.}
Let $P = (u_1, v_1, \ldots, u_k, v_k)$ be an alternating path such that $(u_i, v_i)$ are matched arcs and $(v_i, u_{i + 1})$ are unmatched ones. Let $a_i = (u_i, v_i)$. We often use $(a_1, \ldots, a_k)$ to refer to $P$, i.e., we \emph{omit} specifying unmatched arcs. Nevertheless, it is guaranteed that the input graph contains the unmatched arcs $(v_i, u_{i + 1})$, for each $1 \le i < k$.
If $P$ is an alternating path that starts and/or ends with unmatched arcs, e.g., $P = (x, u_1, v_1, \ldots, u_k, v_k, y)$ where $(x, u_1)$ and $(v_k, y)$ are unmatched while $a_i = (u_i, v_i)$ for $i = 1 \ldots k$ are matched arcs, we use $(x, a_1, \ldots, a_k, y)$ to refer to $P$. In our case, very frequently $x$ and $y$ will be free nodes, usually $x = \alpha$ and $y = \beta$.

\begin{definition}[Concatenation of Alternating Paths]
\label{definition:concatenation}
	Let $P_1 = (a_1, a_2, \ldots, a_k)$ and $P_2 = (b_1, b_2, \ldots, b_h)$ be alternating paths.
	We use $P_1 \circ P_2 = (a_1, a_2, \ldots, a_k, b_1, b_2, \ldots, b_h)$ to denote their concatenation.
	Note that the alternating path $P_1 \circ P_2$ also contains the unmatched edge between $a_k$ and $b_1$.
\end{definition}

\begin{definition}[Path Reverse]
Let $P = (a_1, a_2, \ldots, a_k)$ be an alternating path.
We denote the \emph{reverse} of path $P$ by $\overleftarrow{P} = (\overleftarrow{a_k}, \overleftarrow{a_{k - 1}}, \ldots, \overleftarrow{a_1})$.
\end{definition}

\begin{definition}[The Label of an Arc]
	Each \emph{matched} arc $a$ is assigned a \emph{label} $\ell(a)$, with $\ell(a)$ being a finite positive integer or the value $\infty$ which is larger than any finite integer.
\end{definition}
In \cref{fig:using-arc-labels} we justify why we assign labels per each arc and hence two labels per edge, as opposed to a single label per edge. Intuitively, this enables us to handle odd cycles. 
It is instructive to think of a label $\ell(a)$ as the alternating length of an alternating path between some free node and the arc $a$. Saying it differently, $a$ is the $\ell(a)$-th matched arc along an alternating path originating at some free node $\alpha$. The value of $\ell(a)$ changes throughout the algorithm and at different times of the algorithm can refer to distances to different free nodes. Our algorithm is designed so that the label of an arc never increases.

\begin{figure}
	\centering
	\includegraphics[width=0.35\linewidth]{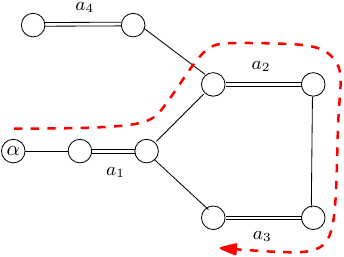}
	\caption{
		Assume that we first perform a DFS style search over the red (dashed) path from $\alpha$ to $a_3$. 
		Moreover, assume that the algorithm maintains the (shortest) path length of the DFS as labels on edges, but not of arcs. In that case, this red path sets labels $\ell(a_1) = 1$, $\ell(a_2) = 2$ and $\ell(a_3) = 3$. 
		Eventually, the search backtracks to $a_1$ and continues over $a_3$ setting the label to $\ell(a_3) = 2$.
		However, this DFS branch cannot continue to $a_4$, even though the path exists, due to the distance label $\ell(a_2) < \ell(a_3) + 1$.
		By considering labels on arcs, we allow $\alpha$ to extend its DFS search over $(\alpha, a_1, a_3, a_2, a_4)$.
	}
	\label{fig:using-arc-labels}
\end{figure}

\begin{definition}[Structure of a Free Node]\label{definition:structure}
	A structure $\structure_{\alpha}$ of a free (i.e., unmatched) node $\alpha$ is a set of arcs with the following properties:
	\begin{enumerate}
		\item\label{prop: disjoint} \textbf{Disjointness}: $\structure_{\alpha}$ is vertex-disjoint from all other structures.
		\item\label{prop:alternating-path} \textbf{Alternating-paths to matched arcs}: The matched arcs of $\structure_{\alpha}$ are \emph{connected} via in-structure alternating paths from $\alpha$, i.e., there is an alternating path within $\structure_\alpha$ from $\alpha$ to each matched arc of $\structure_{\alpha}$.\footnote{$\structure_{\alpha}$ is allowed to contain \emph{unmatched} arcs to which there is no alternating path from $\alpha$.}
		\item\label{prop:unmatched-arcs} \textbf{Endpoints of unmatched arcs:} Each endpoint of an unmatched arc of $\structure_{\alpha}$ is $\alpha$ or an endpoint of a matched arc of $\structure_{\alpha}$.
		\item \textbf{Active path}: The structure is associated with an \emph{active path} $P_\alpha$, where one of the following holds.
		(1) $P_\alpha = (\alpha, a_1, a_2, \ldots, a_k) \subseteq \structure_\alpha$ is alternating;
		or (2) $P_\alpha = \emptyset$.
		\\
		Note that in Case (1) either $P_\alpha = (\alpha)$, or $P_\alpha$ ends with a matched arc.
	\end{enumerate}
\end{definition}

\begin{invariant}[Upper bound on lengths of active paths]
\label{invariant:active-path-length}
	For each active path $P_\alpha = (\alpha, a_1, \ldots, a_k)$, it holds $\ell(a_i) \leq i$, for all $1 \leq i \leq k$ and $k \leq \maxlen$. The value $\maxlen$ is a parameter that we will fix later.
\end{invariant}

\begin{definition}[Active and Inactive Arcs and Vertices]
We say that an arc or a vertex is active if it belongs to an active path. See \cref{definition:structure} for a definition of active paths. If a vertex or an arc is not active, we say that it is inactive.
\end{definition}

\begin{definition}[Removing-directions operator]
\label{definition:RD-operator}
	Given a set of arcs $A$, we use $\RD(A)$ to refer to the set of \emph{edges} obtained by ignoring directions of the arcs in $A$. Formally, $
	\RD(A) = \{\{u, v\}\ |\ (u, v) \in A {\rm\ or\ } (v, u) \in A\}$.
\end{definition}

\begin{definition}[Reachable Arcs]\label{def: reachable}
	We say that an arc $a \in \RD(\structure_\alpha)$ is \emph{reachable}, if there is alternating path in $\RD(\structure_\alpha)$ from the free vertex $\alpha$ to $a$.
	Notice that it can be the case that $a \in \structure_\alpha$ and $\overleftarrow{a} \not \in \structure_\alpha$ but $\overleftarrow{a}$ is reachable.
\end{definition}

\begin{remark*}
In our algorithm descriptions, we mainly consider directed edges, i.e., arcs. The structures are used to keep track of the directed alternating paths found by the algorithm.
Intuitively, upon discovering an odd cycle along a certain direction, we \emph{implicitly} learn of the existence of alternating paths along both directions around the cycle since the underlying graph is undirected.
These implicitly learned alternating paths will come in handy for finding augmenting paths between free vertices. We also note that our definition of reachable arcs (\Cref{def: reachable}) captures these implicitly learned alternating paths to certain arcs.
\end{remark*}

\section{Our Algorithm}\label{subsec:algorithm}

In the following, we will sketch our algorithm, incrementally providing more details. 

\paragraph{Initialization and Outline.}

In the first pass, we apply a simple greedy algorithm to find a maximal matching, hence a $2$-approximation. This $2$-approximate maximum matching is our starting matching. The rest of our algorithm is divided into multiples \emph{phases}. In each phase, we iteratively improve the approximation ratio of our current matching $M$ by finding a set of disjoint $M$-augmenting paths (and performing the augmentations accordingly). We stop the algorithm after certain number of phases to be fixed later (see \Cref{alg:main}). 
In \cref{sec:pass} we show that executing our algorithm for that number of phases yields a $(1 + \eps)$-approximate maximum matching.
  
\begin{algorithm}[h]
	\SetKwProg{function}{function}{}{\KwRet}
	\SetKw{continue}{continue}

  \Input{$G$: a graph \\
		$\eps$: approximation parameter}
	
		Compute, in a single pass, a $2$-approximate maximum matching $M$. \label{line:main-alg-maximal-matching}

		\For{$\poly 1/\eps$ phases \tcc*{Nothing stored from the previous phase.}} {
			$\mathcal{P} \coloneqq \AlgPhase(G, \eps, M)$

			Augment the current matching $M$ using the disjoint augmenting paths in $\mathcal{P}$.

			\tcc*{Augmenting requires no extra passes.}
		}

  \caption{The highest level algorithm description.}
	\label{alg:main}
\end{algorithm}

Next, we outline what the algorithm does in a single phase; this method is called \AlgPhase (see \Cref{alg:phase}). The set $M$ in our description refers to the current matching in this phase. 

\paragraph{A single phase (\AlgPhase).}
The description of executing a single phase, \cref{alg:phase}, is given next.
A phase consists of $\maxtau \eqdef 1/\eps^6$ many $\PassBundle$s executed by the loop starting at \cref{line:phase-loop-pass-bundle}.

\begin{algorithm}[h]
	\SetKwProg{function}{function}{}{\KwRet}
	\SetKw{continue}{continue}

  \Input{$G$: a graph \\
		$\eps$: approximation parameter\\
		$M$: the current matching}
  \Output{A set $\mathcal{P}$ of \textit{disjoint} augmenting paths that augment $M$}

	$\mathcal{P} \gets \emptyset$
	
	$\ell(a) = \infty$, for each arc $a \in M$ \label{line:reset-label}
	
	For each free node $\alpha$, create a structure $\structure_\alpha$ with its active path $(\alpha)$. \label{line: active}
		
		\For{$\PassBundle$ $\tau=1 \ldots \maxtau$ \label{line:phase-loop-pass-bundle}}{
			For each free node $\alpha$, if $\structure_\alpha$ has less than $\sizelimit$ vertices, mark $\structure_\alpha$ as \textbf{not} ``on hold''.\label{line:on-hold}

			$\AlgExtendStructures(\mathcal{P})$ \tcc*{see \cref{sec:extend-structures}}
			
			For each free node $\alpha$, if $\structure_\alpha$ has at least $\sizelimit$ vertices, mark $\structure_\alpha$ as \textbf{``on hold''}.\label{line:not-on-hold}
			
			\AlgBacktrack \tcc*{see \cref{sec:backtracking}}
			
			$\AlgEdgeMerge(\mathcal{P})$ \tcc*{see \cref{sec:edge-merging}}
			
			\AlgAugmentStructures \tcc*{see \cref{sec:augment-structures}} \label{line:AlgAugmentStructures}
		}
		
		\Return 
	  
  \caption{\AlgPhase: the execution of a single phase. The subroutines that might add augmenting paths to the set $\mathcal{P}$ are highlighted by showing the input parameter explicitly. The other parameters are left out for clarity.}
 \label{alg:phase}
\end{algorithm}

\paragraph{A single bundle of passes (\PassBundle).}
Our algorithm executes several methods (invoked within the loop starting at \cref{line:phase-loop-pass-bundle} of \cref{alg:phase}), and for most of them it makes a \emph{fresh} pass over the edges. The term $\PassBundle$ refers to multiple passes during which those routines are executed. Precisely, the routines are: (1) extend structures along active paths (\AlgExtendStructures), (2)  check for edge augmentations (\AlgEdgeMerge), and (3) include (additional) unmatched edges to each structure (\AlgAugmentStructures). Each of these routines is performed in a separate pass over the edges. The \AlgBacktrack method backtracks active paths that were not extended, but does not require a fresh pass. In total, a $\PassBundle$ requires $3$ passes.

\PassBundle $p$ refers to \emph{the entire} execution of the $p$-th iteration of the loop of \cref{alg:phase}. In our proofs, we will use terminology such as ``at the beginning of \PassBundle $p$''. This refer to the state of our algorithm at the very beginning of the $p$-th iteration, before it sees any edge on the stream. Likewise, ``at the end of \PassBundle $p$'' refers to the state of our algorithm after \PassBundle $p$ has executed. In particular, note that ``at the beginning of \PassBundle $p$'' and ``at the end of \PassBundle $p - 1$'' essentially refer to the same state of our algorithm. 
When we say that an event happened \emph{during} \PassBundle $p$, without specifying that it was at the beginning or at the end, the event could have happened at any point during that bundle.



	\subsection{Marking a structure on hold}
	\label{section:on-hold}
Let $\sizelimit \eqdef 1 / \eps^{4}$. If a structure $\structure_\alpha$ contains at least $\sizelimit$ vertices, then we mark $\structure_\alpha$ \emph{on hold}; see \cref{line:on-hold} of \cref{alg:phase}. This operation plays a crucial role in our analysis of pass-complexity of our algorithm, e.g., \cref{lemma:activenodes,lemma: runtime}.


\subsection{Extending Active Paths (\AlgExtendStructures)}
\label{sec:extend-structures}
Before we describe \AlgExtendStructures, we will define several operations that the method uses.

\subsubsection{Description of $\AlgMerge(\mathcal{P}, g)$}
\label{operation:merging}
		\begin{minipage}{0.95\linewidth}
			\begin{mdframed}[backgroundcolor=gray!15, linecolor=red!40!black]
			\emph{Informal description}: This operation finds an augmenting path containing a given unmatched arc $g$ (without seeing any new edge on the stream) and removes the vertices contained in the structures affected by this augmentation.
			\end{mdframed}
		\end{minipage}
		\\\\
	Given an unmatched arc $g$, let $\structure_\alpha$ and $\structure_\beta$ be two structure $g$ is incident to; we assume that such two distinct structures exist. Let $P_A$ be an alternating path between $\alpha$ and $\beta$ in $\RD(\structure_\alpha \cup \structure_\beta \cup \{g\})$, that we also assume exists. (As a reminder, the $\RD$ operator is described in \cref{definition:RD-operator}.)
	Then, $\AlgMerge(\mathcal{P}, g)$ stores the augmenting path $P_A$ in $\mathcal{P}$ and removes from the graph \emph{all} vertices from $\structure_{\alpha}$ and from $\structure_{\beta}$. These vertices remain removed until the end of this phase and our algorithm adds them back when it starts the next phase.
	This guarantees that the paths in $\mathcal{P}$ remain disjoint.	
	Note that since structures are vertex disjoint, the augmenting path $P_A$ contains $g$.
%


	\subsubsection{Description of $\AlgOvertake(P_u, g, \astar)$}
	\label{section:AlgOvertake}
	
	\label{operation:overtaking}
		\begin{minipage}{0.95\linewidth}
			\begin{mdframed}[backgroundcolor=gray!15, linecolor=red!40!black]
			\emph{Informal description}: This operation reduces the label of a matched arc $\astar$ via ``a short'' alternating path $P_u$ originating at a free node $\alpha$. Also, $\astar$ is added to $\structure_\alpha$ and removed from another structure $\structure_\beta$ it might belong to. The removal operation, called \emph{overtaking}, might cause structure $\structure_\beta$, that contains $\astar$, to become disconnected. To ensure that $\structure_\beta$ remains connected, all the arcs of $\structure_\beta$ that become disconnected from $\beta$ via alternating paths are also overtaken, i.e., removed from $\structure_\beta$ and added to $\structure_\alpha$. If some of the arcs removed from $\structure_\beta$ are active, their labels get updated as well.
			\end{mdframed}
		\end{minipage}
		\\\\

\noindent	The method $\AlgOvertake(P_u, g, \astar)$ gets $P_u$, $g$ and $\astar$ as input, whose details are elaborated below:
		\begin{itemize}
			\item $P_u = (\alpha, a_1, \ldots, a_k)$ is an alternating path in a structure $\structure_{\alpha}$ such that $|P_u| < \maxlen$, and either $k = 0$ or $a_k$ is a matched arc. Moreover, $P_u$ as a prefix contains the current active path of $\structure_\alpha$.
			\item $\astar$ is a matched arc and $g = (x, y)$ is an unmatched arc such that: $|P_u| + 1 < \ell(\astar)$ and $P_u \circ (\astar)$ is alternating (see \cref{definition:concatenation} for the definition of operator $\circ$); and, $x$ is the head of $P_u$ and $y$ is the tail of $\astar$.
			\item If $\astar \in \structure_\beta$ for $\beta \neq \alpha$, then there is no alternating path between $\alpha$ and $\beta$ in $\RD(\structure_\alpha \cup \structure' \cup \{g\})$. In particular, we cannot invoke $\AlgMerge(\mathcal{P}, g)$ to find an augmenting path between $\alpha$ and $\beta$.
		\end{itemize}
		
		Next, we define $P_{\astar}$:
		\begin{itemize}
			\item If $\astar$ is not an active arc, let $P_\astar = (\astar)$.
			\item If $\astar$ is an active arc, in which case it belongs to $\structure_\beta$, let $P_{\astar}$ be the suffix of the active path of $\structure_\beta$ starting at $\astar$.
		\end{itemize}
		Then, \AlgOvertake performs a sequence of updates outlined next.
		\begin{enumerate}
			\item \textbf{Add $g$ and $\leftarc{g}$ to $\structure_\alpha$.}
			\item \textbf{Move alternating paths from $\structure_\beta$ to $\structure_\alpha$:}
			 If $\astar \in \structure_\beta$ for $\beta \neq \alpha$, let $\cA$ be the collection of arcs in $S_{\beta}$ such that for each $b \in \cA$, every alternating path in $S_\beta$ from $\beta$ to $b$ goes through at least one of the arcs in $P_{\astar}$; otherwise, let $\cA = \emptyset$. Add $P_{\astar}$ (together with unmatched arcs) to $\structure_\alpha$. Also, add to $\structure_\alpha$ each alternating path (together with unmatched arcs) of $\structure_\beta$ that originates at $P_{\astar}$ and ends in $\cA$.
			If $\astar \in \structure_\beta$, remove from $\structure_\beta$ all matched arcs in $P_\astar \cup \cA$.
			\item\label{item:clean-unmatched} \textbf{Clean unmatched arcs from $\structure_\beta$:} If $\astar \in \structure_\beta$ for $\beta \neq \alpha$, remove from $\structure_\beta$ each unmatched arc $(u, v)$ such that $u$ or $v$ is different than $\beta$ and is not endpoint of any matched arc of $\structure_\beta$. (After this, each endpoint of each unmatched arc of the updated $\structure_\beta$ is $\beta$ or is an endpoint of a matched arc of $\structure_\beta$, as desired by \cref{definition:structure}-\ref{prop:unmatched-arcs}.)
			\item	\textbf{Update active paths:} Set the new active path of $\structure_\alpha$ to $P_u \circ P_{\astar} = (\alpha, c_1, \ldots, c_t)$. Update the active path of $\structure_\beta$ by removing from it any arc in $P_\astar$ (note that if $\astar$ is inactive, then the active path of $\beta$ remains the same).
			\item\label{item:update-labels}	\textbf{Update labels:} Set $\ell(c_i) := \min\{i, \ell(c_i)\}$, for each $i = 1 \ldots t$. Observe that $\ell(c_i)$ after this update is \emph{at most} the number of matched arcs on the path from $\alpha$ to $c_i$ along $P_u \circ P_{\astar}$.
		\end{enumerate}
In \cref{sec:proof-of-lemma-overtaking} we show that overtaking is a sound operation, as made formal with the following claim.
\begin{restatable}[\AlgOvertake Preserves Structures]{lemma}{lemmaPreserve} \label{lemma: overtaking}
Let $P_u$, $g$ and $\astar$ be the input (as describe in \cref{section:AlgOvertake}) to an invocation of \AlgOvertake. 
After the invocation, the updated $\structure_\alpha$ is a structure per \cref{definition:structure}. If $\astar \in \structure_\beta$ for $\alpha \neq \beta$, after the same invocation $\structure_\beta$ is also a structure per \cref{definition:structure}. Moreover, if \Cref{invariant:active-path-length} holds before this invocation, it holds after the invocation as well.
\end{restatable}


\subsubsection{Description of \AlgExtendStructures}
\label{sec:description-of-extend-structures}

\begin{algorithm}[h]
	\SetKwProg{function}{function}{}{\KwRet}
	\SetKw{continue}{continue}

  \Input{$G$: a graph \\
		$\eps$: approximation parameter\\
		$M$: the current matching\\
		$\structure_\alpha$ for each free node $\alpha$\\
		$\mathcal{P}$: a set of disjoint augmenting paths
		}

		\tcc{A single pass over an arbitrary ordered stream of edges/arcs. When an edge $\{u, v\}$ arrives on the stream, we pass to our algorithm two arcs $(u, v)$ and $(v, u)$, one after the other.}
		\For{each arc $g = (u, v)$ on the stream} {
			\If{$u$ or $v$ was removed by \AlgMerge in this phase \label{line:ES-check-augmented}} {
				Continue with next arc.
			}
			\If{$g$ is matched, or $u$ does not belong to a structure, or $v$ is a free node belonging to the same structure as $u$\label{line:ES-basic-test}} {
				Continue with next arc.
			}
			\If{$u$ belongs to a structure that has already been extended via \cref{line:ES-AlgOvertake} in \textbf{the current} invocation of \AlgExtendStructures \label{line:ES-check-if-already-extended}} {
				Continue with next arc.
			}
			\If{$\RD(\structure_\alpha \cup \structure_\beta \cup g)$ contains an $\alpha$-$\beta$ alternating path for some $\alpha \neq \beta$\label{line:ES-check-alpha-beta-merge}} {
				$\AlgMerge(\mathcal{P}, g)$ \label{line:ES-alpha-beta-merge} \tcc*{Store augmentation, remove $\structure_\alpha$ and $\structure_\beta$.}
				
				Continue with next arc.
			}			
			{Let $P_u = P_\alpha \circ (b_1, \ldots, b_h)$ be an alternating path s.t.: {\tcc*{$h$ can be $0$} \label{line:extend-structure-P_u}}
			\begin{itemize}
				\item $\alpha$ is a free node
				\item $\structure_\alpha$ is \textbf{not} on hold and $P_\alpha$ is the active path of $\structure_\alpha$,
				\item $P_u$ is in $\structure_\alpha$,
				\item $b_i$ is a matched arc and $\ell(b_i) \leq |P_\alpha| + i$, for $i = 1 \ldots h$,
				\item $u$ is the ending vertex of $P_u$.
			\end{itemize}
			}
			
			
			\If{$P_u$ does not exist} {
				Continue with next arc. \tcc*{No active path can ``extend'' along $g$.}
			}
			
			
				%
			
			$\astar = (v, x)$ is a matched arc\label{line:ES-define-astar} \label{line: astar} \tcc*{$v$ is not free due to checks in \cref{line:ES-basic-test,line:ES-check-alpha-beta-merge}.}

			\If{$\ell(\astar) > |P_u| + 1$ and $|P_u| + 1 \le \maxlen$ \label{line:ES-check-astar-label}} {
				$\AlgOvertake(P_u, g, \astar)$ \label{line:ES-AlgOvertake} \label{line:ES-overtake}
					%
%
			}
		}
		
  \caption{The execution of \AlgExtendStructures (\cref{sec:description-of-extend-structures}).}
  \label{algo-extend-structures}
\end{algorithm}

In \Cref{algo-extend-structures}, we provide a high-level pseudo-code of \AlgExtendStructures, and below provide more details.
An invocation of \AlgExtendStructures takes a single pass over the edges.
The task of \AlgExtendStructures is to extend/grow the active path of each free node $\alpha$ by adding at least one additional \newtext{matched} arc $\astar$ to the active path, while respecting the labels. Since active paths are alternating, this means that for each added matched arc the method adds an unmatched arc to the active path as well. 
We fix the maximum number of matched arcs in an active path to be $\maxlen \eqdef \lceil 1/\eps \rceil$.
In \cref{algo-extend-structures} we provide a pseudo-code of this method, and in the rest of this section we describe that pseudo-code in more detail.

\AlgExtendStructures scans over unmatched edges; note that matched edges are in the algorithm's memory, so there is no new information to be gained by seeing a matched edge on the stream. Let $\{u, v\}$ be the current unmatched edge on the stream. Then, the algorithm considers separately $g = (u, v)$ and $\leftarc{g} = (v, u)$. We describe the algorithm's steps for $g$.

If $u$ does not belong to any structure, the algorithm does nothing with $g$ and continues to the next arc. Also, if $v$ equals free node $\alpha$ and $u \in \structure_\alpha$, then our algorithm also has not use of $g$ and continues to the next arc.

So, assume that $u$ belongs to $\structure_\alpha$ and $v \neq \alpha$. The goal of \AlgExtendStructures now is, if possible, to find an augmentation over $g$ or to use $g$ to reduce the label of the matched arc incident to $g$ while extending the active path of $\structure_\alpha$.
If there is an augmenting path between two free nodes over $g$, the algorithm invokes $\AlgMerge(g)$ (\cref{line:ES-alpha-beta-merge}). (As a reminder, our algorithm applies an augmentation whenever there is one in the set of arcs it has in its memory. Hence, upon receiving a new arc $g$, it first verifies whether it is possible to augment (\cref{line:ES-check-alpha-beta-merge}).) Note that since structures are vertex-disjoint (see \cref{definition:structure}), if it is possible to find an augmentation when $g$ is added to the memory of our algorithm, then it implies the augmentation contains $\RD(g)$.

Assume that it is not possible to augment over $g$. Let $P_\alpha = (\alpha, a_1, \ldots, a_k)$ be the active path of $\alpha$, where each $a_i$ is a matched arc. (Recall that in our notation we do not explicitly list unmatched arcs in alternating paths, see the paragraph below \cref{definition:alternating-paths}.)
Let $P_u = P_\alpha \circ (b_1, \ldots, b_h)$ be an alternating path that ends at vertex $u$ such that $\ell(b_i) \leq k + i$ for $i = 1 \ldots h$, where it is allowed that $h = 0$, i.e., $(b_1, \ldots, b_h)$ is allowed to be empty. \cref{line:extend-structure-P_u} lists certain properties that $P_u$ should satisfy. We note that the property that $\alpha$'s active path has not been extended so far (imposed by \cref{line:ES-check-if-already-extended}) is a mere technicality; we use that property only to simplify certain simulation of our general framework presented in \cref{section:general-AlgExtendStructures}.
If $P_u$ does not exist, the algorithm proceeds with the next arc as it is unable to use $g$ to reduce the label of a matched arc.

Note that $b_i$ arcs are allowed to be used \textbf{only} if their labels are small enough so that they cannot be reduced along $P_u$. If $\ell(b_i) > k + i$ for any $i$, then \AlgExtendStructures would have a chance to reduce $\ell(b_i)$ in the same pass, but across an arc other than $g$; for instance, along the unmatched arc adjacent to $b_{i - 1}$ and $b_i$. We say that \AlgExtendStructures ``\emph{jumps}'' over the arcs $b_i$.
In \cref{fig:jumping-example} we illustrate why allowing active paths to jump over arcs is crucial for our algorithm to make progress.


\begin{definition}[Jumping over arcs]
\label{definition:jumping}
If the alternating path $(b_1, \ldots, b_h)$ defined at \cref{line:extend-structure-P_u} of \AlgExtendStructures is non-empty, we say that our algorithm \emph{jumps} over the arcs $(b_1, \ldots, b_h)$.
\end{definition}

\noindent Recall that $g = (u, v)$ and $u$ is the ending vertex of $P_u$.
Finally, we discuss the matched arc $\astar$ that can be reached via $P_u$, i.e., $\astar \not \in P_u$ and $v$ is the tail of $\astar$ as described in \Cref{line: astar}.
Note that $v$ cannot be a free vertex, as otherwise \cref{line:ES-basic-test} or \cref{line:ES-check-alpha-beta-merge} would evaluate to true. 
The algorithm now checks whether it is possible to reduce the label of $\astar$ (\cref{line:ES-check-astar-label}) by an alternating path whose length does not exceed $\maxlen$. If it is possible, the label is reduced by invoking $\AlgOvertake(P_u, g, \astar)$.

\begin{figure}
	\centering
	\includegraphics[width=0.85\linewidth]{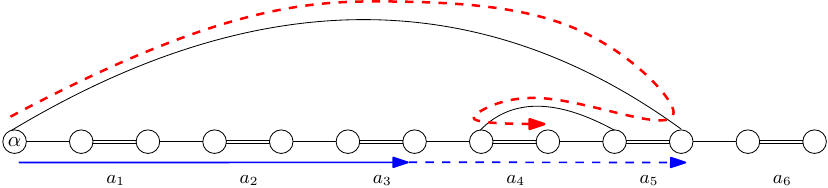}
	\caption{
		In this example, $\alpha$ is a free node, black (full) single-segments are unmatched and black (full) double-segments are matched edges. Assume that the algorithm first explores the red (dashed) path from $\alpha$ to $a_4$. This red path sets labels $\ell(\leftarc{a_5}) = 1$ and $\ell(a_4) = 2$. After this exploration, in the next two passes $\alpha$ backtracks along the red path and sets the active path to be $(\alpha)$. 
Then, $\alpha$ continues extending its active path (over three $\PassBundle$s) along $a_1$, $a_2$ and $a_3$, setting labels $\ell(a_1) = 1$, $\ell(a_2) = 2$ and $\ell(a_3) = 3$, illustrated with the blue solid line. 
However, notice that $\ell(a_4) < 4$ and hence $\alpha$ does not extend the active path to become $(\alpha, a_1, a_2, a_3, a_4)$. On the other hand, $\ell(a_5) = \infty$ and our algorithm should enable $\alpha$ to reach $a_5$. This is achieved by enabling $\alpha$ to ``jump'' over $a_4$ and let the active path become $(\alpha, a_1, a_2, a_3, a_4, a_5)$, illustrated with the blue dashed line. We also set $\ell(a_5) \coloneqq 5$. (In this example, the arc $a_4$ plays role of $b_1$ in the description of \AlgExtendStructures.) Subsequently, in the next $\PassBundle$ $\alpha$ will extend its active path to $(\alpha, a_1, a_2, a_3, a_4, a_5, a_6)$. 
	}
	\label{fig:jumping-example}
\end{figure}



\subsection{$\AlgBacktrack$}
\label{sec:backtracking} 
		\begin{minipage}{0.95\linewidth}
			\begin{mdframed}[backgroundcolor=gray!15, linecolor=red!40!black]
			\emph{Informal description}: \AlgExtendStructures can be seen as performing a Depth First Search (DFS) along active paths. When an active path does not get extended in a pass then, just like in DFS, \AlgBacktrack backtracks on this active path (in our case by one matched and one unmatched arc) and continues the DFS from that shorter active path. The backtracking is not applied to the structures that are on hold, as those structures did not attempt extending their active paths in the corresponding streaming pass. 
			\end{mdframed}
		\end{minipage}
		\\\\
	Let $\alpha$ be a free node with active path $P_\alpha = (\alpha, a_1, \ldots, a_{k - 1}, a_k)$ such that:
	\begin{enumerate}[(1)]
		\item In the last invocation of \AlgExtendStructures the node $\alpha$ did not extend $P_\alpha$.
		\item\label{item:no-overtake} No other free node $\beta$ overtook any arc of $P_\alpha$, i.e., \cref{line:ES-overtake} of \AlgExtendStructures was not invoked on $\astar$ belonging to $P_\alpha$.
		\item $\structure_\alpha$ was \emph{not} on hold in that invocation of \AlgExtendStructures.
	\end{enumerate}
	Then, the algorithm $\AlgBacktrack$ \emph{backtracks} on $P_\alpha$, meaning that the new active path of $\structure_\alpha$ becomes $(\alpha, a_1, \ldots, a_{k - 1})$ and the arc $a_k$ is marked as inactive. If $P_\alpha = (\alpha)$, then the new active path of $\structure_\alpha$ becomes $\emptyset$ and the free node $\alpha$ is marked as inactive.
It trivially holds that $\structure_\alpha$ is still a structure per \Cref{definition:structure} and that \Cref{invariant:active-path-length} is maintained after an execution of $\AlgBacktrack$.

We now illustrate why the Condition~\eqref{item:no-overtake} above is important, i.e., why the algorithm does not backtrack on an active path that was shrunk due to invoking \AlgOvertake. Let $A$ be an active path at the beginning of the current pass.
Assume that during the pass $A$ was reduced to $A'$ during an invocation of \AlgOvertake. Let arc $a'$ be the head of $A'$. Let the $i$-th arc on the stream was the one which lead to $A$ getting reduced to $A'$. Some of the arcs on the stream appearing before the $i$-th one could be extensions of $A'$. However, since the head of $A$ is not $a'$, our algorithm did not have a chance to explore extensions of $A'$ preceding the $i$-th arc; recall that \AlgExtendStructures attempts to extend active paths, not its prefixes. That is why the algorithm does not backtrack on $a'$, but rather attempt to extend it in the next pass. This property is leveraged in our proofs, in particular in the proof of \cref{obs: no new active}.

\subsection{Check for Edge Augmentations (\AlgEdgeMerge)}
\label{sec:edge-merging}
	In a new pass, for each edge $e = \{u, v\}$ in the stream, the algorithm checks whether the structure containing $u$ and the structure containing $v$, if such structures exist, can augment over $e$. If it is possible, via \AlgMerge the algorithm augments the two corresponding free nodes and removes their structures, as described in \cref{operation:merging}. Note that if an augmentation exists, then due to \cref{definition:structure}-\ref{prop:unmatched-arcs} the edge $e$ is \emph{unmatched}. The goal of this operation is to ensure that our algorithm detects an augmentation even if the corresponding structures did not extend, which we use to show some arguments in \cref{lemma: augment} and prove that \cref{obs: opposing arcs} holds. 

\begin{algorithm}[h]
	\SetKwProg{function}{function}{}{\KwRet}
	\SetKw{continue}{continue}

  \Input{$G$: a graph \\
		$M$: the current matching\\
		$\mathcal{P}$: the current set of augmenting paths \\
		$\structure_\alpha$ for each free node $\alpha$}

		\For{each arc $g = (u, v)$ on the stream} {
			\If{$u$ or $v$ was removed by \AlgMerge in this phase} {
				Continue with next arc.
			}

			\If{$\RD(\structure_\alpha \cup \structure_\beta \cup g)$ contains an $\alpha$-$\beta$ alternating path for some $\alpha \neq \beta$}{
				$\AlgMerge(\mathcal{P}, g)$ \tcc*{Store augmentation, remove $\structure_\alpha$ and $\structure_\beta$.}
			}
		}

  \caption{The execution of \AlgEdgeMerge (\cref{sec:edge-merging}).}
\end{algorithm}
\subsection{Augment Structures with Unmatched Edges (\AlgAugmentStructures)}
\label{sec:augment-structures}
	In a new pass, for each \emph{unmatched} edge $e = \{u, v\}$, if $u$ and $v$ are in the same structure $\structure$, then add $e$ to $\structure$. This algorithm is outlined as \cref{algo-augment-structures}.
	Observe that \cref{line:AugmentStructures-unmatched-if} of the algorithm enforces property \cref{definition:structure}-\ref{prop:unmatched-arcs}. It now trivially holds that $\structure_\alpha$ is still a structure per \cref{definition:structure} and that \cref{invariant:active-path-length} is maintained after an execution of $\AlgAugmentStructures$.
\begin{algorithm}[h]
	\SetKwProg{function}{function}{}{\KwRet}
	\SetKw{continue}{continue}

  \Input{$G$: a graph \\
		$M$: the current matching\\
		$\structure_\alpha$ for each free node $\alpha$}
	  
		\For{each edge $e = \{u, v\}$ on the stream} {
			\If{$u$ or $v$ was removed by \AlgMerge in this phase} {
				Continue with next edge.
			}

			\If{$e$ is \emph{unmatched}, and $u$ and $v$ belong to the same structure $\structure_\alpha$ \label{line:AugmentStructures-unmatched-if}} {
				Add $e$ to $\structure_\alpha$.
			}			
		}
		
  \caption{The execution of \AlgAugmentStructures (\cref{sec:augment-structures}).}
  \label{algo-augment-structures}
\end{algorithm}

	Later in our proofs we point out where the existence of method \AlgAugmentStructures is used, and here we illustrate its importance in the example in \cref{fig:jumping-example}. Recall that in that example $\alpha$ ``jumps'' over $a_4$ to reach $a_5$. Here, jumping actually means that, upon seeing on the stream the unmatched edge between $a_4$ and $a_5$, $\alpha$ is aware of the existence of the alternating paths $(\alpha, a_1, a_2, a_3, a_4, a_5)$. For that, $\alpha$ has to be aware of the existence of the unmatched edge between $a_3$ and $a_4$ as well; notice that, however, this unmatched edge was never explored via an active path. Here is why \AlgAugmentStructures is helpful, as it makes sure that in the $\PassBundle$ where the active path is $(\alpha, a_1, a_2, a_3)$ the algorithm will store the unmatched edge between $a_3$ and $a_4$ (recall that $a_4$ was previously explored by $\alpha$ via the red (dashed) path). Consequently, this ensures that $\alpha$ can jump over $a_4$. Hence, informally speaking, \AlgAugmentStructures can be thought as of a way of ensuring that our algorithm will have sufficient information to perform jumping.

	


\section{Correctness}\label{sec: correctness}
Consider one phase of our algorithm.
Notice that according to the description of \AlgPhase, the phase is executed for a fixed amount of iterations.
However, ignore that detail for the moment and suppose that the phase is executed until each free vertex becomes inactive.
We will handle this detail in the approximation analysis.
For our approximation guarantees, the crucial property that we need from our algorithm is that, for each $\alpha - \beta$ alternating path $(\alpha, a_1, \ldots, a_{k}, \beta)$ of length at most $\maxlen$, at least one of the following happens in the phase.
\begin{itemize}
	\item The matching is improved by augmenting over the augmenting path $\alpha - \beta$.
	\item At least one vertex is removed from the $\alpha - \beta$ path due to an augmentation over another augmenting path.
\end{itemize}
This property is shown as the following claim.

\begin{restatable}[Proof for finding short augmenting paths]{lemma}{lemmaaugment}
\label{lemma: augment}
Let $P = (\alpha, a_1, \ldots, a_k, \beta)$ be an augmenting path of length $k \le \maxlen$ that so far was not discovered by our algorithm. Then, $P$ contains an active vertex.
\end{restatable}

\paragraph{Roadmap of the proof of \Cref{lemma: augment}.}
The proof of \Cref{lemma: augment} is based on several structural properties that the algorithm maintains during its execution.
The first simple observation is that if $\alpha$ has found an alternating path of length $i$ to some arc $a_i$, there are two options that can prevent $\alpha$ from extending its path to $a_{i + 1}$.
One option is that $\leftarc{a_{i + 1}}$ is part of the active path.
The other option is that some other free vertex has an active path of length at most $i + 1$ to $a_{ i + 1}$.
Informally speaking, the key observations are that in the former case, by \Cref{lemma: lastleft}, (a suffix of) the active path must form an odd cycle.
A very convenient property of odd cycles is that as soon as they are discovered by the algorithm, their arcs can never belong to two distinct structures of the free vertices.
Otherwise, we will find an augmentation and we have that an augmenting path satisfying one of the two desired properties has been found.
This property is formalized in \Cref{obs: opposing arcs} and the process for finding these odd cycles is formalized in \Cref{def: settled} and \Cref{lemma: sameStr}.

Our main challenge is that on the path $\alpha - \beta$, there can be many events by active paths of many distinct free vertices, where some active paths are blocked by other active paths and others form odd cycles.
Our main technical contribution is to sort this mess and show that certain positive properties are maintained throughout the phase.
Suppose that our algorithm, at some point during the execution of a phase, has found an alternating path of length at most $j$ to each $a_j$ for $1 \leq j < k$.
Here, finding corresponds to having an active path of length $j$ to $a_j$ (at some point in the phase).
We show that before the paths to all $\{ a_1, \ldots, a_{j} \}$ have been found and the corresponding active paths have backtracked without finding an alternating path of length at most $j + 1$ to $a_{j + 1}$, it must be the case that they belong to an odd cycle (this is somewhat imprecise, we actually show that they satisfy a slightly different property, which can be used as if they were in the same odd cycle, cf. \Cref{def: settled}).
Importantly, they will all be in the same structure for the remainder of the phase.

Then, we argue that eventually, the odd cycle formed by $\{ a_1, \ldots, a_{j} \}$ can be used to extend a short active path to $a_j$, either by $\alpha$ or some other free vertex.
These statement are formalized in \Cref{lemma: tagging} and \Cref{lemma: upd-active}.
From this, we can inductively derive that eventually, either all $\{ a_1, \ldots, a_k \}$ form an odd cycle or an augmentation has been found involving some of these arcs.
Once this (large) odd cycle is found, $\alpha$ learns about an augmenting path to $\beta$, and we have again reached one of the desirable properties.
We begin our formal proof by showing the structural properties related to discovered odd cycles.

\begin{observation}\label{obs: opposing arcs}
	Let $a$ and $b$ be matched arcs such that there is an alternating path $(a, g, b)$ in the graph, where $a$ and $b$ are matched arcs and $g$ is an unmatched arc. 
	Assume that $a$ and $\leftarc{b}$ are not removed (due to a found augmentation) in a given phase.
	Let $p$ be a \PassBundle of the phase in which both $a$ and $\overleftarrow{b}$ are reachable. Then, from the end of $p$ until the end of the phase, $a$ and $\leftarc{b}$ cannot belong to different structures.
\end{observation}
\begin{proof}
	Towards a contradiction, assume that $a$ and $\overleftarrow{b}$ are reachable at the end of \PassBundle $p$, but there is \PassBundle after the $p$-th one during which $a$ and $\overleftarrow{b}$ belong to different structures.
	
	By construction, $a$ will belong to some (potentially different ones during the phase) structure at any point in the phase. 
	Hence, $a$ and $b$ are reachable, which implies that $a \in \structure_1$ and $\leftarc{b} \in \structure_2$ after step \AlgExtendStructures of the $p$-th \PassBundle. 
	If $\structure_1 \neq \structure_2$, then \AlgEdgeMerge will find an augmentation between the free nodes of $\structure_1$ and $\structure_2$, which would contradict our assumption that $a$ and $\leftarc{b}$ are not removed in the phase. So, $\structure_1$ and $\structure_2$ are the same structure.
	
	Further, since $a$ and $\leftarc{b}$ belong to the same structure, \AlgAugmentStructures will in the $p$-th \PassBundle add to the memory of our algorithm the unmatched edge $e$ connecting these two arcs. 
	The edge $e$ will remain in the memory of our algorithm until the end of the phase. 
	Suppose that after an edge from the stream is processed the arcs $a$ and $\leftarc{b}$ do not belong to the same structure anymore, e.g., due to an invocation of $\AlgOvertake$.
	Now since the structures containing $a$ and $\leftarc{b}$ are connected and disjoint, our algorithm will find an augmentation consisting of the paths to $a$ and $\leftarc{b}$ in their corresponding structures and the edge $e$.
	This would again contradict our assumption that $a$ and $\leftarc{b}$ are not removed until the end of the phase.
\end{proof}


\begin{definition}[Settled]\label{def: settled}
	We say that an alternating path $P = (a_1, a_2, \ldots, a_k)$ is \emph{settled} if for all $1 \leq j \leq k$, both $a_j$ and $\leftarc{a_j}$ are reachable.
\end{definition}
Let $P$ be an alternating path belonging to $\structure_\alpha$. A convenient property of having $P$ settled is that, once $P$ becomes settled, we show that all the arcs in $P$ at any point belong to the same structure (this is made formal in \cref{lemma: sameStr}). Notice that invoking the method \AlgAugmentStructures ensures that our algorithm, in its memory, also stores unmatched arcs belonging to $P$. Having all this, it enables us to think of $P$ as an odd cycle. In particular, if any free node $\beta \neq \alpha$ reaches $P$, then the algorithm augments between $\alpha$ and $\beta$, which makes progress. Crucially, and deviating from standard approaches to finding augmenting paths, we \emph{allow} our augmenting paths to be much longer than $1/\eps$. (In fact, we do not guarantee that augmenting via settled paths will result in augmenting paths of length at most $1/\eps$.) Nevertheless, in \cref{sec:pass} we show how to handle these longer augmenting paths by executing additional phases.

\begin{remark*}
	For technical reasons and for convenience, we may also say that sets containing arcs and their reverses are settled, e.g., $\{ a_1, \leftarc{a_1}, \leftarc{a_2} \}$ is the same as path $(a_1, a_2)$ being settled.
\end{remark*}

\begin{lemma}\label{lemma: sameStr}
	Consider an alternating path $(a_1, a_2, \ldots, a_{k + 1} )$ and suppose that $S = \{a_2, \ldots, a_k\}$ is already settled at the beginning of $\PassBundle$ $\tau$.
	Moreover, assume that no $\{a_2, \ldots, a_k\}$ is removed (due to a found augmentation) in a given phase.
	Then, the following holds from the beginning of $\tau$ until the end of the phase:
	\\
	No two arcs in $S$ can belong to different structures.
	Furthermore, if $\ell(a_1) < \infty$ already at the beginning of $\tau$, then $a_1$ belongs to the same structure as $S$, and the same claim holds for $\leftarc{a_{k + 1}}$ as well.
	Notice that different structures might hold $S$ during the phase.
\end{lemma}
\begin{proof}
	By definition, $\ell(a_j) < \infty$ and $\ell(\leftarc{a_{j + 1}}) < \infty$ for all $2 \leq j \leq k - 1$ already after the execution of \AlgExtendStructures in $\PassBundle$ $\tau - 1$.
	Then, the lemma follows from \Cref{obs: opposing arcs} applied to all $a_j$ and $\leftarc{a_{j + 1}}$ for $1 \leq j \leq k$.
\end{proof}

\subsection{Jumping Invariant}
Let $P = (\alpha, a_1, \ldots, a_k, \beta)$ be an $\alpha-\beta$ augmenting path of length at most $1/\eps$. \Cref{lemma: augment} can be seen as a way of saying that there is no situation in which our algorithm blocks itself in a sense that (i) $P$ is not discovered by our algorithm so far; and (ii) $P$ will not be discovered even if the phase lasts indefinitely.
In an ideal situation, our algorithm visits $a_1$ from $\alpha$, then $a_2$ along $(\alpha, a_1)$ and so on, until it finally reaches $\beta$ to augment. However, it could be the case that some other free node already has $a_i$ in its active path by the time $\alpha$'s active path reaches $a_i$. Moreover, it could be the case that $\alpha$ attempts to extend to $a_i$ via an active path that \emph{already} contains $\leftarc{a_i}$, and hence this exploration/extension is impossible (as it would lead to the active path containing a cycle). This is somewhat undesirable but, nevertheless, in situation like these we will often argue that then something else good had happened, e.g., a set of arcs has become settled. In addition to settled, we will introduce the notion of \emph{jumping invariant} (as outlined next), and a significant part of our proofs will be devoted to showing that this invariant in fact holds throughout the entire execution. (As a reminder, the notion of jumping over arcs is defined in \cref{definition:jumping}.)

In particular, a crucial part of our proof is to show that whenever we perform a jump over a set of arcs, these arcs must already be settled.
Informally speaking, we can think that these arcs were already ``dealt with earlier'' and can be ignored in the future steps of the algorithm\footnote{This idea is somewhat similar to contracting blossoms but we need to keep track of the lengths of the paths.}.
\begin{invariant}[Jumping Invariant] \label{invariant: jump}
	Consider an active path $P$ and let $P(a)$ be the prefix of $P$ until some arc $a$.
	If $|P(a)| > \ell(a)$, then $\leftarc{a}$ is reachable.
\end{invariant}

\noindent Our proof of \cref{invariant: jump} is built on top of the following claim which will be proven in \Cref{sec: jump-proof}.
\begin{lemma}\label{lemma: jump-proof}
	Suppose that \Cref{invariant: jump} holds by the beginning of $\PassBundle$ $\tau$.
	Suppose that in $\PassBundle$ $\tau$, an active path $P_\alpha = (\alpha, a_1, \ldots, a_k)$ performs a jump, i.e., the active path is extended to $P_\alpha = (\alpha, a_1, \ldots, a_k) \circ (b_1, \ldots, b_h) \circ (b_{h + 1})$.
	Then, at the beginning of \PassBundle $\tau$ it holds that $\{b_1, \ldots, b_h\}$ is already settled.
\end{lemma}

\paragraph{Proof of the Jumping Invariant (\cref{invariant: jump}).}
Our next step is to prove that \cref{invariant: jump} holds at every point of the execution.
Clearly, the invariant holds in the beginning of the execution, since all active paths are trivial.
Suppose then that \cref{invariant: jump} holds by the beginning of some \PassBundle $\tau$.
Recall from \Cref{subsec:algorithm} the specifications of \AlgExtendStructures.
Consider an active path $P$ and let $P(a)$ be the prefix of $P$ until some arc $a$.
Notice that if $|P(a)| > \ell(a)$, for a prefix $P(a)$ of $P$, it must have been the case that \AlgExtendStructures jumped over arc $a$ when $a$ was added to the active path. 
In \Cref{lemma: jump-proof}, we show that for each path that is jumped over in \PassBundle $\tau$, the invariant holds, i.e., we show that such path is already settled.
Hence, the invariant also holds throughout \PassBundle $\tau$ (and hence by the beginning of $\PassBundle$ $\tau + 1$).
Then, the invariant holds for any \PassBundle by the principle of induction.
So, in the rest, we focus on showing \cref{lemma: jump-proof} and then applying those results to prove \Cref{lemma: augment}.

\subsection{Analysis of \cref{lemma: jump-proof}} \label{sec: jump-proof}
\newtext{
The following two lemmas essentially say that if a short active path $A$ reaches a sequence of settled arcs $\{a_f, \ldots, a_k\}$, then $A$ will also reach $a_{k + 1}$ unless $A$ has already visited an arc $\{\leftarc{a_f}, \ldots, \leftarc{a_{k}}\}$.
}
\begin{lemma}\label{obs: no new active}
	Consider an alternating path  $P = (a_{f - 1}, \ldots, a_{k + 1})$ and suppose that $S = \{ a_f, \ldots, a_k \}$ is already settled at the beginning of $\PassBundle$ $\tau$.
	Assume that 
	\begin{itemize}
		\item in $\PassBundle$ $\tau$ there is an active path (or its prefix) $P_\alpha(b_{k'}) = (\alpha, b_1, \ldots, b_{k'})$ such that $b_{k'} = a_{f - 1}$ and $k' \leq f - 1$
		\item there is no $b_h$ for $h < k'$ such that $b_h = \leftarc{a_j}$, for any $j$ such that $f \leq j \leq k + 1$.
	\end{itemize}
	Let $f - 1 \leq r \leq k + 1$ be the largest index such that $a_r \in P_\alpha(b_{k'})$.
	Then, before our algorithm backtracks arc $a_{r}$ in $\PassBundle$ $\tau$ or later, it holds that $\ell(a_{k  + 1}) \leq k + 1$.
\end{lemma}
\begin{proof}
	Note first that an index $r$ as specified in the lemma statement exists since $a_{f - 1} \in P_\alpha(b_{k'})$.
	Let $P_\alpha$ be the active path of $\structure_\alpha$.
	Our goal is to show that eventually, before all arcs in $P_\alpha \cap P$ are inactive, there is an active arc $a_r$, for some $f - 1 \leq r \leq k + 1$, such that no arc $a_j$ or $\leftarc{a_j}$ belongs to the active path, for $r < j \leq k + 1$. We will show that it holds for an arc $a_r \in P_\alpha(b_{k'})$.	
	The claim that no such $\leftarc{a_j}$ exists if $a_r \in P_\alpha(b_{k'})$ is provided by the lemma statement.	
	It remains to show that no such $a_j$ exists. 

	Now, we show that the prefix $P_{\alpha}(a_r) = (\alpha, b_1, \ldots, b_{j - 1}, a_r)$ of the active path can be extended to $a_{k + 1}$ along $P$ before backtracking from $a_r$. We also show that the length of such extended path is at most $k + 1$.
	The claim about length is immediate as $r \geq f - 1$ and $|P_{\alpha}(a_r)| \leq r$.

	For the former claim, we make a few observations and conclude that there is a pass when $a_r$ is active and no $\{ a_{j'}, \leftarc{a_{j'}} \mid r < j' \leq k + 1 \} = S'$ is active and all arcs in $S'$ are in the same structure as $a_r$.
	
	\begin{itemize}
		\item By \Cref{lemma: sameStr}, no two arcs $\{ a_{f - 1}, \ldots, a_{k} \}$ can belong to different structures. This holds at latest starting from $\PassBundle$ $\tau$.
		\item Observe that an arc cannot become inactive unless backtracked or removed. This applies to the active arc $a_{r}$ as well.
		\item We show that no arc in $S'$ can be added to the prefix of $P_{\alpha}(a_r)$ before $a_r$ is backtracked.
	Suppose that there is such an arc $c \in S'$.
	Then, there must be a \PassBundle $\tau' \ge \tau$, where some $\structure_\gamma$, that does not contain $a_r$, makes $c$ active.
	Otherwise, $c$ cannot be added to the prefix of the active path containing $a_r$ in \PassBundle $\tau'$.
	Suppose without loss of generality that $a_r \in \structure_\alpha$.
	Notice that in \PassBundle $\tau'$, it cannot be the case that $\gamma$ overtakes\footnote{\emph{Overtaking} here refers to an invocation of \AlgOvertake in which the arc $c$ is added to $\structure_\gamma$.} $a_r$ since \AlgOvertake moves only those matched arcs from $\structure_\alpha$ to $\structure_\gamma$ for which all alternating paths from $\alpha$ contain $c$. Since $a_r$ is in the active path of $\structure_\alpha$, there is an alternating path from $\alpha$ to $a_r$ that does not contain $c$.	
	We concluded above that $a_r$ and $c$ are in the same structure and, hence, no such arc $c$ exists.	
	\end{itemize}

\noindent Thus, it must be the case that eventually arc $a_r$ is active and no  arc in $\{ a_j, \leftarc{a_j} \mid r < j \leq k + 1 \}$ is active.
	Furthermore, since $\{ a_{f}, \ldots, a_{k} \}$ are reachable, the structure whose active arc is $a_r$ (implicitly) contains all the arcs related to the path $(a_r, \ldots, a_k)$.
	Then, before backtracking $a_r$, our algorithm extends the active path to $a_{k + 1}$ and sets $\ell(a_{k + 1}) \leq k + 1$.
\end{proof}


\begin{lemma}\label{lemma: lastleft}
	Consider an alternating path  $P = (a_{f - 1}, \ldots, a_{k + 1})$ and suppose that $S = \{a_f, \ldots, a_k\}$ is already settled at the beginning of $\PassBundle$ $\tau$.
	Assume that in $\PassBundle$ $\tau$ there is an active path (or its prefix) $P_\alpha(b_{k'}) = (\alpha, b_1, \ldots, b_{k'})$ such that $b_{k'} = a_{x}$, for some $f - 1 \leq x \leq k$ and $k' \le x$.	
	Suppose that from $\PassBundle$ $\tau$ it holds that $\ell(a_{k + 1}) > k + 1$ at least while an arc in $\{a_{f - 1}, \ldots, a_{k}\}$ is active.
	Then, there is an index $s < k'$ such that $b_s = \leftarc{a_j}$ for $x  < j \leq k + 1$.
	Suppose that $s < k'$ is the smallest such index.
	Then, there is no $s' < s$ such that $b_{s'} = a_{j'}$ for some $x \leq j' \leq k + 1$.
\end{lemma}
\begin{proof}
	This lemma is a corollary of \Cref{obs: no new active}.
	First, we want to show that an index $s$, as specified in the lemma statement, exists.
	For a contradiction, suppose that there is no such index $s$.
	By an application of \Cref{obs: no new active}, such that $x$ corresponds to $f - 1$ in \Cref{obs: no new active}, we obtain that $\ell(a_{k + 1}) \leq k + 1$ before every arc in $\{a_{x}, \ldots, a_{k}\}$ is backtracked yielding a contradiction.

	Then, we prove the latter claim in the lemma statement.
	For a contradiction, suppose that there is an index $s'$ as in the lemma.
	By definition, we have a prefix of the active path $P_\alpha(b_{s'}) = (\alpha, b_1, \ldots, a_{j'})$ such that $\{a_{j' + 1}, \ldots, a_k\}$ is settled.
	Furthermore, there is no $b_h = \leftarc{a_{j}}$ such that $h < s'$ and $s' < h \leq k$ by the lemma statement.
	We can apply \Cref{obs: no new active}, where in the application $j'$ from this proof corresponds to $f-1$ in \Cref{obs: no new active}, to obtain that $\ell(a_{k + 1}) \leq k + 1$ before every arc in $\{a_{f - 1}, \ldots, a_{k}\}$ is backtracked yielding a contradiction.
\end{proof}

\newtext{
\noindent The following claim is one of crucial building blocks in our coming proofs. 
Informally speaking, it gives us a handle to (inductively) argue that once all arcs in $(\alpha, a_1, \ldots, a_k)$ have been searched with a short augmenting path, but $a_{k + 1}$ has not, it must either be the case that (at least) one of them is still active, or $(a_1, \ldots, a_k)$ is settled.
This can then, in turn, be used to argue that before $\alpha$ becomes inactive, it will have a chance to visit $a_{k + 1}$ through $(a_1, \ldots, a_k)$.
Note that in the statement below condition that $\{ a_{i}, \leftarc{a_{i}} \mid 2 \leq i \leq k \}$ are reachable implies that $\{ a_{i}, \leftarc{a_{i}} \mid 2 \leq i \leq k \}$ are settled.
}
\newcommand{\cmax}{\ensuremath{c_{\textrm{max}}}}
\begin{lemma}\label{lemma: non-active cycle}
	Consider an alternating path $P = (a_1, a_2, \ldots, a_{k + 1})$ such that, for a free node $\alpha$, $a_1, \leftarc{a_{k + 1}} \in \structure_\alpha$ and $\{ a_{i}, \leftarc{a_{i}} \mid 2 \leq i \leq k \}$ are reachable.
	The alternating distance from $\alpha$ to $a_1$ does not need to be $1$.	
	Suppose that there is an alternating path $P(a_{x})$, for $2 \leq x \leq k + 1$, from $\alpha$ to $a_x$ such that $a_j, \leftarc{a_j} \not \in P(a_x)$ for all $1 \leq j < x$.
	Then, $\leftarc{a_1}$ is reachable, i.e., there is an alternating path from $\alpha$ to $\leftarc{a_1}$ within $\RD(\structure_\alpha)$.
\end{lemma}
\begin{proof}
	Toward a contradiction, we will assume that $\leftarc{a_1}$ is not reachable, i.e., there is no alternating path from $\alpha$ to $\leftarc{a_1}$ within $\structure_\alpha$.
	The idea of the proof is to show that if $\leftarc{a_1}$ is not reachable, then it is the case that all paths to all arcs $\{ a_j, \leftarc{a_j} \mid 2 \leq j \leq a_{k} \}$ contain arc $a_1$.
	Then, we observe that this is a contradiction with the fact that $P(a_x)$ does not contain $a_1$.
	For an illustration of the inductive step, please refer to \Cref{fig: cycles}.
	
	\paragraph{Base Case:}
	For the base case, suppose for a contradiction that there is an alternating path $P(\leftarc{a_2})$ to $\leftarc{a_2}$ such that $a_1 \not \in P(\leftarc{a_2})$.
	Then, $P(\leftarc{a_2}) \circ (\leftarc{a_1})$ is an alternating path leading to a contradiction with the fact that $\leftarc{a_1}$ is not reachable. Hence, $a_1 \in P(\leftarc{a_2})$.
	\\\\
	Now, suppose for a contradiction that there is an alternating path $P(a_2) = (\gamma, b_1, b_2, \ldots, b_h)$ to the arc $b_h = a_2$ such that $a_1 \not \in P(a_2)$ and $\gamma$ is some free vertex.
	Consider the suffix $P(a_1, \leftarc{a_2}) = (c_1, \ldots, c_h)$ of $P(\leftarc{a_2})$, where $c_1 = a_1$ and $c_{h} = \leftarc{a_2}$.
	Notice that such an alternating path must exist since $a_1$ is on every alternating path to $\leftarc{a_2}$.
	Let $j$ be the smallest index such that $b_j \in P(a_1, \leftarc{a_2})$ or $\leftarc{b_j} \in P(a_1, \leftarc{a_2})$.
	Notice that such a $j$ must exist since $b_h = a_2$. Since $a_1 \not \in P(a_2)$, then $b_j \neq a_1$.
	Consider first the case that $b_j = c_{j'}$ for some $j' > 1$.
	Now, we have that $(\gamma, b_1, \ldots, b_{j - 1}) \circ (c_{j'}, c_{j' + 1}, \ldots, \leftarc{a_2}, \leftarc{a_1})$ is an alternating path.
	Then, consider the case that $b_j = \leftarc{j'}$ for some $j'$.
	We have that $(\gamma, b_1, \ldots, b_{j - 1}) \circ (\leftarc{c_{j'}}, \leftarc{c_{j' + 1}}, \ldots, \leftarc{c_1})$ in an alternating path, which is a contradiction since $\leftarc{c_1} = \leftarc{a_1}$ and we assumed that $\leftarc{a_1}$ is not reachable.

	\paragraph{Inductive Step:}
	Suppose now that the claim holds for all $a_1, \ldots, a_i$, i.e., all paths (in $\structure_\alpha$) to all arcs in $S = \{ a_j, \leftarc{a_j} \mid 1 \leq j \leq i  \}$ contain the arc $a_1$.
	We do a similar case distinction as before.
	We begin with showing that also all paths to $\leftarc{a_{ i + 1 }}$ must contain the arc $a_1$.
	Let $P(\leftarc{a_{ i + 1 }})$ be an alternating path to $\leftarc{a_{ i + 1 }}$.
	First, we observe that if no arc $a_j, \leftarc{a_j} \in S$, for $j \leq i$, is contained in $P(\leftarc{a_{ i + 1 }})$, then $P(\leftarc{a_{ i + 1 }}) \circ (\leftarc{a_{i}}, \ldots, \leftarc{a_1})$ is an alternating path to $\leftarc{a_1}$.
	Therefore, some arc in $S$ must belong to $P(\leftarc{a_{j + 1}})$.
	By the inductive assumption, we know that all paths to all arcs in $S$ contain $a_1$ and hence, any path to $P(\leftarc{a_{ i + 1 }})$ must also contain $a_1$.
	
	Now we proceed to the more involved case that all paths to $a_{i + 1}$ contain the arc $a_1$.
	Let $P(a_{i + 1}) = (\alpha, b_1, b_2, \ldots, b_h)$ be an alternating path to the arc $b_h = a_{i + 1}$ such that $a_1 \not \in P(a_{i + 1})$.
	Furthermore, let $P(\leftarc{a_{i + 1}})$ be an alternating path to $\leftarc{a_{i + 1}}$ and consider the suffix $P(a_1, \leftarc{a_{i + 1}}) = (c_1, \ldots, c_{h'})$ of $P(\leftarc{a_{i + 1}})$, where $c_1 = a_1$ and $c_{h'} = \leftarc{a_{i + 1}}$.
	Such a suffix must exists since $a_1$ is on every alternating path to $\leftarc{a_{i + 1}}$.
	
	Let $j$ be the smallest index such that $b_j \in P(a_1, \leftarc{a_{i + 1}})$ or $\leftarc{b_j} \in P(a_1, \leftarc{a_{i + 1}})$.
	Notice that such a $j$ must exist since $b_h = a_{i + 1}$. Since $a_1 \not \in P(a_{i + 1})$, then $b_j \neq c_1 = a_1$.
	Consider first the case that $b_j = \leftarc{c_{j'}}$ for some $2 \leq {j'} \leq h'$.
	Then, we have that $(\gamma, b_1, \ldots, b_{j - 1}) \circ (\leftarc{c_{j'}}, \leftarc{c_{j' - 1}}, \ldots, \leftarc{a_1})$ is an alternating path to $\leftarc{a_1}$, which is a contradiction with $\leftarc{a_1}$ not being reachable.
	
	Consider then the case $b_j = c_{j'}$.
	Observe that since $a_1 \not \in (b_1, \ldots, b_{j - 1})$ and $a_1 \not \in (c_{2}, \ldots, c_{h'})$, no arc $a_x \in S$ can be contained in $(b_1, \ldots, b_{j - 1}) \circ (c_{j'}, \ldots, c_{h'})$.
	Otherwise, we would have a path to $a_x$ that does not contain $a_1$, which cannot be the case by the inductive assumption.
	Furthermore, since $j$ was the smallest index, we have that $(\gamma, b_1, \ldots, b_{j - 1}) \circ (c_{j'}, \ldots, c_{h' - 1}) \circ (\leftarc{a_{i - 1}}, \leftarc{a_i}, \ldots, \leftarc{a_1})$ is an alternating path to $\leftarc{a_1}$.
\end{proof}

\begin{figure}
	\centering
	\includegraphics[width=0.95\linewidth]{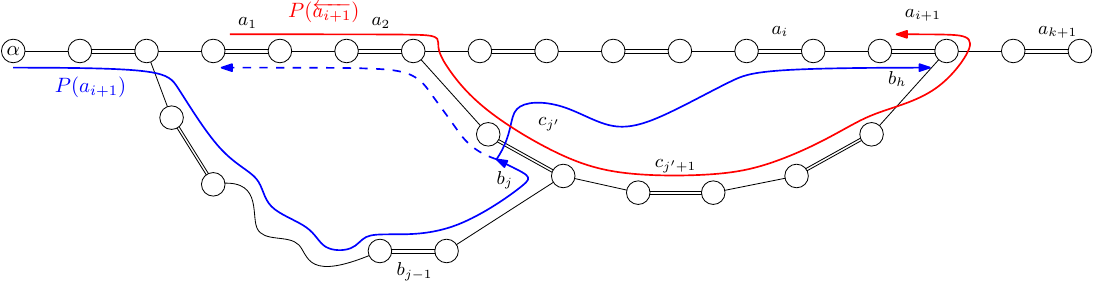}\\\vspace{0.7cm}
	\includegraphics[width=0.95\linewidth]{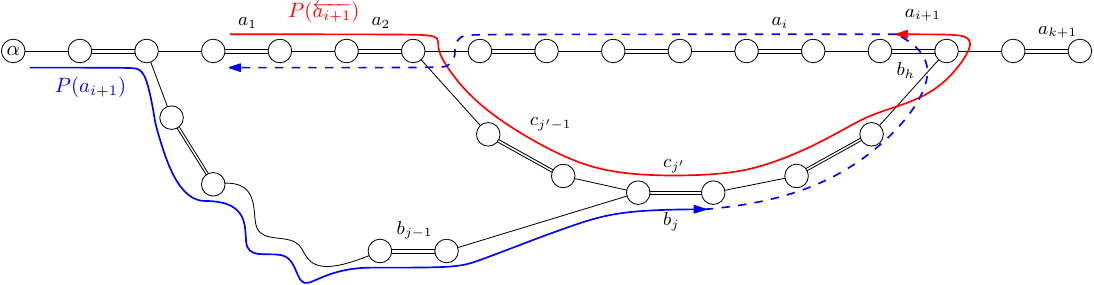}
	\caption{Illustrations for the inductive step of \Cref{lemma: non-active cycle}. The upper picture illustrates the case where the alternating path $P(a_{i + 1}) = (b_1, \ldots, b_h)$ intersects the suffix $(c_1, \ldots, c_{h'})$ alternating path $P(\protect\leftarc{a_{i + 1}})$ at some $\protect\leftarc{b_j} = c_{j'}$. In this case, we can complete the prefix $(b_1, \ldots, b_{j})$ into a path to $\protect\leftarc{a_1}$ using the prefix of $P(\protect\leftarc{a_{i + 1}})$. In the lower, we consider the case that $b_j = c_{j'}$. In this case, we can complete the prefix $(b_1, \ldots, b_{j})$ into an alternating path to $\protect\leftarc{a_1}$ using the suffix of $P(\protect\leftarc{a_{i + 1}})$ and the path $\protect\overleftarrow{a_1, \ldots, a_{i + 1}}$. In the proof, we show that $(b_1, \ldots, b_{j - 1}) \circ (c_{j'}, \ldots c_{h'})$ cannot intersect the path $(a_1, \ldots, a_{i + 1})$.
	}
	\label{fig: cycles}
\end{figure}

\begin{lemma}\label{lemma: tagging}
	Consider an alternating path  $(a_1, \ldots, a_{k + 1})$.
	Suppose that $S = \{a_f, \ldots, a_k\}$ is already settled and $\ell(\leftarc{a_{k + 1}}) < k$ at the beginning of $\PassBundle$ $\tau$. 
	At $\tau$, let $P_\alpha$ be an active alternating path (or its prefix) $(\alpha, b_1, \ldots, b_h)$ such that $b_h = a_x$ for some $f \leq x \leq k$. 
	Furthermore, let $j$ be the largest index such that $j < f$ and at the beginning of the $\PassBundle$ $\tau$ either 
	\begin{enumerate}
		\item[(1)] $a_j$ or $\leftarc{a_j}$ is active and $\ell(a_j) \leq j$, or
		\item[(2)] $\ell(a_{j}) > j$.
	\end{enumerate}
	If such $a_j$ does not exist, then let $j = 0$.
	Then, $\{a_{j + 1}, \ldots, a_{f - 1}\} \cup S$ is also already settled at the beginning of \PassBundle $\tau$.
\end{lemma}
\begin{proof}
	Notice that if $j = f - 1$, then there is nothing to show.
	Hence, suppose that $j < f - 1$.	
	We want to show that under the assumption that $\{a_s, \ldots, a_k\}$ is already settled at the beginning of $\PassBundle$ $\tau$, for some $s$ such that $j + 1 < s \leq f$, then $\{a_{s - 1}\} \cup \{a_s, \ldots, a_k\}$ is also already settled at the beginning of $\PassBundle$ $\tau$. To that end, we consider several cases.
	
	\noindent{\bf Case: $\leftarc{a_{s-1}}$ is reachable at the beginning of $\PassBundle$ $\tau$.}
		Observe that $j + 1 < s$ implies $j < s - 1$. This in particular implies that $\ell(a_{s-1}) \le s - 1$; if it would hold $\ell(a_{s-1}) > s - 1$, then we would have $j \ge s - 1$, contradicting our choice of $j$. Together with $\leftarc{a_{s-1}}$, being reachable we now have that $a_{s-1}$ is settled.
	\\\\
	\noindent{\bf Case: $\ell(a_{s-1}) < \infty$ at the beginning of $\PassBundle$ $\tau$.}
	We obtain this by an application of \Cref{lemma: non-active cycle} and next, we show that the prerequisites for the lemma are satisfied.

	By \Cref{lemma: sameStr}, after the \PassBundle where $S$ became settled, it cannot be the case that arcs in $\{ a_{s - 1} \} \cup \{ \leftarc{a_{k + 1}} \} \cup S = S'$ belong to different structures after that $\PassBundle$.
	Notice that due to this lemma statement, for any arc $b \in S'$ for which it holds that $b$ is reachable at the beginning of the $\PassBundle$ $\tau$, it is the case that they belong to same structure.
	
	Next, let $s \leq x' \leq x$ be the  smallest index such that $a_{x'} \in P_{\alpha}$ or $\leftarc{a_{x'}} \in P_{\alpha}$. Notice that such $x'$ exists as $a_x \in P_\alpha$.
	Notice that in the case that $\leftarc{a_{x'}} \in P_{\alpha}$, we immediately obtain that there is an alternating path from $\alpha$ to $\leftarc{a_{s - 1}}$ through the active path and $(\leftarc{a_{x' - 1}}, \ldots, \leftarc{a_{s-1}})$. (All the unmatched edges on this path were added to the corresponding structure before $\PassBundle$ $\tau$ via \AlgAugmentStructures.)
	Hence, $\leftarc{a_{s-1}}$ is reachable and $\{a_{s - 1}, \ldots, a_k\}$ settled.
	
	Now, suppose that $a_{x'} \in P_\alpha$.
	In this case we apply \Cref{lemma: non-active cycle} and similarly obtain that $\leftarc{a_{s - 1}}$ is reachable. To apply \Cref{lemma: non-active cycle}:
	\begin{itemize}
		\item We let $\{a_s, \ldots, a_k\}$  from this proof correspond to $\{ a_{i}, \leftarc{a_{i}} \mid 2 \leq i \leq k \}$ from \Cref{lemma: non-active cycle}.
		\item We let $a_{x'}$ from this proof correspond to $a_x$ from \Cref{lemma: non-active cycle}.
		\item We let $a_{s - 1}$ from this proof correspond to $a_1$ from \Cref{lemma: non-active cycle}.
		\item Finally, we observe that $\ell(a_{s-1}) \le s - 1$, as otherwise we would have $j \ge s - 1$ but we assumed $j + 1 < s$. Then, also observe that neither $a_{s - 1}$ nor $\leftarc{a_{s-1}}$ is active, as otherwise and together with $\ell(a_{s-1}) \le s - 1$ we would again have $j \ge s - 1$. This now implies that $a_i \notin P_{\alpha}$ and $\leftarc{a_i} \notin P_{\alpha}$ for every $s - 1 \le i < x'$.
	\end{itemize}
	
	Therefore, in both cases, we have that $\{a_{s - 1}\} \cup \{a_s, \ldots, a_k\}$ is settled.

	We now apply the argumentation inductively as long as there is an $s$ such that $j + 1 < s$ and $(a_s, \ldots, a_k)$ is settled, and the claim follows.
\end{proof}

\begin{lemma}\label{lemma: upd-active}
	Suppose that \cref{invariant: jump} and all of the below holds at the beginning of \PassBundle $\tau$.
	Consider an alternating path  $(a_1, \ldots, a_k, a_{k + 1})$.
	Suppose that $S = \{ a_f, \ldots, a_k \}$ is settled, $\ell(a_{f - 1}) \leq f - 1$, and $\ell(a_{k + 1}) > k + 1$ at least while an arc in $\{a_{f - 1}, \ldots, a_k\}$ is active, where ``is active'' is measured from \PassBundle $\tau$ onward.
	Let $\alpha$ be a free vertex and $P_\alpha = (\alpha, b_1, \ldots, b_{k'})$ be an active path (or its prefix) such that $b_{k'} = a_{x}$ for some $f - 1 \leq x \leq k$ and  $k' \leq x$.
	Also, suppose that either $a_{f - 1}$ or $\leftarc{a_{f - 1}}$ is active. 
	Then, $\{a_{f - 1}\} \cup S$ is settled in the beginning of \PassBundle $\tau$.
\end{lemma}
\begin{proof}
	We consider two cases separately: $\leftarc{a_{f - 1}}$ is active and $a_{f - 1}$ is active.
	\\\\
	\noindent{\bf Case: $\leftarc{a_{f - 1}}$ is active.}
	Since $\leftarc{a_{f - 1}}$ is active, we have that $\leftarc{a_{f - 1}}$ is reachable. By our assumption it holds $\ell(a_{f - 1}) \leq f - 1$. 
	Hence, $S \cup \{a_{f - 1}\}$ is settled. 
	\\\\
	\noindent{\bf Case: $a_{f- 1}$ is active.}
	Now, we want to show that either the algorithm can extend the active path to $a_{k + 1}$ and set $\ell(a_{k + 1}) \leq k + 1$, or that the algorithm has jumped $a_{f - 1}$. 
	We have the following two subcases.
	Let $P_\alpha(a_{f - 1})$ be the prefix of the active path until $a_{f - 1}$.
	If $|P_\alpha(a_{f - 1})| > f - 1$, we have by \cref{invariant: jump} that $\leftarc{a_{f - 1}}$ is reachable, concluding this case.
	Finally, suppose that $|P_\alpha(a_{f - 1})| \leq f - 1$.	
	
	Our goal is to apply \Cref{lemma: lastleft}, where $f-1$ from this proof plays the role of $x$ in \Cref{lemma: lastleft}. For that, let us confirm that we satisfy all the prerequisites:
	\begin{itemize}
		\item By this lemma statement, we have that $\{a_f, \ldots, a_k\}$ is settled and $\ell(a_{k + 1}) > k + 1$ at least while an arc in $\{a_{f - 1}, \ldots, a_k\}$ is active.
		\item By our assumption, we also have $|P_\alpha(a_{f - 1})| \leq f - 1$, which corresponds to the constraints ``$k' \le x$'' of \Cref{lemma: lastleft}.
	\end{itemize}
	Then, by \Cref{lemma: lastleft}, there is an index $s < k'$ such that $b_s = \leftarc{a_j}$ for some $j \geq f$ and there is no $s' < s$ such that $b_{s'} = a_{j'}$ for any $f-1 \leq j' \leq k + 1$. Hence, there is an alternating path $P_\alpha(\leftarc{a_j}) \circ (\leftarc{a_{j - 1}}, \ldots, \leftarc{a_{f - 1}})$.
	\\
	Therefore, $\leftarc{a_{f - 1}}$ is reachable and $\{a_{f - 1}\} \cup S$ is settled.
	
	We have now gone through all the cases and the claim follows.
\end{proof}

\begin{lemma}\label{lemma: settle before passive}
	Consider an alternating path $(a_f, \ldots, a_k, a_{k + 1})$. Suppose that at the beginning of $\PassBundle$ $\taustar$ the following is true:
	\begin{itemize}
		\item \cref{invariant: jump} holds during every $\PassBundle$ so far,
		\item $\ell(a_j) \leq j$ for all $f \leq j \leq k$,
		\item $\ell(a_{k + 1}) > k + 1$, and $k + 1 \leq \maxlen$,
		\item no $a_j$ is active for $f \leq j \leq k$.
	\end{itemize} 
	Then, $\{ a_j, \leftarc{a_j} \mid f \leq j \leq k \}$ is settled in the beginning of $\PassBundle$ $\taustar$.
\end{lemma}
\begin{proof}
	The goal is to give an inductive proof, starting from showing that $a_k$ gets settled before \PassBundle $\taustar$.
	In $\PassBundle$ $\tau$, let $s$ be the minimum index such that $\{ a_j, \leftarc{a_j} \mid s \leq j \leq k \}$ is settled.
	If such an index does not exist, let $s = k + 1$.
	Denote $S_\tau = \{ a_j, \leftarc{a_j} \mid s \leq j \leq k \}$.
	Our induction maintains the crucial property until all $a_j$ for $f \leq j \leq k$ are settled, it is the case that if $s > f$, then $\ell(a_{s - 1}) > s - 1$.
	
	\paragraph{Base Case:} 
	In the beginning of the phase, we have that $\ell(a_k) = \infty > k$.
	Consider a \PassBundle $\tau$, where an active path $P_{\gamma}(\gamma, b_1, \ldots, b_{k'})$ sets $\ell(a_k) = k' \leq k$, where $\gamma$ is a free vertex.
	Suppose for a contradiction that $\leftarc{a_{k + 1}} \not \in P_{\gamma}$.
	Then, by \Cref{obs: no new active}, we have that before our algorithm backtracks the active arc $a_k$, we have that $\ell(a_{k + 1}) \leq k + 1$, leading to a contradiction. (Note that here we refer to the first backtracking that happens after \PassBundle $\tau$. Moreover, since $a_k$ is not active at the beginning of $\taustar$, this backtracking and labeling $\ell(a_{k + 1}) \leq k + 1$ happens before $\taustar$.)
	Hence, $\leftarc{a_{k + 1}} \in P_{\gamma}$ and $\ell(\leftarc{a_{k + 1}}) < k$.
	Since $a_k$ does not belong to the prefix of $P_{\gamma}$ ending at $\leftarc{a_{k + 1}}$, it follows that $\leftarc{a_{k}}$ is reachable and hence,  $\{a_k\}$ is settled.

	Now it is left to show that there is an index $\jstar < k$ such that $\ell(a_\jstar) > \jstar$ and  $\{a_{\jstar + 1}, \ldots, a_k\}$ is settled during \PassBundle $\tau$.
	Towards that goal, suppose that $\jstar$ is the \emph{largest} index such that $\ell(a_\jstar) > \jstar$.
	If such an index does not exist, let $\jstar = 0$.
	If $\jstar = k - 1$, we have found our desired index.
	
	So, assume that $\jstar < k - 1$.
	Suppose that for some $s$, where $\jstar < s \leq k$, $\{a_s, \ldots, a_k\}$ is settled at $\PassBundle$ $\tau$.
	We want to show that then, either $s = \jstar + 1$ or also $\{a_{s - 1}, \ldots, a_k\}$ is settled during $\PassBundle$ $\tau$.
	Note that $s = \jstar + 1$ concludes our base case, so assume that $s > \jstar + 1$, i.e., $\jstar < s - 1$.
	We have the following cases (notice that $\ell(a_{s - 1}) \leq s - 1$ holds by our choice of $\jstar$).
	\begin{enumerate}
		\item[] {\bf Neither arc $a_{s - 1}$ nor $\leftarc{a_{s - 1}}$ is active and $\ell(a_{s - 1}) \leq s - 1$.}  Recall that the length of the active path $P_{\gamma}$ to $\leftarc{a_{k + 1}}$ is at most $k$ and $a_{s - 1} \notin P_{\gamma}$.
		To apply \Cref{lemma: tagging}:
		\begin{itemize}
			\item Observe that $\{a_{s}, \ldots, a_k\}$ is settled, and hence $s$ in this proof plays the role of $f$ in the statement of \Cref{lemma: tagging}.
			\item We now want to upper-bound $j$ from the statement of \Cref{lemma: tagging}. To that end, let $j < s$ be the largest index such that either $a_{j}$ or $\leftarc{a_{j}}$ is active and $\ell(a_{j}) \le j$, or $\ell(a_{j}) > j$. Note that in the former case we have $j < s - 1$ as neither arc $a_{s - 1}$ nor $\leftarc{a_{s - 1}}$ is active. In the later case we also have $j < s - 1$ as $\jstar < s - 1$. Therefore, it holds that $j < s - 1$.
		\end{itemize}
		Hence, by \Cref{lemma: tagging}, we have that $\{a_{s - 1}\} \cup \{a_{s}, \ldots, a_k\}$ is settled in \PassBundle $\tau$ (not necessarily its beginning). 
		\\\\
		\emph{Remark:} We apply \Cref{lemma: tagging} for the beginning of $\PassBundle$ $\tau + 1$. Moreover, we cannot apply \Cref{lemma: tagging} for the beginning of $\PassBundle$ $\tau$ as only at some point (and after its very beginning) of $\PassBundle$ $\tau$ we have that $\ell(a_k)$ becomes at most $k$.
		Nevertheless, note that ``the beginning of $\PassBundle$ $\htau$'' refers to the time in $\PassBundle$ $\htau$ before any edge on the stream has been seen. Hence, if an edge is settled at the beginning of $\PassBundle$ $\tau + 1$ then it was also settled in $\PassBundle$ $\tau$, though not necessarily at the beginning of $\PassBundle$ $\tau$, which justifies the terminology we use in this proof.
		
		\item[] {\bf Arc $\leftarc{a_{s - 1}}$ or $a_{s - 1}$  is active and $\ell(a_{s - 1}) \leq s - 1$.} 
		Our goal now it so apply \Cref{lemma: upd-active}. To that end, the active path $P_\gamma$ from this proof that has length at most $k$ to $a_k$ serves as the path $P_\alpha$ to $a_x$ in \Cref{lemma: upd-active}. 
		Furthermore, we assumed that \cref{invariant: jump} holds until \PassBundle $\tau$ and from the lemma statement we have that $\ell(a_{k + 1}) > k + 1$ (at least)  until $a_k$ is backtracked and becomes inactive.
		Therefore, by \Cref{lemma: upd-active} where the index $s$ in this proof corresponds to the index $f$ in \Cref{lemma: upd-active}, we have that $\{a_{s - 1}\} \cup \{a_s, \ldots, a_k\}$ is settled in \PassBundle $\tau$.	
	\end{enumerate}
	
	\paragraph{Inductive Step:}
	Suppose then that the inductive claim holds for some $s$, where $f \leq s \leq k$, in $\PassBundle$ $\tau$.
	In other words, $s$ is the minimum index such that $\{a_s, \ldots, a_k\}$ is settled and $\ell(a_{s - 1}) > s - 1$.
	Then, by our lemma statement, we have that there is a \PassBundle $\tau' \leq \tau^*$ such that $\tau' \ge \tau$ and an active path $P_{\gamma}(\gamma, b_1, \ldots, b_{k'})$ sets $\ell(a_{s - 1}) \leq s - 1$ where $\gamma$ is a free vertex.
	From the base case, we have that $\ell(\leftarc{a_{k + 1}}) < k$ by $\PassBundle$ $\tau'$.
	Next, we want to find an index $j < k$ such that $\ell(a_j) > j$ and $\{a_{j + 1}, \ldots, a_k\}$ is settled in \PassBundle $\tau'$.
	The rest of the proof of the inductive step is analogous to the base case. 
	Then, we can repeat the inductive argument until we obtain that $\{a_f, \ldots, a_k\}$ is settled in some $\PassBundle$ $\htau \leq \tau^*$.
\end{proof}

\begin{proof}[Proof of \cref{lemma: jump-proof}]
	The proof is an application of \Cref{lemma: settle before passive}.
	Now, we will show that the prerequisites of the lemma are satisfied.
	Right before $\PassBundle$ $\tau$, it must be the case that no arc in $(b_1, \ldots, b_h)$ is active.
	Furthermore, we have that $\ell(b_j) \leq k + j$ for all $1 \leq j \leq h$ where, in particular, $\ell(b_h) \leq k + h$.
	It is also the case that $\ell(b_{h + 1}) > k + h + 1$.
	Therefore, we satisfy the prerequisites and by \Cref{lemma: settle before passive}, we have that $\{ b_1, \ldots, b_h \}$ is settled at the beginning of $\PassBundle$ $\tau$.	
\end{proof}

\subsection{Proof of \cref{lemma: augment}}
\lemmaaugment*
\begin{proof}
The goal is to prove the following statement inductively, under the assumption that no vertex in $P$ belongs to an augmentation found in this phase so far, as otherwise $P$ would not exist. (Recall that augmentations that our algorithm finds are removed from the current graph in this phase.)

\vspace{3pt}
\begin{minipage}{0.95 \linewidth}
\begin{framed}
	Let $i \leq k$.
	For all $j \leq i$, we have $\ell(a_j) \leq j$, or $\ell(a_j) > j$ and at least one vertex in $(\alpha, a_1, \ldots, a_j) \subset \Path{aug}$ is active.
\end{framed}
\end{minipage}
\vspace{3pt}

\noindent For the proof, recall from \Cref{sec:overview} that the arcs $a_1, \ldots, a_k$ in $P$ correspond to the matched edges of the alternating path.
Furthermore, in our algorithm, only matched arcs (the arcs corresponding to matched edges) have a label.
Consider the for-loop over the input stream in \AlgExtendStructures, \Cref{algo-extend-structures}, and suppose $g = (u, v)$.
For the sake of brevity in this proof, we say that $\alpha$ \emph{considers} the matched arc $a^* = (v, x)$ when the path $P_u$ corresponds to a trivial path or an alternating path that starts at $\alpha$ and ends in a matched arc $(y, u)$.

\paragraph{Base Case:}
In the beginning of the execution, the free vertex $\alpha$ is active with an active path of $(\alpha)$.
Let $g = (u, v)$ be the non-matched arc between $\alpha$ and $a_1$.
Before $\alpha$ becomes inactive, there is a pass where the active path of $\alpha$ is $(\alpha)$ and $\alpha$ considers the arc $\astar = a_1$.
In this case, we have that $P_u = (\alpha)$, $k = 0$, and $h = 0$.
If the label $\ell(a_1) > 1 = k + h + 1$, then $\alpha$ extends its path to $(\alpha, a_1)$ and $\ell(a_1)$ will be updated to $\ell(a_1) \coloneqq k + h + 1 = 1$.
Notice that this update happens whether an overtake takes place or not.
This concludes the base case.

\paragraph{Inductive Step:}
Assume that the inductive hypothesis holds for some $i$.
We will split the analysis into two cases and for both cases, we suppose for a contradiction that there is a \PassBundle such that $\ell(a_{i + 1}) > i + 1$ and no vertex is active in $(\alpha, a_1, \ldots, a_i)$.

First, consider the case that after each arc $a_j$, for $1 \leq j \leq i$, has become inactive and $\alpha$ is active with the trivial path $(\alpha)$.
Now, we apply \Cref{lemma: settle before passive} on the whole path $(a_1, \ldots, a_{i + 1})$ obtaining that $\{a_1, \ldots, a_i\}$ is settled.
Hence, it must be the case that eventually, $\alpha$ considers to extend its active path to $a_{i + 1}$ along the settled path and sets the label $\ell(a_{i + 1}) \leq i + 1$.
Notice that no other free vertex can obtain this settled path $(a_1, \ldots, a_{i})$ without finding an augmentation to $\alpha$.

Second, consider the case that $\alpha$ becomes inactive before all arcs in $(a_1, \ldots, a_i)$ are inactive.
Also in this case, we want to show that $\{a_1, \ldots, a_i\}$ is settled before the last arc is backtracked.
Let either $a_j \in \structure_\gamma$ or $\leftarc{a_j} \in \structure_\gamma$ be the last active arc in $\{a_1, \leftarc{a_1}, \ldots, a_i, \leftarc{a_i} \}$ and consider some \PassBundle $\tau$ where $a_j$ or $\leftarc{a_j}$ is the only active in $\{ a_1, \leftarc{a_1}, \ldots, a_i, \leftarc{a_i} \}$.
We first apply \Cref{lemma: settle before passive} to obtain that $\{a_{j + 1}, \ldots, a_i\}$ is settled at the beginning of \PassBundle $\tau$.
Notice that the prerequisites of \Cref{lemma: settle before passive} are satisfied since $a_{j}$ or $\leftarc{a_{j}}$ is the last active arc and hence, no arc in $\{a_{j + 1}, \leftarc{a_{j + 1}}, \ldots, a_i, \leftarc{a_i} \}$ can be active.
Then, we want to show that also $\{a_j\} \cup \{a_{j + 1}, \ldots, a_i\}$ is settled.
Here, we have several subcases.
\begin{itemize}
	\item Suppose that $\leftarc{a_j}$ is active. Then, $\leftarc{a_j}$ is reachable and hence, $\{a_j\} \cup \{a_{j + 1}, \ldots, a_i\}$ is settled.
	\item Suppose that $a_j$ is active and let $P_\gamma(a_j)$ be the corresponding active path. If $|P_\gamma(a_j)| > j$, then by \Cref{invariant: jump}, we have that $a_j$ is reachable and hence, $\{a_j\} \cup \{a_{j + 1}, \ldots, a_i\}$ is settled.
	\item Suppose that $a_j$ is active and let $P_\gamma(a_j)$ be the corresponding active path. Suppose that $|P_\gamma(a_j)| \leq j$. Then, by \Cref{obs: no new active}, we have that $\ell(a_{i + 1}) \leq i + 1$ before all arcs $\{ a_j, \ldots, a_i \}$ are backtracked. This contradicts with our assumption that $\ell(a_{i + 1}) > i + 1$ until no vertex is active in $(\alpha, a_1, \ldots, a_i)$.
\end{itemize}
Therefore, we have that $\{a_j, \ldots, a_i\}$ is settled at the beginning of \PassBundle $\tau$.
As the last step, we want to apply \Cref{lemma: tagging} to obtain that $\{a_1, \ldots, a_{j - 1}\} \cup \{a_j, \ldots, a_i\}$ is settled in \PassBundle $\tau$.
Towards that end, notice that since $\ell(a_{i}) \leq i$, there must have been a \PassBundle $\tau' < \tau$, such that $a_i$ was active with an active path of length at most $i$.
Therefore, by \Cref{lemma: lastleft}, it must be the case that $\ell(\leftarc{a_{i + 1}}) < i$ in \PassBundle $\tau'$.
Furthermore, since $a_j$ or $\leftarc{a_j}$ is the only active arc in \PassBundle $\tau$, we have that there is no active arc $a_h$ or $\leftarc{a_h}$ for $h < j$ in \PassBundle $\tau$.
We can now apply \Cref{lemma: tagging} to get that $\{a_1, \ldots, a_{j - 1}\} \cup \{a_j, \ldots, a_i\}$ is settled.
Therefore, we have an augmenting path from $\gamma$ to $\alpha$, which will be detected in \Cref{line:ES-check-alpha-beta-merge} of \Cref{algo-extend-structures}.
This implies that the augmenting path $\alpha - \beta$ will be removed from the graph in \PassBundle $\tau$.
\end{proof}


\section{Analysis of Approximation and Pass Complexity}\label{sec:pass}

There are several values that we tie together to provide the approximation analysis. One of them is $\sizelimit$. Another one is the number of phases we have, let that value be $\phases$. In phase $i$ the algorithm makes $\tau_i$ passes; $\tau_i$ is computed implicitly in the following way. A phase is executed as long as there are at least $\delta |M|$ active free vertices. To upper-bound $\tau_i$ from this condition, we use the fact that there are at most $2 |M| / \eps$ different label updates the matched arcs can have in one phase. Then, except for at most $n / \sizelimit$ many active free vertices, each active free node is either reducing a label or backtracking from an arc whose label it has reduced previously.

Then, we show how much improvement in matching is made in a phase. Performing multiple phases gives better approximation, at least as long as the matching we have so far is not a $(1+\eps)$-approximate one.
The number of phases in terms of structure-size and the number of active nodes is analyzed in \cref{lemma:number-of-phases}.

Finally, when the algorithm finds an augmentation between $\alpha$ and $\beta$, it removes $\structure_\alpha$ and $\structure_\beta$. We want to say that $\structure_\alpha$ and $\structure_\beta$ are not too big; the larger a structure is the more potential augmenting paths its removal destroys. Unfortunately, the sizes of $\structure_\alpha$ and $\structure_\beta$ can exceed $\sizelimit$. To see that, notice that although we stop extending a structure $S$ once its size exceeds $\sizelimit$, another structure $S'$ can overtake $S$. This can potentially make the size of $S'$ almost $2 \cdot \sizelimit$. Then, another structure $S''$ can overtake $S'$, making $S''$ of size almost $3 \cdot \sizelimit$, and so on. By \cref{lemma:structuresize}, we show that this type of growth cannot proceed for ``too long''.

In this section we analyze the number of passes required for our algorithm to output a $(1+\eps)$-approximate maximum matching. This analysis at the same time upper-bounds the number $T$ of phases, as well as the maximum structure size $\Ssize$. Recall from \cref{subsec:algorithm} that we set $\maxtau\eqdef 1/\eps^6$ and $\sizelimit\eqdef 1/\eps^4$. 

For our analysis, we introduce the following function.
\begin{definition}\label{definition:delta(k)}
	We define
		$\delta(k) \eqdef \frac{1}{2 k (k + 2)}$.
\end{definition}

We first upper bound the number $T$ of phases our algorithm requires, given that there exists an upper bound $\Ssize$ on the size of structures and the number of active nodes at the end of a phase. 
\begin{lemma}[Upper Bound on Number of Phases]\label{lemma:number-of-phases}
	 Let $\Ssize$ be an upper bound on the structure size and $h(\eps) |M|$ be the maximum number of active nodes at the end of a phase. Then, we need at most $T\leq \frac{1 + 2 \Ssize}{2 \delta(2/\eps) - h(\eps) / \eps}$ phases to get a $(1 + \eps)$-approximate maximum matching.
\end{lemma}
The proof of this lemma is given in \cref{sec:upper-bound-number-phases}.
Then, in \cref{sec:upper-bound-structure-size}, we upper bound the size of a structure at any point during an execution of the algorithm.
From this bound, we also almost immediately get a bound on the memory requirement of our algorithm.
\begin{lemma}[Upper Bound on Structure Size]\label{lemma:structuresize}
	Let $\tau$ be the number of $\PassBundle$s our algorithm makes. Let $S_{\alpha}$ be a structure of $\alpha$ at any point of a phase. Then, the number of vertices in $S_\alpha$ is at most $\tau (1/\eps + 1) \cdot \sizelimit$.
\end{lemma}
We now justify why \cref{lemma:structuresize} is needed, and why an upper-bound on structure size is not $O(\sizelimit)$ trivially.
In the overview of our approach, we introduced the notion of ``putting a structure on hold'' (see \cref{section:on-hold}), which is applied when the number of vertices in a structure is above $\sizelimit$. This is done to obtain a $(1+\eps)$-approximation in the desired number of passes, as discussed next. Assume that we do not impose any limit on the size of structures, and that a structure $S_\alpha$ has size $\exp(1/\eps)$. When our algorithm finds an augmentation containing $\alpha$, it will remove the entire structure $S_\alpha$ from further consideration in the current phase. On the other hand, removal of $S_\alpha$ can remove many potential augmentations not containing $\alpha$. If this number of removed augmentations is $\exp(1/\eps)$, it could mean that the algorithm made little progress in the current phase, i.e., the algorithm potentially found only a $1 / \exp(1/\eps)$-fraction of a maximum matching. Having structures ``being on hold'' aims to avoid this situation. 

Nevertheless, observe that the ``on hold'' operation does not imply that each structure will contain $O(\sizelimit)$ vertices. To see that, consider the following sequence of operations, which can even happen within \emph{the same} phase. $S_{\alpha_1}$ grows to $\sizelimit$ vertices and is set on hold. The next edge arrives and $S_{\alpha_2}$, whose current size is say $\sizelimit - 1$, overtakes most of the structure of $S_{\alpha_1}$ (while $S_{\alpha_1}$ is still on hold). At that point, the new size of $S_{\alpha_2}$ can be almost $\size(S_{\alpha_1}) + \sizelimit - 1 = 2\cdot \sizelimit - 1$. Then, if another edge arrives and $S_{\alpha_3}$ overtakes the majority of $S_{\alpha_2}$, the size of $S_{\alpha_3}$ can be as large as $3\cdot \sizelimit$, and so on. We will now show that, despite this scenario, the size of any structure is upper-bounded by $O(\poly \eps)$.
\begin{observation}\label{obs: memory}
	Our algorithm requires $n \cdot \poly(1/\eps)$ words of memory, where a word corresponds to the maximum amount of memory required to store an edge.
\end{observation}
\begin{proof}
	Notice that in any given time, there are $O(n)$ active free vertices.
	By combining \Cref{lemma:structuresize} and \Cref{lemma: runtime}, the structure $\structure_\alpha$ maintained by any free vertex $\alpha$ contains $\poly 1/\eps$ vertices. In addition, our algorithm maintains all matched edges, in total $O(n)$ of them, and unmatched edges with their endpoints belonging to the \emph{same} structure (see the description of \AlgAugmentStructures). Note that the number of such unmatched edges in $\structure_\alpha$ is upper-bounded by the square of the number of vertices in $\structure_\alpha$. Since there are $O(n)$ free nodes, it implies that our algorithm maintains $O(n \cdot \poly(1/\eps))$ unmatched edges.
	Therefore, the total memory required is $n \cdot \poly(1/\eps)$.
\end{proof}

Finally, in \cref{sec:upper-bound-active-free-nodes}, we upper bound the number of active nodes at the end of each phase.
\begin{lemma}[Upper Bound on Number of Active Nodes]\label{lemma:activenodes}
Let $\tau$ be the number of $\PassBundle$s executed in a phase and $M$ be the matching at the begining of the phase. Then, at the end of that phase there are at most $h(\eps) |M|$ active free nodes, where $h(\eps)$ is upper-bounded as
	\[
		h(\eps) \le \frac{4 + 2/\eps}{\eps \cdot \tau} + \frac{2}{\sizelimit}.
	\]
\end{lemma}

Combining these three bounds, we easily get the following bound on the number of passes. 

\begin{lemma}\label{lemma: runtime}
The number of passes of our algorithm is $\poly(1/\eps)$.
\end{lemma}
\begin{proof}
The total number of passes is at most $3T\maxtau$, since each of at most $T$ phases has at most $\maxtau$ $\PassBundle$s, each of which needs 3 passes, one for each of the three operations.

Combining this with \cref{lemma:number-of-phases}, \cref{lemma:structuresize}, and \cref{lemma:activenodes}, which give us upper bounds on $T$, the structure size $\Ssize$, and $h(\eps)$, respectively, we can upper bound the pass complexity by
\begin{equation*}
3T\maxtau\leq 
\frac{3 \maxtau+ 6 \maxtau^2(1/\eps + 1) \cdot \sizelimit}{2 \delta(2/\eps) - \frac{\frac{4 + 2/\eps}{\eps \cdot \maxtau} + \frac{2}{\sizelimit}}{ \eps}}=
\Theta\left(\frac{\maxtau+\maxtau^2\cdot\sizelimit/\eps}{\eps^2-1/(\eps^3\maxtau)-1/(\eps \cdot \sizelimit)}\right).
\end{equation*}
Recalling from \cref{subsec:algorithm} that $\maxtau\eqdef 1/\eps^6$ and  $\sizelimit\eqdef 1/\eps^4$, we have $1/(\eps^3\maxtau) = \eps^3$ and $1/(\eps \cdot \sizelimit) = \eps^3$. Thus, the denominator is lower bounded by $\eps^2-2\eps^3$, which is in $\Omega(\eps^2)$ for $\eps <1/2$. Taken together, we can bound the pass complexity by $O(1/\eps^{19})$, which concludes the proof. 
\end{proof}

\main*
\begin{proof}
	By \Cref{lemma:number-of-phases}, it follows that \AlgPhase yields an $(1 + \eps)$-approximation to Maximum Matching.
	By \Cref{lemma: runtime}, we have that the pass complexity is $\poly(1/\eps)$.
	The memory requirement is satisfied by \Cref{obs: memory}, proving the theorem.
\end{proof}

\subsection{Upper bound on the Number $T$ of Phases: Proof of \cref{lemma:number-of-phases}}
\label{sec:upper-bound-number-phases}

The rough idea of the proof is as follows. First, we observe that having a small number of short augmenting paths is a certificate for a good approximation, as formalized in \cref{lemma:certificate-for-maximality}. We use this observation to show in \cref{lemma:LB-progress-per-phase} that whenever we do not have a good approximation yet, we must find many augmenting paths, that is, we can lower bound the progress we are making in every phase. This then can be used directly to upper bound the number of phases, leading to the proof of \cref{lemma:number-of-phases}.  

\begin{definition}[Inclusion-Maximal Matching]
	Let $M \subseteq E$ be a matching of graph $G = (V, E)$.
	Then, we say that $M$ is inclusion-maximal, or just maximal, if there is no matching $M^*$, such that $M \subsetneq M^*$.
	In other words, no edge $e \in E$ can be added to $M$ such that $M + \{ e \}$ is a matching.
\end{definition}

We first introduce the notion of \emph{$k$-alternating-path} to get a hold on the size of augmenting paths. 
\begin{definition}[$k$-alternating-path]
	We say that $P$ is a \emph{$k$-alternating-path} if $P$ is an alternating path that starts with an unmatched edge and has length \emph{at most} $2k + 1$.
\end{definition}

The following lemma can be found in several papers in different forms. The one we state below is from \cite[Lemma~2]{Eggert2012}. This lemma states that a small number of short augmenting paths is a certificate for a good approximation. 
\begin{lemma}\label{lemma:certificate-for-maximality}
	Let $M$ be an inclusion-maximal matching. Let $\cY$ be an inclusion-maximal set of disjoint $M$-augmenting $k$-alternating-paths such that $|\cY| \le 2 \delta(k) |M|$ (see \cref{definition:delta(k)} for a definition of $\delta(k)$). Then, $M$ is a $(1 + 2/k)$-approximation of a maximum matching.
\end{lemma}
\begin{proof}
	We refer a reader to \cite{Eggert2012} for a proof of this claim.
\end{proof}

Next, we use \cref{lemma:certificate-for-maximality} to show that in every phase we must make significant progress, meaning that the matching size increases by a factor. 
\begin{lemma}\label{lemma:LB-progress-per-phase}
	Let $\Ssize$ be an upper bound on the structure size and $h(\eps) |M|$ be the maximum number of active nodes at the end of a phase. Fix a phase. Let $k \ge 3$ be an integer parameter. Let $M$ be the matching at the beginning of that phase, and assume that $M$ is not a $(1 + 2/k)$ approximation of a maximum matching. Then, by the end of the same phase the matching size will increase by factor
	\[
		1 + \frac{2 \delta(k) - h(\eps) / \eps}{1 + 2 \Ssize}.
	\]
	(See \cref{definition:delta(k)} for a definition of $\delta(k)$.)
\end{lemma}
\begin{proof}
	Let $\cYstar$ be the \emph{maximum} number of disjoint $M$-augmenting $k$-alternating-paths. From \cref{lemma:certificate-for-maximality} we have
	\begin{equation}\label{eq:cYstar-lower-bound}
		\cYstar > 2 \delta(k) |M|.
	\end{equation}
	
	Let $h(\eps)$ be the number of active nodes at the end of the phase. By \cref{lemma: augment} we have that each augmenting path remaining at the end of the phase intersects an active path.
	Since there are at most $h(\eps) |M|$ active paths and each active path has length at most $\maxlen$, the maximum number of disjoint augmentations appearing at the end of the phase is upper-bounded by $h(\eps) |M| / \eps$.
	
	Now, let $X$ be the number of augmentations found during the phase. Each augmentation found between free nodes $\alpha$ and $\beta$ results in removing the $\alpha$- and $\beta$-structure. Hence, each found augmentation results in potentially ``removing'' $2 \Ssize$ other augmentations. This implies that in total during the phase there are $X$ augmentations found, at most $h(\eps) |M| / \eps$ not found, and at most $2 X \Ssize$ augmentations removed (but those augmentation will re-appear in the next phase). Moreover, it holds
	\[
		X + h(\eps) |M| / \eps + 2 X \Ssize \ge \cYstar \stackrel{\cref{eq:cYstar-lower-bound}}{>} 2 \delta(k) |M|.
	\]
	This implies
	\[
		X > \frac{(2 \delta(k) - h(\eps) / \eps) \cdot |M|}{1 + 2 \Ssize}.
	\]
	Since each augmentation increases the matching size by exactly $1$, the claim follows.
\end{proof}

	From \cref{lemma:LB-progress-per-phase} we have that, as long as the current matching is not $(1+2/k)$-approximate, the matching size increases by a factor of at least $1 + \frac{2 \delta(k) - h(\eps) / \eps}{1 + 2 \Ssize}$. Let $k = 2/\eps$. Assuming that $2 \delta(2/\eps) - h(\eps) / \eps > 0$, within $\frac{1 + 2 \Ssize}{2 \delta(2/\eps) - h(\eps) / \eps}$ phases the current matching would either become a $(1 + \eps)$-approximate maximum one, or its size would double. Since we started with a $2$-approximate maximum matching, that number of phases suffices to construct a $(1+\eps)$-approximate maximum matching.

\subsection{Upper bound $\Ssize$ on Structure Size: Proof of \cref{lemma:structuresize}}
\label{sec:upper-bound-structure-size}

	We start by making the following observation. When a free node $\alpha$ overtakes a set of new vertices $X \notin S_\alpha$ via an edge $\astar$, leading to the structure $S_\alpha'$, the following two things happen: the label of $\astar$ gets reduced by at least $1$; and each alternating path within $S_\alpha'$ from $\alpha$ to each matched arc $a \in X$ goes through $\astar$ (at least until $\AlgAugmentStructures$ is executed in the same $\PassBundle$, \cref{line:AlgAugmentStructures} of \cref{alg:phase}).
	
	In this proof we will also need the following simple structural property.
	\begin{observation}\label{obs:overtaking-reduces-head-label}
		Let $P = (a_1, \ldots, a_k)$ be an active path and $a_i \in P$ a matched arc such that $\ell(a_i) = i$. Then, if an arc $a_j \in P$ for $j < i$ gets overtaken, this overtake leads to reducing a label of $a_i$. Moreover, let $Q = (b_1, \ldots, b_h, a_j, \ldots, a_i, \ldots, a_k)$ be the overtaken active path. Then, the new label of $a_i$ equals $h + i - j + 1$.
	\end{observation}
	\begin{proof}
		Before the overtake, we have $\ell(a_j) \le j$. Therefore, $h < \ell(a_j) - 1 \le j - 1$. This implies that $h + i - j + 1 < i$, and hence the new label of $a_i$ is set to the length of the prefix of $Q$ ending at $a_i$. This length equals $h + i - j + 1$ as desired.
	\end{proof}
		
	Consider a $\PassBundle$ $p$ and the execution of $\AlgExtendStructures$. By design of our algorithm, if an arc $a$ is active at the beginning of $p$, it remains active until the end of the execution of $\AlgExtendStructures$. Only possibly at the end of $\AlgExtendStructures$ the algorithm backtracks on $a$, casting it inactive.
	
	Let $A$ be an upper bound on the number of vertices in each structure at the beginning of $p$. For each free node $\alpha$, let $\Vbase(\alpha)$ be the set of vertices owned by $\alpha$ at the beginning of $p$, and let $\Vnew(\alpha)$ be the set of vertices added to $\alpha$'s structure after overtaking or an extension.
	
	Now we would like to upper bound $|\Vnew(\alpha)|$. Let $\alpha_0 = \alpha$ and, for arbitrarily large $k$, let $\alpha_k, \ldots, \alpha_0$ be the free nodes that in that order participated in a sequence of overtakes leading to $\Vnew(\alpha_0)$, where $\alpha_{i - 1}$ overtook a part of the structure of $\alpha_i$, for each $i = 1, \ldots, k$.
	Moreover, assume that when $\alpha_{i - 1}$ overtakes a part of $\alpha_{i}$, it also overtakes a part of $\Vbase(\alpha_i)$. Otherwise, if $\alpha_{i - 1}$ overtakes only a part of $\Vnew(\alpha_i) \subseteq \Vnew(\alpha_{i + 1}) \cup \Vbase(\alpha_{i + 1})$, we could as well think that $\alpha_{i - 1}$ overtook directly a part owned by $\alpha_{i + 1}$ and hence avoid having $\alpha_i$ in the chain of overtakes.
	
	This last property is especially important for us as we want to say that each time an overtake is performed, then a label of a specific active arc is reduced by at least one. 
	In more detail, let $\astar_i$ be the head of the active path originating at $\alpha_i$ \emph{before} $\alpha_i$ performs the overtake.
Since $\alpha_{i - 1}$ overtakes a part of $\Vbase(\alpha_i)$, we have that if $\alpha_{i - 1}$ overtakes a part of $\Vnew(\alpha_i)$, then it overtakes $\astar_i$ as well. As a reminder, this is the case since: (1) there is an alternating path to each matched arc in $\Vnew(\alpha_i)$ that goes through $\astar_i$; and (2) in order to overtake a part of $\Vbase(\alpha_i)$, $\alpha_{i - 1}$ has to overtake an arc in $\Vbase(\alpha_i)$ (there is an alternating path from $\alpha_i$ to each matched arc in $\Vbase(\alpha_i)$ that does not intersect $\Vnew(\alpha_i)$).
After this overtake, the active path originating at $\alpha_{i-1}$ contains both $\astar_{i-1}$ and $\astar_i$.

	This implies that $\alpha_0$ can overtake parts of $\Vnew(\alpha_j)$ for $j \le \maxlen$, but not for any $j > \maxlen$. More formally, assume that $\alpha_0$ overtakes parts of $\Vnew(\alpha_{\jbig})$, where $\jbig > \maxlen$ is the largest in this chain of overtakes for which $\Vnew(\alpha_{\jbig}) \neq \emptyset$. Let $\afirst$ be the very first arc that was overtaken by $\alpha_{\jbig}$ -- hence, the label of $\afirst$ at that moment was equal the length of that active path. Then, since each overtake in the chain reduces the label of $\afirst$, as shown by \cref{obs:overtaking-reduces-head-label}, $\ell(\afirst)$ would become negative, which would contradict properties of our algorithm.
	
	Observe that $\alpha_i$ overtakes a part of $S_{\alpha_{i + 1}}$ via $\astar_i$. Also, observe that $S_{\alpha_j}$ for all $j < k$ is not on hold, and hence $|\Vbase(\alpha_j)| \le \sizelimit$. However, it might be the case that $|\Vbase(\alpha_k)| = A$.

	Each vertex from $\Vbase(\alpha_j)$, for $j > 0$, that is overtaken by $\alpha_{j - 1}$, belongs to $\Vnew(\alpha_{j - 1})$. Also, observe that $\Vnew(\alpha_{i - 1}) \subseteq \Vnew(\alpha_{i}) \cup \Vbase(\alpha_{i})$, and thus $\Vnew(\alpha_0) \subseteq \bigcup_{j = 1}^{1 + 1/\eps} \Vbase(\alpha_j)$. Hence, we conclude that $|\Vbase(\alpha) \cup \Vnew(\alpha)| \le (\maxlen + 1) \cdot \sizelimit + A$. If the number of $\PassBundle$s made in a phase is $\tau$, then by an induction the size of a structure is upper bounded by $\tau (\maxlen + 1) \sizelimit$ at any point of the phase.

\subsection{Upper Bound on Number of Active Free Nodes: Proof of \cref{lemma:activenodes}}
\label{sec:upper-bound-active-free-nodes}
Let $\tau$ be the number of $\PassBundle$s our algorithm makes and let $h(\eps) |M|$ be the number of active nodes at the end of a phase. We upper-bound $h(\eps)$ as a function of $\tau$. The high-level idea of our approach is as follows.
During a $\PassBundle$, each active free node is either (1) on hold, or (2) reduces labels of at least one arc~\footnote{Labels of multiple arcs can be reduced during an overtake (see \cref{operation:overtaking}).}, or (3) backtracks, or (4) waits for another $\PassBundle$ in case another free node overtook its active arc. For each step (2) or (4), a label of at least one matched arc reduces by at least one.
For each step (3), a node backtracks on exactly one matched arc $a$. The algorithm either reduced the label of $a$ previously (and now is backtracking), or jumped over $a$ to reduce the label of some other arc.
Nevertheless, the algorithm jumps over at most $\maxlen - 1$ arcs to reduce a label of another arc. So, we can use the fact that total number of label-changes is at most $2 \maxlen \cdot |M|$ (the factor $2$ is coming from the fact that each of the $|M|$ matched edges is treated as two arcs) to upper-bound the number of steps (2)-(4) by $(1 + 1 + \maxlen) \cdot 2 \maxlen \cdot |M|$.
However, this argument does not apply to the nodes on hold. Luckily, we can upper-bound the number of such nodes by using the value $\sizelimit$.

Recall that $\sizelimit$ is a lower-bound on the number of vertices (excluding the corresponding free node) a structure on hold has. We now argue that such a structure contains at least $\sizelimit / 2$ matched edges. Observe that by a construction, each vertex in a structure $S_\alpha$, except $\alpha$ itself, is incident to a matched edge. In addition, no two matched edges share a vertex. This implies that $S_\alpha$ contains at least $\sizelimit / 2$ matched edges. Hence, the number of free nodes that are on hold at the beginning of any $\PassBundle$ is at most $|M| / (\sizelimit / 2)$.
So, during each $\PassBundle$, at least $h(\eps) |M| - 2 |M| / \sizelimit$ many active free nodes perform one of the steps (2)-(4).
Since the total number of these steps is $(2 + \maxlen) \cdot 2 \maxlen \cdot |M|$, by replacing $\maxlen$ with $1/\eps$ we get that
\[
	\rb{h(\eps) |M| - 2 |M| / \sizelimit} \tau \le 4|M| / \eps + 2|M| / \eps^2,
\]
implying
\[
	\rb{h(\eps) - 2 / \sizelimit} \tau \le 4 / \eps + 2/\eps^2.
\]
This now results in
\[
	h(\eps) \le \frac{4 + 2/\eps}{\eps \cdot \tau} + \frac{2}{\sizelimit}
\]
and concludes the proof.


\section{A General Framework}
\label{section:general-framework}

\newcommand{\Taggregate}{T_{\rm{aggregate}}}
\newcommand{\Tedgescan}{T_{\rm{edge-traverse}}}
\newcommand{\Tpropagate}{T_{\rm{propagate}}}
\newcommand{\Aaggregate}{A_{\rm{aggregate}}}
\newcommand{\Aedgescan}{A_{\rm{edge-traverse}}}
\newcommand{\Apropagate}{A_{\rm{propagate}}}

\newcommand{\Cmatching}{C_{\rm{matching}}}
\newcommand{\cMM}{c_{\rm{MM}}}

In this section we describe how to generalize our algorithm to other computation models. We begin by describing what procedures our framework requires the access to. The input contains a graph $G$ and an approximation parameter $\eps$.
\begin{enumerate}[(i)]
	\item\label{item:matching} Let $\Amatching$ be an algorithm that computes a $\cMM$-approximate maximum matching in time $\Tmatching(G, \eps)$.
	\item\label{item:edge-scan} Let $\Aedgescan$ be an algorithm that traverses all the edges of a graph in time $\Tedgescan(G, \eps)$.
	\item\label{item:propagate} Let $G$ be partitioned into (disjoint) connected components (these components will be structures in our case) each of size $\poly 1/\eps$. Let $\Apropagate$ be a method that propagates an update of size $\poly 1/\eps$ across the entire component. Let $\Tpropagate(G, \eps)$ be an upper-bound of the time $\Apropagate$ needs for this operation. Moreover, let $\Apropagate$ be such that these updates can be carried over in parallel across all the components.
	\item\label{item:aggregate} Let $G$ be partitioned into (disjoint) connected components (these components will be structures in our case) each of size $\poly 1/\eps$. Let $\Aaggregate$ be a method that aggregates an the entire component to a single working unit. We assume that each working unit (e.g., a machine, or local memory) has enough memory to store at least two components. Let $\Taggregate(G, \eps)$ be an upper-bound of the time $\Aaggregate$ needs for this operation. Moreover, let $\Aaggregate$ be such that these aggregations can be carried over in parallel across all the components.
	\item\label{item:storage} Let there be a storage, called \storage, containing $\Omega(n \poly 1/\eps)$ words available to the algorithm. Let \storage be such that each of its memory cells has a size of at least $\Omega(\poly 1 / \eps)$ words, e.g., a vertex can store $\Omega(\poly 1 / \eps)$ words of data in addition to its id.
	
		\storage can also be distributed in a sense that each matched edge and each vertex needs access to $\poly 1/\eps$ memory cells, but a vertex/edge does not need a direct access to memory cells of other vertices/edges.
\end{enumerate}
In a distributed/parallel setting, the aforementioned ``time'' should be understood as the number of rounds. All the times listed above are a function of $G$ and $\eps$, but for the sake of brevity we drop these parameters in the rest of this section.
\begin{theorem}
\label{theorem:framework}
	Let $G$ be a graph on $n$ vertices and let $\eps \in (0, 1/2)$ be a parameter. Given the access to \storage as described in \eqref{item:storage}, and to algorithms described in \eqref{item:matching} for $\cMM \in O(1)$, \eqref{item:edge-scan}, \eqref{item:propagate} and \eqref{item:aggregate}, there exists an algorithm that outputs a $(1 + 2 \eps)$-approximate maximum matching in $G$ in time $O((\Tmatching + \Tedgescan + \Tpropagate + \Taggregate) \cdot \poly 1/\eps)$.
\end{theorem}


We prove \cref{theorem:framework} by constructing a framework that simulates our streaming algorithm. The most involved part is simulating \AlgExtendStructures, which is described in \cref{section:general-AlgExtendStructures}. That is the only algorithm which our framework does not simulate exactly.

\paragraph{Remark:} As pointed out by \cite[arXiv Version 2, Footnote 4]{huang20221}, some of our procedures are stated in a form ``until matching-size $< X$'', which leaves the impression that an algorithm needs to know the global matching size. This is particularly important for models such as \CONGEST, where learning the size of a global matching might be round-wise costly.
Nevertheless, our proofs provide explicit upper bounds on how many times such loops will be executed. Hence, if an upper bound is $1/\eps^k$, a repeat-until loop can be replaced by a for-loop executed $1/\eps^k$ times. Moreover, this adjustment does not affect the proofs. We decided to phrase the algorithms using repeat-until loops as we believe it benefits the ease of readability.

\subsection{Simulation of \cref{alg:main}}
\label{sec:framework-init-with-O(1)-MM}

Recall that \cref{line:main-alg-maximal-matching} of \cref{alg:main} finds a maximal matching in a given graph. However, our framework only has access to a $\cMM$-approximate maximum matching. So, instead of executing a maximal matching, the framework invokes $\Amatching$ as a simulation of \cref{line:main-alg-maximal-matching} of \cref{alg:main}. We now discuss how this affects the pass-complexity.

Our analysis, and in particular the one at the end of \cref{sec:upper-bound-number-phases}, uses the initial matching size to upper-bound the number of phases. Concretely, until the current matching is a $(1+\eps)$-approximate one, the analysis states that the matching size doubles within the number of phases stated at the end of \cref{sec:upper-bound-number-phases}. Let that number of phases be $P$. Hence, our framework requires $\log \cMM \cdot P$ phases.

The framework executes the rest of \cref{alg:main} directly.

\subsection{Simulation of \cref{alg:phase}} 
\cref{line:reset-label,line: active} are executed by appropriately marking matched edges in \storage. \cref{line:on-hold,line:not-on-hold} are executed by invoking $\Apropagate$ on each structure. In the sequel we describe how to simulate \AlgExtendStructures, \AlgEdgeMerge, and \AlgAugmentStructures.

\subsection{Simulation of \AlgAugmentStructures} 
To execute this method, our framework has to know (1) whether a vertex is incident to a matched edge, and (2) what structure (if any) a vertex belongs to. This information is maintained in \storage, but also can be obtained via $\Apropagate$. A single traversal by $\Aedgescan$ over the unmatched edges suffices to add them to the corresponding structure. 

\begin{lemma}
	Given the access to \storage as described in \eqref{item:storage}, and to algorithm $\Aedgescan$ as described in \eqref{item:edge-scan} and to algorithm $\Apropagate$ as described in \eqref{item:propagate}, there is an algorithm that simulates \AlgAugmentStructures in time $O(\Tedgescan + \Tpropagate)$.
\end{lemma}

\subsection{Simulation of $\AlgMerge(\mathcal{P}, g)$}
Given an unmatched arc $g = (u, v)$, the algorithm first fetches from \storage the structures containing $u$ and $v$; let $\structure_\alpha$ and $\structure_\beta$ be those structures, containing $u$ and $v$ respectively. Then, the algorithm finds an alternating path from $\alpha$ to $u$ and from $\beta$ to $v$ by invoking $\Aaggregate$ and performing the search in-place. The augmentation is performed along $g$ and those two alternating paths via $\Apropagate$. Moreover, given the guarantees on $\Aaggregate$ and $\Apropagate$, the simulations of these method can be performed for \emph{multiple} arcs $g$ in parallel, as long as the corresponding structures do not overlap.

\begin{lemma}\label{lemma:simulate-AlgMerge}
	Given the access to \storage (as described in \eqref{item:storage}), and to algorithms $\Apropagate$ (as described in \eqref{item:propagate}) and $\Aaggregate$ (as described in \eqref{item:aggregate}), there is an algorithm that simulates an invocation $\AlgMerge(\mathcal{P}, g)$ in time $O(\Tpropagate + \Taggregate)$. In the same time can be simulated even multiple invocations of $\AlgMerge$ in parallel, as long as the structures containing endpoints of arcs $g$ used in those invocations are distinct.
\end{lemma}

\subsection{Simulation of $\AlgOvertake(P_u, g, \astar)$}
\label{section:simulate-AlgOvertake}
Given an unmatched arc $g = (u, v)$, the algorithm first invokes $\Aaggregate$ to aggregate the structure containing $u$ and the structure containing $v$. $\AlgOvertake(P_u, g, \astar)$ is performed in-place then.

\begin{lemma}\label{lemma:simulate-AlgOvertake}
	Given the access to \storage as described in \eqref{item:storage}, and to algorithm $\Aaggregate$ as described in \eqref{item:aggregate}, there is an algorithm that simulates $\AlgOvertake(P_u, g, \astar)$ in time $O(\Taggregate)$. In the same time can be simulated multiple invocations of $\AlgOvertake$ in parallel, as long as the structures containing endpoints of arcs $g$ used in those invocations are distinct.
\end{lemma}

\subsection{(An Almost) Simulation of \AlgEdgeMerge} 
\label{section:simulation-AlgEdgeMerge}

Our goal is to reduce the simulation of \AlgEdgeMerge to the task of finding an $O(1)$-approximate maximum matching. 

\subsubsection{Graph representation}
\label{section:graph-representation-simulation-AlgEdgeMerge}
Given the current state of structures, labels and matched edges, we now describe how to build a graph $H$ that serves to simulate \AlgEdgeMerge. 
\begin{itemize}
	\item $V(H)$ is the set of free nodes that have not been removed in this phase.
	\item For two free nodes $\alpha$ and $\beta$ in $V(H)$, there is an edge $\{\alpha, \beta\}$ in $H$ if there is an alternating path between $\alpha$ and $\beta$ within $\RD(\structure_\alpha \cup \structure_\beta \cup \{g\})$, where $g$ is an unmatched edge. Record $g$ along with the edge $\{\alpha, \beta\}$ in $E(H)$. In case of ties, record any such $g$.
\end{itemize}

\newcommand{\Erem}{E_{\rm{rem}}}
\newcommand{\cArem}{\cA_{\rm{rem}}}
\newcommand{\Hfinal}{H_{\rm{final}}}
\subsubsection{Algorithm}
The general-framework simulation of \AlgEdgeMerge is provided as \cref{alg:framework-AlgEdgeMerge}.
\begin{algorithm}[h]
  \Input{$G$: a graph \\
		Free nodes \\
		The current matching $M$ \\
		A label for each matched arc
		}
		\tcc{We assume we are given an access to $\cMM$-approximate maximum matching algorithm.}
		Let $q \eqdef \eps^{31} / (d \cdot \log \cMM \cdot \cMM)$, where $d$ is an absolute constant. \label{line:framework-define-q}
		
		\Repeat{$|M_H| < q \cdot |M|$ \label{line:framework-AlgEdgeMerge-until-condition}} {
			Build a graph $H$ as described in \cref{section:graph-representation-simulation-AlgEdgeMerge}.
		
			Let $M_H$ be a $\cMM$-approximate maximum matching in $H$.
			
			\For{each $e = \{\alpha, \beta\} \in M_H$} {
				Perform an augmentation between $\alpha$ and $\beta$ by simulating $\AlgMerge(\mathcal{P}, g)$, where $g$ is an unmatched edge leading to the augmentation.
			}
		}

		Build a graph $H$ as described in \cref{section:graph-representation-simulation-AlgEdgeMerge}. \label{line:H-final}
		
		\For{each $e = \{\alpha, \beta\} \in H$ \label{line:loop-g-to-remove-from-G}} {
			Remove \emph{each} edge $g$ in $G$ corresponding to $e$. Note that there might be multiple such edges $g$ as there might be multiples ways to construct an augmentation between $\alpha$ and $\beta$.\label{line:remove-g-from-G}
		}

  \caption{An almost simulation of \AlgEdgeMerge.}
	\label{alg:framework-AlgEdgeMerge}
\end{algorithm}

We point out that one of the main reasons in removing the edges in the for-loop on \cref{line:loop-g-to-remove-from-G} is to ensure that \cref{obs: opposing arcs} holds.

We now analyze the correctness and the running time of this simulation.
\paragraph{Correctness.}
Let $\Erem$ be the set of edges $g$ removed from $G$ on \cref{line:remove-g-from-G}.
Observe that our algorithm \AlgEdgeMerge does not impose any ordering of the edges in the stream. So, in particular, assume that the stream first lists unmatched edges $g$ recorded on $M_H$ (in any order), and then the stream presents all the remaining edges of $E(G \setminus \Erem)$. It is now easy to see that our simulation described above is equivalent to \AlgEdgeMerge executed on such stream.
By removing $\Erem$ from $G$ on \cref{line:remove-g-from-G} we assure that no other augmentation is possible after our simulation by \cref{alg:framework-AlgEdgeMerge} is over.~\footnote{Note that if instead of $O(1)$-approximate maximum matching we used maximal matching to compute $M_H$, the removal of edges on \cref{line:remove-g-from-G} would not be needed to assure that no other augmentation is possible than those our simulation finds.}

We also note that no edge in $\Erem$ belongs to a structure. As such, removing $\Erem$ does not affect properties of structures.

Removing $\Erem$ ensures that \cref{obs: opposing arcs} holds.

\paragraph{Running time.}
The edge set $E(H)$ is not created explicitly, but on-the-fly as follows. First, unmatched edges that do not have endpoints in distinct structures are removed from the input to $\Amatching$. Then, when $\Amatching$ accesses an edge $\{u, v\} \in E(G)$ such that an augmentation is possible, the framework replaces that edge by $\{\alpha, \beta\}$ such that $u \in \structure_\alpha$ and $v \in \structure_\beta$.

To simulate the loop on \cref{line:loop-g-to-remove-from-G} it is not even necessary to create $H$, but rather to remove each edge from $G$ over which is possible augmentation via the corresponding structures.

Observe that each iteration of the repeat-loop \cref{alg:framework-AlgEdgeMerge} increases the matching size our algorithm maintains by $|M_H|$; each edge in $M_H$ leads to an augmentation in $G$. Also, recall that $|M|$ is a $\cMM$-approximate maximum matching; see \cref{sec:framework-init-with-O(1)-MM} for more details. Therefore, the total sum of $|M_H|$ across all iterations does not exceed $(\cMM - 1) \cdot |M|$.
In addition, in each iteration, except the very last one, of the repeat-loop it holds that $|M_H| \ge q \cdot |M|$. Therefore, the number of repeat-loop iterations is at most $(\cMM - 1) / q$.

This discussion together with \cref{lemma:simulate-AlgMerge} lead to the following claim.
\begin{lemma}
	Given the access to \storage as described in \eqref{item:storage}, and to algorithms described in \eqref{item:matching}, \eqref{item:edge-scan}, \eqref{item:propagate} and \eqref{item:aggregate}, \cref{alg:framework-AlgEdgeMerge} can be implemented in time $O((\Tmatching + \Tedgescan + \Tpropagate + \Taggregate) \cdot \cMM / q)$, where $q$ is defined on \cref{line:framework-define-q}.
	If $\cMM \in O(1)$, then $q \in \poly(1/\eps)$.
\end{lemma}

\paragraph{Approximation guarantee.}
\cref{alg:framework-AlgEdgeMerge} removes certain edges, the set $\Erem$, from $G$. Our aim now is to upper-bound the number of augmentations in $G$ that are lost due to the removal of $\Erem$.

Let $\Hfinal$ be the graph $H$ constructed on \cref{line:H-final}. Given the condition on \cref{line:framework-AlgEdgeMerge-until-condition}, the maximum matching in $\Hfinal$ is at most $\cMM \cdot q \cdot |M|$. From \cref{lemma:structuresize} and our setup of parameters, a structure contains at most $Y \eqdef 2/\eps^{11}$ vertices, for $0 < \eps \le 1/2$.

We now claim that removing $\Erem$ destroys at most $2 \cdot \cMM \cdot q \cdot |M| \cdot Y$ vertex-disjoint augmentations in $G$. To see that, observe that the minimum vertex cover in $\Hfinal$ has size at most $2 \cdot \cMM \cdot q \cdot |M|$.
This further implies that the minimum vertex cover in $\Erem$ has size at most $2 \cdot \cMM \cdot q \cdot |M| \cdot Y$. The maximum number of vertex-disjoint augmentations destroyed by removing $\Erem$ is upper-bounded by the maximum matching size in $\Erem$ -- each such augmentation contains at least one edge $g$ in $\Erem$ and no other destroyed augmentation contains an endpoint of $g$. Hence, the number of destroyed augmentations is upper-bounded by $2 \cdot \cMM \cdot q \cdot |M| \cdot Y$.

By \cref{lemma: runtime} and its proofs, there exists an absolute constant $c$ such that our streaming algorithm makes at most $c / \eps^{19}$ passes. Taking into account our discussion in \cref{sec:framework-init-with-O(1)-MM}, it corresponds to $\log \cMM \cdot c / \eps^{19}$ passes simulated by the framework. Therefore, if in each almost-simulation of \AlgEdgeMerge by \cref{alg:framework-AlgEdgeMerge} are destroyed $2 \cdot \cMM \cdot q \cdot |M| \cdot Y$ augmentations, the simulation of \cref{alg:main} misses at most
\[
	\frac{2 \cdot c \cdot \log \cMM \cdot \cMM \cdot q \cdot |M| \cdot Y}{\eps^{19}} = \frac{4 \cdot c \cdot \log \cMM \cdot \cMM \cdot q \cdot |M|}{\eps^{30}}
\]
augmentations. By letting $q = \eps^{31} / (4 \cdot c \cdot \log \cMM \cdot \cMM)$ we conclude that our general framework outputs $(1 + 2 \eps)$-approximate maximum matching if it uses \cref{alg:framework-AlgEdgeMerge} instead of \AlgEdgeMerge.

\begin{lemma}
	Using \cref{alg:framework-AlgEdgeMerge} in place of \AlgEdgeMerge in \cref{alg:main} outputs a $(1 + 2\eps)$-approximate maximum matching.
\end{lemma}

\subsection{(An Almost) Simulation of \AlgExtendStructures} 
\label{section:general-AlgExtendStructures}
Finally, it remains to discuss the implementation of \AlgExtendStructures. Ideally, we want to reduce this task to one of finding a maximal matching. As we will see, instead of simulating \AlgExtendStructures entirely, we will simulate it for all but a tiny fraction of free nodes.

\subsubsection{Graph representation}
\label{section:graph-representation}
First, the framework invokes $\Aaggregate$. Then, for each vertex $u$ of each structure, the framework computes a path $P_u$ as described on \cref{line:extend-structure-P_u} of \AlgExtendStructures (\cref{algo-extend-structures}). In case of ties, a shortest $P_u$ is chosen. Also, for every vertex $u \in \structure_\alpha$ for each free node $\alpha$, the framework computes whether that is an alternating path in $\structure_\alpha$ that originates at $\alpha$ and ends by a matched arc whose head is $u$. This information will be useful to check whether there is an augmentation between $\alpha$ and some other free node. Finally, the framework invokes $\Apropagate$ to notify each vertex $u$ in the structure of $P_u$.

Second, the framework builds a graph $H$ that serves to almost simulate \AlgExtendStructures. 
\begin{itemize}
	\item $V(H)$ is the union of the free nodes \emph{and} the matched edges that \emph{do not} yet belong to any structure.
	\item Let $\alpha$ be a free node such that $\alpha$ is not on hold, $\alpha$ has not been augmented (this condition corresponds to \cref{line:ES-check-augmented} of \cref{algo-extend-structures} evaluating to $\false$) and $\alpha$'s active path has not been already extended in this invocation of $\AlgExtendStructures$ (this condition corresponds to \cref{line:ES-check-if-already-extended} of \cref{algo-extend-structures} evaluating to $\false$).
	
	For a free node $\beta$, there is an edge $\{\alpha, \beta\}$ in $H$ if:
	\begin{itemize}
		\item $\astar$ defined on \cref{line:ES-define-astar} of \cref{algo-extend-structures} is in $\structure_\beta$; and
		\item \cref{line:ES-check-astar-label} of \cref{algo-extend-structures} evaluates to $\true$ for a path $P_u$ originating at $\alpha$. Note that we allow $\alpha = \beta$, i.e., $H$ can have self-loops.
	\end{itemize}
		
	In analogous way, for a free node $\alpha$ and a matched edge $e$ that does not belong to any structure yet, there is an edge $\{\alpha, e\}$ in $H$ if:
	\begin{itemize}
		\item $\astar$ defined on \cref{line:ES-define-astar} of \cref{algo-extend-structures} is such that $\RD(\astar) = e$; and 
		\item \cref{line:ES-check-astar-label} of \cref{algo-extend-structures} evaluates to $\true$ for a path $P_u$ originating at $\alpha$.
	\end{itemize}
	In addition, for each edge $\{x, y\}$ in $H$ record the unmatched arc $g$ (in case of ties, choose any $g$) and $\astar$ that led to adding $\{x, y\}$ to $H$. These arcs $g$ and $\astar$ corresponds to $g$ and $\astar$ in \AlgExtendStructures.
\end{itemize}

\subsubsection{Algorithm}
We are now ready to state an almost simulation of \AlgExtendStructures (see \cref{alg:framework-AlgExtendStructure}) that will lead to our general framework.
\begin{algorithm}[h]
  \Input{$G$: a graph \\
		Free nodes \\
		The current matching $M$ \\
		A label for each matched arc
		}
		
		\tcc{We assume we are given an access to $\cMM$-approximate maximum matching algorithm.}
		
		Let $r \eqdef \eps^{32} / (d \cdot \log \cMM \cdot \cMM)$, where $d$ is an absolute positive constant. \label{line:framework-define-r}
		
		\Repeat{$|M_H| < r \cdot |M|$} {
			Build a graph $H$ as described in \cref{section:graph-representation}.
			
			Let $SL$ be the set of vertices in $H$ having self-loops.
			
			Let $\hH$ be obtained by removing $SL$ from $H$.
			
			Let $M_\hH$ be a $\cMM$-approximate maximum matching in $\hH$ obtained by executing $\Amatching$.
			
			Let $M_H := \{\{\alpha, \alpha\}\ |\ \alpha \in SL\} \cup M_\hH$. \label{line:matching-in-H}
			
			\For{each $e \in M_H$} {
				Let $g = (u, v)$ be the unmatched arc and $\astar$ be the matched arc recorded with $e$.
				
				Let $\structure_\alpha$ be the structure containing $u$.
				
				\If{$v \in \structure_\beta \neq \structure_\alpha$, and there is an augmentation in $\RD(\structure_\alpha \cup \structure_\beta \cup \{g\})$ \label{line:framework-if-augmentation}} {
					Simulate $\AlgMerge(\mathcal{P}, g)$. \label{line:simulation-AlgEdgeMerge-invoke-AlgMerge}
				}
				\Else {
					Simulate $\AlgOvertake(P_u, g, \astar)$. \label{line:simulation-AlgEdgeMerge-invoke-AlgOvertake}
				}
			}
		}

		Simulate \AlgEdgeMerge as described in \cref{section:simulation-AlgEdgeMerge}.\label{line:framework-AlgEdgeMerge}

  \caption{An almost simulation of \AlgExtendStructures.}
	\label{alg:framework-AlgExtendStructure}
\end{algorithm}

\subsubsection{Running time} The framework builds $H$ described in \cref{section:graph-representation} on-the-fly. To decide whether an unmatched edge $\{u, v\}$ belongs to $H$, the algorithm test whether \cref{line:ES-check-astar-label} of \cref{algo-extend-structures} evaluates to $\true$ for $P_u$ and the matched arc $\astar$ (if any) whose tail is $v$. $P_u$ is computed as described in the beginning of \cref{section:graph-representation}. If it evaluates to $\true$, to $H$ is added either $\{\alpha, \beta\}$ such that $u \in \structure_\alpha$ and $\astar \in \structure_\beta$, or $\{\alpha, e\}$ is added to $H$ in case $v \in \structure_\alpha$ and $e = \RD(\astar)$ does not belong to any structure.

To test whether \cref{line:framework-if-augmentation} of \cref{alg:framework-AlgExtendStructure} evaluates to $\true$, it suffices to check whether there is an alternating path in $\structure_\beta$ that originates at $\beta$ and ends with a matched arc whose head is $v$. As a reminder, that information is pre-computed for each vertex, as described in \cref{section:graph-representation}.

Now, we upper-bound the number of loop iterations. Each iteration of the main loop of \cref{alg:framework-AlgExtendStructure} updates/reduces or removes at least $|M_H|$ matched arcs. Label-reduction occurs via \cref{line:simulation-AlgEdgeMerge-invoke-AlgOvertake}, while edge-removal occurs via \cref{line:simulation-AlgEdgeMerge-invoke-AlgMerge}.
First, there can be at most $2 \maxlen |M|$ different label-updates; $\maxlen$ per each of the $2 |M|$ matched arcs.  
Second, in each iteration, except the very last one, of the repeat-until loop of \cref{alg:framework-AlgExtendStructure} it holds that $|M_H| \ge r |M|$.
Therefore, the repeat-until loop of \cref{alg:framework-AlgExtendStructure} executes at most $(2 \maxlen |M| + |M|) / (r |M|) = (2/\eps + 1) / r$ times. Here, $+|M|$ in $2 \maxlen |M| + |M|$ upper-bounds the number of removed matched edges.

\begin{lemma}
	Given the access to \storage (as described in \eqref{item:storage}), and to algorithms described in \eqref{item:matching}, \eqref{item:edge-scan}, \eqref{item:propagate} and \eqref{item:aggregate}, \cref{alg:framework-AlgExtendStructure} can be executed in time $O((\Tmatching + \Tedgescan + \Tpropagate + \Taggregate) / (\eps \cdot r))$, where $r$ is defined on \cref{line:framework-define-r}. If $\Amatching$ is an $O(1)$-approximate maximum matching algorithm, then $r \in \poly(1/\eps)$.
\end{lemma}


\subsubsection{Approximation guarantee}
The primary goal of \AlgExtendStructures is to extend active paths of a maximal (not necessary maximum) number of distinct free nodes with respect to a given ordering of arcs. \cref{alg:framework-AlgExtendStructure} does not achieve the same guarantee. As a consequence of such behavior of \cref{alg:framework-AlgExtendStructure}, \AlgBacktrack potentially reduces some active paths although those active paths can be extended, leading to some augmentations not be found. In this section we upper-bound the number of such ``lost'' augmentations.

We point out that the repeat-until loop of \cref{alg:framework-AlgExtendStructure} might also miss to find certain augmentations, e.g., like those found by \cref{line:ES-check-alpha-beta-merge} of \AlgExtendStructures. To remedy that, \cref{alg:framework-AlgExtendStructure} invokes the simulation of \AlgEdgeMerge on \cref{line:framework-AlgEdgeMerge}. We note that this simulation is only for a closer resemblance between \cref{alg:framework-AlgExtendStructure} and \AlgExtendStructures (\cref{algo-extend-structures}), as \AlgEdgeMerge is invoked as a part of the simulation of \cref{alg:phase} anyways.

\paragraph{Upper-bounding missed $\AlgOvertake$ invocations.}
The repeat-until loop of \cref{alg:framework-AlgExtendStructure} terminates when a $\cMM$-approximate maximum matching in $H$ becomes smaller than $r \cdot |M|$. This implies that by removing at most $2 \cdot \cMM \cdot r \cdot |M|$ structures no structure would execute \AlgOvertake. From \cref{lemma:structuresize} and our setup of parameters, a structure contains at most $2/\eps^{11}$ matched arcs, for $0 < \eps \le 1/2$. Hence, $2 \cdot \cMM \cdot r \cdot |M|$ structures correspond to at most $\maxlen \cdot 2/\eps^{11} \cdot 2 \cdot \cMM \cdot r\cdot  |M|$ distinct label updates. Since each invocation of \AlgOvertake reduces a label, a termination of the repeat-until loop misses at most $4r \cdot \cMM /\eps^{12} \cdot |M|$ extensions of active paths.

By \cref{lemma: runtime} and its proofs, there exists an absolute constant $c$ such that our streaming algorithm makes at most $c / \eps^{19}$ passes. Taking into account our discussion in \cref{sec:framework-init-with-O(1)-MM}, it corresponds to $\log \cMM \cdot c / \eps^{19}$ passes simulated by the framework. Therefore, if in each almost-simulation of \AlgExtendStructures by \cref{alg:framework-AlgExtendStructure} are missed $4r \cdot \cMM /\eps^{12} \cdot |M|$ extensions of active paths that potentially lead to augmentations, the simulation of \cref{alg:main} misses at most $4 c \cdot r \cdot \log \cMM \cdot \cMM /\eps^{31} \cdot |M|$ augmentations. By letting $r = \eps^{32} / (4 c \cdot \log \cMM \cdot \cMM)$ we conclude that using \cref{alg:framework-AlgExtendStructure} instead of \AlgExtendStructures outputs a $(1 + 2 \eps)$-approximate maximum matching.

\begin{lemma}
	Using \cref{alg:framework-AlgExtendStructure} in place of \AlgExtendStructures in \cref{alg:main} outputs a $(1 + 2\eps)$-approximate maximum matching.
\end{lemma}

\section{Proof of \cref{lemma: overtaking}}\label{sec:proof-of-lemma-overtaking}
	\lemmaPreserve*

\noindent First, consider the case that $\astar \in \structure_{\beta}$, for $\alpha \neq \beta$, is an inactive arc.
	We start by showing that both structures remain connected.

	
	\paragraph{$\structure_\beta$ remains connected.}
	Consider an arc $d_q \in \structure_\beta$ that remains in $S_\beta$ after the overtake.
	By the definition of overtaking, there exists an alternating path $P(d_q) = (\beta, d_1, d_2, \ldots, d_q) \in \structure_\beta$ that does not include $\astar$; otherwise, $d_q$ would be added to $\structure_\alpha$.
	Now, we conclude that there is also an alternating path to each $d_i$ (a prefix of $P(d_q)$), $1 \leq i \leq q$, that does not contain $\astar$ and hence, the whole path $P(d_q)$ remains in $\structure_\beta$ after the overtake.
	Thus, $\structure_\beta$ remains connected in the case that $\astar$ is inactive.
	
	\paragraph{$\structure_\alpha$ remains connected.}	
	Let $d_q$ be an overtaken arc by $\structure_\alpha$ from $\structure_\beta$.
	Consider $d_q$ before the overtake and let $P(d_q) = (\beta, d_1, \ldots, d_q)$ be an alternating path in $\structure_\beta$.
	Suppose that $\astar = d_i \in P(d_q)$ for some $1 \leq i < q$; notice that $\astar \in P(d_q)$ as otherwise it cannot be the case that $\alpha$ takes the arc $d_q$ from $\beta$.
	We now argue that unless there is an alternating path between $\alpha$ and $\beta$, then $\alpha$ takes all arcs $d_i, d_{i + 1}, \ldots, d_q$ into $S_\alpha$ and we obtain connectivity.
	
	Towards a contradiction, assume that $d_j$ is not overtaken, for some $i < j < q$. Then, there is an alternating path $P(d_j) \in \structure_\beta$ to $d_j$ such that $d_i = \astar \not \in P(d_j)$. 
	Let $j'$ be the index $i < j' \le q$ such that $d_{j'}$ or $\leftarc{d_{j'}}$ is the first intersection between $P(d_j)$ and edges $\{d_i, \ldots, d_q\}$, where ``first'' is taken with respect to $P(d_j)$.
	Now, we consider two cases: $d_{j'}$ or $\leftarc{d_{j'}}$ is the intersection. In case $d_{j'}$ is the intersection, then there is an alternating path to $d_q$ within $\structure_\beta$ that does not contain $\astar$ (the path goes along $P(d_j)$ until $d_{j'}$ and then along $d_{j'+1}, \ldots, d_q$), contradicting that $\alpha$ overtook $d_q$. In case $\leftarc{d_{j'}}$ is the intersection, then there is an augmenting path between $\alpha$ and $\beta$ that goes from $\beta$ to $\leftarc{d_{j'}}$ along $P(d_j)$, and then to $\alpha$ along $\leftarc{d_{j' - 1}}, \ldots, \leftarc{d_{i}}$, contradicting the presumption of $\AlgOvertake$ (c.f. \Cref{section:AlgOvertake}) that there is no alternating path between $\alpha$ and $\beta$ in $\RD(\structure_\alpha \cup \structure' \cup \{g\})$.
	Hence, $\structure_\alpha$ is connected after the overtake.
	
	Then, consider the case, where $\astar$ in an active arc of $S_\beta$. 
	Let $P_{\astar}$ be the suffix of the active path $P_\beta$ of $\beta$ starting at $\astar$.	
	The connectivity follows from applying the above argumentation sequentially to each arc in $P_\astar$.

	\paragraph{Lengths of active paths; maintaining \cref{invariant:active-path-length}.}
	Now, we want to show that the lengths of the active paths are at most $\maxlen$.	
	This is clear for $\structure_\beta$ since its active path does not grow during an overtake.
	
	If $\astar$ is \emph{not} an active arc of $\structure_\beta$, the active path of $\structure_\alpha$ grows by just one, and since the algorithm only extends an active path in case its length is smaller than $\maxlen$, the resulting path is of length at most $\maxlen$.
	
	Then, consider the second case where $\astar$ is active and let $P_\beta = (\beta, d_1, d_2, \ldots, d_q)$ be the active path of $\beta$. Let $P_u$ be a path on the input of $\AlgOvertake$.
	Let $\astar = d_i$ for some $1 \leq i \leq q$ and let $P_{\astar}$ be the suffix of $P_\beta$ starting at $\astar$ (this is the same definition as of $P_{\astar}$ in \cref{section:AlgOvertake}).
	Notice that since $|P_\beta| \leq \maxlen$ and \cref{invariant:active-path-length} holds, it is the case that $|P_\astar| \leq \maxlen - (i - 1) = \maxlen - i + 1$.
	By \cref{invariant:active-path-length} we have that $\ell(\astar) \leq i$, and by \cref{algo-extend-structures} $\alpha$ only invokes $\AlgOvertake$ if $|P_u| < i$.
	The new active path of $\alpha$ is $P_u \circ P_\astar$ and we have that $|P_u \circ P_\astar| \leq (i - 1) + \maxlen - i + 1 = \maxlen$.
	Hence, from this arguments and from the label update performed by Step~\eqref{item:update-labels} in \cref{section:AlgOvertake}, we obtain the desired bounds for the length of the new active paths and for the labels of its matched arcs.
	
	\paragraph{Property \eqref{prop:unmatched-arcs} of \cref{definition:structure}.}
		Note that each unmatched arc added to $\structure_\alpha$ is done by adding alternating paths to $\structure_\alpha$ that start and end with matched arcs. Hence, this property is preserved for $\structure_\alpha$.		
		Step~\eqref{item:clean-unmatched} in \cref{section:AlgOvertake} ensures that \cref{definition:structure}-\ref{prop:unmatched-arcs} holds for $\structure_\beta$ as well.

	\paragraph{Disjointness.}
	Finally, we want to show that the structures remain vertex disjoint. Consider the state of $\structure_\alpha$ and $\structure_\beta$ after the overtake.	
	Since a matched arc added to $\structure_\alpha$ is removed from $\structure_\beta$ and no matched arc is added to $\structure_\beta$, it holds that no matched arc belongs to both $\structure_\alpha$ and $\structure_\beta$.	
	However, we also show that if a matched arc $a \in \structure_\alpha$ it cannot be the case that $\leftarc{a} \in \structure_\beta$.
	We prove that if that was the case, then the algorithm could find an augmenting path from $\alpha$ to $\beta$, which is assumed not to exist.
	
	Towards a contradiction, assume that $d_q \in \structure_\alpha$ and $\leftarc{d_q} \in \structure_\beta$. Let $P(d_q)$ be an alternating path to $d_q$ from $\alpha$ within $\structure_\alpha$ and $P(\leftarc{d_q})$ be an alternating path from $\beta$ to $\leftarc{d_q}$ within $\structure_\beta$; from our discussion on connectivity of $\structure_\alpha$ and $\structure_\beta$ it follows that such alternating paths exist.
	Let $b$ be the first arc along $P(d_q)$ such that $b \in P(d_q)$ and $\leftarc{b} \in P(\leftarc{d_q})$; note that such $b$ exists as those two alternating paths intersect at $d_q$. Then, there is an augmenting path between $\alpha$ and $\beta$ that goes to $b$ along $P(d_q)$, and then from $\leftarc{b}$ to $\beta $ along $P(\leftarc{d_q})$, again contradicting there is no alternating path between $\alpha$ and $\beta$ in $\RD(\structure_\alpha \cup \structure' \cup \{g\})$. 
	
	This now implies that the vertices induced by \emph{matched} arcs of $\structure_\alpha$ and the vertices induced by \emph{matched} arcs of $\structure_\beta$ are disjoint. Since Property~\eqref{prop:unmatched-arcs} of \cref{definition:structure} holds for both $\structure_\alpha$ and $\structure_\beta$ after the updates by $\AlgOvertake$, as shown above, we conclude that the updated structures are vertex disjoint.

\paragraph{The case that $\astar \not \in \structure_\beta$ for $\alpha \neq \beta$.}
Finally, we note that we did not explicitly consider the case where $\astar$ is not part of any structure or $\astar \in \structure_\alpha$ at the time when $\AlgOvertake$ is invoked. If $\astar \in \structure_\alpha$, note that connectedness of $\structure_\alpha$ follows directly, while $\ell(\astar) \le \maxlen$ is guaranteed by \cref{line:ES-check-astar-label} of \cref{algo-extend-structures}. This guarantees that \cref{invariant:active-path-length} is maintained. Disjointness of structures follows trivially. Finally, property \eqref{prop:unmatched-arcs} of \cref{definition:structure} follows from the fact that the newly added unmatched edge to $\structure_\alpha$ is a part of an alternating path ending by $\astar$.

The case where $\astar$ does not belong to any structure follows from analogous arguments.


\subsection*{Acknowledgments}
We thank MohammadTaghi Hajiaghayi and Peilin Zhong for their helpful feedback on the exposition of this writeup. We thank Mohsen Ghaffari and Christoph Grunau for their insightful discussions. In particular, we thank Mohsen Ghaffari for suggesting how to use $O(1)$-approximate maximum matchings in our analysis in \cref{section:simulation-AlgEdgeMerge}.
We also thank anonymous reviewers for their valuable comments.

\bibliographystyle{alpha}
\bibliography{ref}

\newcommand{\etalchar}[1]{$^{#1}$}
\begin{thebibliography}{KMNFT20}

\bibitem[AB21]{Assadi2021}
Sepehr Assadi and Soheil Behnezhad.
\newblock {Beating Two-Thirds For Random-Order Streaming Matching}.
\newblock In {\em International Colloquium on Automata, Languagesand
  Programming (ICALP)}, 2021.

\bibitem[ABB{\etalchar{+}}19]{assadi2019coresets}
Sepehr Assadi, MohammadHossein Bateni, Aaron Bernstein, Vahab Mirrokni, and
  Cliff Stein.
\newblock {Coresets Meet EDCS: Algorithms for Matching and Vertex Cover on
  Massive Graphs}.
\newblock In {\em Proceedings of the ACM-SIAM Symposium on Discrete Algorithms
  (SODA)}, pages 1616--1635. SIAM, 2019.

\bibitem[AG11]{ahn2011linear}
Kook~Jin Ahn and Sudipto Guha.
\newblock {Linear Programming in the Semi-Streaming Model with Application to
  the Maximum Matching Problem}.
\newblock In {\em International Colloquium on Automata, Languages, and
  Programming (ICALP)}, pages 526--538, 2011.

\bibitem[AG18]{Ahn2018}
Kook~Jin Ahn and Sudipto Guha.
\newblock {Access to Data and Number of Iterations: Dual Primal Algorithms for
  Maximum Matching under Resource Constraints}.
\newblock {\em ACM Trans. Parallel Comput.}, 4(4):1--40, 2018.

\bibitem[AJJ{\etalchar{+}}22]{assadi2022semi}
Sepehr Assadi, Arun Jambulapati, Yujia Jin, Aaron Sidford, and Kevin Tian.
\newblock Semi-streaming bipartite matching in fewer passes and optimal space.
\newblock In {\em Proceedings of the 2022 Annual ACM-SIAM Symposium on Discrete
  Algorithms (SODA)}, pages 627--669. SIAM, 2022.

\bibitem[AKL17]{assadi2017estimating}
Sepehr Assadi, Sanjeev Khanna, and Yang Li.
\newblock {On Estimating Maximum Matching Size in Graph Streams}.
\newblock In {\em the Proceedings of the Symposium on Discrete Algorithms
  (SODA)}, pages 1723--1742, 2017.

\bibitem[AKLY16]{assadi2016maximum}
Sepehr Assadi, Sanjeev Khanna, Yang Li, and Grigory Yaroslavtsev.
\newblock {Maximum Matchings in Dynamic Graph Streams and the Simultaneous
  Communication Model}.
\newblock In {\em the Proceedings of the Symposium on Discrete Algorithms
  (SODA)}, pages 1345--1364, 2016.

\bibitem[ALT21]{assadi2021auction}
Sepehr Assadi, S~Cliff Liu, and Robert~E Tarjan.
\newblock An auction algorithm for bipartite matching in streaming and
  massively parallel computation models.
\newblock In {\em Symposium on Simplicity in Algorithms (SOSA)}, pages
  165--171. SIAM, 2021.

\bibitem[AN21]{assadi2021GraphStreamingLowerBounds}
Sepehr Assadi and Vishvajeet N.
\newblock {Graph Streaming Lower Bounds for Parameter Estimation and Property
  Testing via a Streaming {XOR} Lemma}.
\newblock {\em CoRR}, abs/2104.04908, 2021.

\bibitem[Ass22]{assadi2022two}
Sepehr Assadi.
\newblock A two-pass (conditional) lower bound for semi-streaming maximum
  matching.
\newblock In {\em Proceedings of the 2022 Annual ACM-SIAM Symposium on Discrete
  Algorithms (SODA)}, pages 708--742. SIAM, 2022.

\bibitem[Ber88]{bertsekas1988auction}
Dimitri~P Bertsekas.
\newblock The auction algorithm: A distributed relaxation method for the
  assignment problem.
\newblock {\em Annals of operations research}, 14(1):105--123, 1988.

\bibitem[Ber20]{Bernstein2020}
Aaron Bernstein.
\newblock {Improved Bounds for Matching in Random-order Streams}.
\newblock In {\em International Colloquium on Automata, Languages, and
  Programming (ICALP)}, 2020.

\bibitem[BHH19]{behnezhad2019exponentially}
Soheil Behnezhad, Mohammad~Taghi Hajiaghayi, and David~G Harris.
\newblock Exponentially faster massively parallel maximal matching.
\newblock In {\em 2019 IEEE 60th Annual Symposium on Foundations of Computer
  Science (FOCS)}, pages 1637--1649. IEEE, 2019.

\bibitem[BK19]{Bhattacharya2019}
Sayan Bhattacharya and Janardhan Kulkarni.
\newblock {Deterministically Maintaining a $(2 + \varepsilon)$-Approximate
  Minimum Vertex Cover in $O(1/\varepsilon^2)$ Amortized Update Time}.
\newblock In {\em Proceedings of the Symposium on Discrete Algorithms (SODA)},
  2019.

\bibitem[BK22]{Behnezhad2022}
Soheil Behnezhad and Sanjeev Khanna.
\newblock {New Trade-Offs for Fully Dynamic Matching via Hierarchical EDCS}.
\newblock In {\em Proceedings of the Symposium on Discrete Algorithms}, 2022.

\bibitem[bKK20]{khalilK20}
Lidiya~Khalidah binti Khalil and Christian Konrad.
\newblock {Constructing Large Matchings via Query Access to a Maximal Matching
  Oracle}.
\newblock In {\em {IARCS} Conference on Foundations of Software Technology and
  Theoretical Computer Science (FSTTCS)}, volume 182, pages 26:1--26:15, 2020.

\bibitem[B{\L}M19]{Behnezhad2019}
Soheil Behnezhad, Jakub {\L}acki, and Vahab Mirrokni.
\newblock {Fully Dynamic Matching:Beating 2-Approximation in
  $\Delta^\varepsilon$ Update Time}.
\newblock In {\em Proceedings Symposium on Discrete Algorithms (SODA)}, 2019.

\bibitem[BS15]{bury2015sublinear}
Marc Bury and Chris Schwiegelshohn.
\newblock {Sublinear Estimation of Weighted Matchings in Dynamic Data Streams}.
\newblock In {\em the Proceedings of the European Symposium on Algorithms
  (ESA)}, pages 263--274. Springer, 2015.

\bibitem[BS16]{bernstein2016faster}
Aaron Bernstein and Cliff Stein.
\newblock {Faster Fully Dynamic Matchings with Small Approximation Ratios}.
\newblock In {\em the Proceedings of the Symposium on Discrete Algorithms
  (SODA)}, pages 692--711, 2016.

\bibitem[BST19]{buchbinder2019online}
Niv Buchbinder, Danny Segev, and Yevgeny Tkach.
\newblock {Online Algorithms for Maximum Cardinality Matching with Edge
  Arrivals}.
\newblock {\em Algorithmica}, 81(5):1781--1799, 2019.

\bibitem[CCE{\etalchar{+}}16]{chitnis2016kernelization}
Rajesh Chitnis, Graham Cormode, Hossein Esfandiari, MohammadTaghi Hajiaghayi,
  Andrew McGregor, Morteza Monemizadeh, and Sofya Vorotnikova.
\newblock {Kernelization via Sampling with Applications to Finding Matchings
  and Related Problems in Dynamic Graph Streams}.
\newblock In {\em the Proceedings of the Symposium on Discrete Algorithms
  (SODA)}, pages 1326--1344, 2016.

\bibitem[C{\L}M{\etalchar{+}}19]{czumaj2019round}
Artur Czumaj, Jakub {\L}acki, Aleksander Madry, Slobodan Mitrovic, Krzysztof
  Onak, and Piotr Sankowski.
\newblock Round compression for parallel matching algorithms.
\newblock {\em SIAM Journal on Computing}, 49(5):STOC18--1, 2019.

\bibitem[CTV15]{chiplunkar2015randomized}
Ashish Chiplunkar, Sumedh Tirodkar, and Sundar Vishwanathan.
\newblock {On Randomized Algorithms for Matching in the Online Preemptive
  Model}.
\newblock In {\em the Proceedings of the European Symposium on Algorithms
  (ESA)}, pages 325--336, 2015.

\bibitem[DH03]{drake2003improved}
Doratha~E. Drake and Stefan Hougardy.
\newblock {Improved Linear Time Approximation Algorithms for Weighted
  Matchings}.
\newblock In {\em Approximation, Randomization, and Combinatorial Optimization:
  Algorithms and Techniques (APPROX/RANDOM)}, pages 14--23. Springer, 2003.

\bibitem[DP14]{duan2014linear}
Ran Duan and Seth Pettie.
\newblock {Linear-Time Approximation for Maximum Weight Matching}.
\newblock {\em Journal of the ACM (JACM)}, 61(1):1--23, 2014.

\bibitem[Edm65a]{edmonds1965maximum}
Jack Edmonds.
\newblock {Maximum Matching and a Polyhedron with $0,1$-Vertices}.
\newblock {\em Journal of Research of the National Bureau of Standards B},
  69(125-130):55--56, 1965.

\bibitem[Edm65b]{edmonds_blossom}
Jack Edmonds.
\newblock {Paths, Trees, and Flowers}.
\newblock {\em Canadian Journal of Mathematics}, 17:449–467, 1965.

\bibitem[EHL{\etalchar{+}}18]{esfandiari2018streaming}
Hossein Esfandiari, Mohammadtaghi Hajiaghayi, Vahid Liaghat, Morteza
  Monemizadeh, and Krzysztof Onak.
\newblock {Streaming Algorithms for Estimating the Matching Size in Planar
  Graphs and Beyond}.
\newblock {\em ACM Transactions on Algorithms (TALG)}, 14(4):1--23, 2018.

\bibitem[EHM16]{esfandiari2016finding}
Hossein Esfandiari, MohammadTaghi Hajiaghayi, and Morteza Monemizadeh.
\newblock {Finding Large Matchings in Semi-Streaming}.
\newblock In {\em International Conference on Data Mining Workshops (ICDMW)},
  pages 608--614, 2016.

\bibitem[EKMS12]{Eggert2012}
Sebastian Eggert, Lasse Kliemann, Peter Munstermann, and Anand Srivastav.
\newblock {Bipartite Matching in the Semi-Streaming Model}.
\newblock {\em Algorithmica}, pages 490--508, 2012.

\bibitem[ELMS11]{epstein2011improved}
Leah Epstein, Asaf Levin, Juli{\'a}n Mestre, and Danny Segev.
\newblock {Improved Approximation Guarantees for Weighted Matching in the
  Semi-streaming Model}.
\newblock {\em SIAM Journal on Discrete Mathematics}, 25(3):1251--1265, 2011.

\bibitem[ELSW13]{epstein2013}
Leah Epstein, Asaf Levin, Danny Segev, and Oren Weimann.
\newblock {Improved Bounds for Online Preemptive Matching}.
\newblock In {\em 30th International Symposium on Theoretical Aspects of
  Computer Science (STACS)}, volume~20, pages 389--399, 2013.

\bibitem[FKM{\etalchar{+}}04]{feigenbaum2005graph}
Joan Feigenbaum, Sampath Kannan, Andrew McGregor, Siddharth Suri, and Jian
  Zhang.
\newblock {On Graph Problems in a Semi-Streaming Model}.
\newblock In {\em International Colloquium on Automata, Languages and
  Programming (ICALP)}, volume 3142, pages 531--543, 2004.

\bibitem[FKM{\etalchar{+}}05]{feigenbaum2005graphDistances}
Joan Feigenbaum, Sampath Kannan, Andrew McGregor, Siddharth Suri, and Jian
  Zhang.
\newblock {Graph Distances in the Streaming Model: the Value of Space}.
\newblock In {\em the Proceedings of the Symposium on Discrete Algorithms
  (SODA)}, volume~5, pages 745--754, 2005.

\bibitem[FS21]{feldman2021maximum}
Moran Feldman and Ariel Szarf.
\newblock Maximum matching sans maximal matching: A new approach for finding
  maximum matchings in the data stream model.
\newblock {\em arXiv preprint arXiv:2109.05946}, 2021.

\bibitem[Gab90]{gabow1990data}
Harold~N. Gabow.
\newblock {Data Structures for Weighted Matching and Nearest Common Ancestors
  with Linking}.
\newblock In {\em the Proceedings of the Symposium on Discrete Algorithms
  (SODA)}, pages 434--443, 1990.

\bibitem[GGK{\etalchar{+}}18]{ghaffari2018improved}
Mohsen Ghaffari, Themis Gouleakis, Christian Konrad, Slobodan Mitrovi{\'c}, and
  Ronitt Rubinfeld.
\newblock {Improved Massively Parallel Computation Algorithms for MIS,
  Matching, and Vertex Cover}.
\newblock In {\em Proceedings of the 2018 ACM Symposium on Principles of
  Distributed Computing (PODC)}, pages 129--138, 2018.

\bibitem[GKK12]{goel2012communication}
Ashish Goel, Michael Kapralov, and Sanjeev Khanna.
\newblock {On the Communication and Streaming Complexity of Maximum Bipartite
  Matching}.
\newblock In {\em the Proceedings of the Symposium on Discrete Algorithms
  (SODA)}, pages 468--485, 2012.

\bibitem[GKM{\etalchar{+}}19]{GamlathKMSW19}
Buddhima Gamlath, Michael Kapralov, Andreas Maggiori, Ola Svensson, and David
  Wajc.
\newblock {Online Matching with General Arrivals}.
\newblock In {\em the Proceedings of the Symposium on Foundations of Computer
  Science (FOCS)}, pages 26--37, 2019.

\bibitem[GKMS19]{gamlath2019weighted}
Buddhima Gamlath, Sagar Kale, Slobodan Mitrovic, and Ola Svensson.
\newblock Weighted matchings via unweighted augmentations.
\newblock In {\em Proceedings of the 2019 ACM Symposium on Principles of
  Distributed Computing}, pages 491--500, 2019.

\bibitem[GU19]{ghaffari2019sparsifying}
Mohsen Ghaffari and Jara Uitto.
\newblock {Sparsifying Distributed Algorithms with Ramifications in Massively
  Parallel Computation and Centralized Local Computation}.
\newblock In {\em Proceedings of the Thirtieth Annual ACM-SIAM Symposium on
  Discrete Algorithms (SODA)}, pages 1636--1653. SIAM, 2019.

\bibitem[GW19]{GhaffariW19}
Mohsen Ghaffari and David Wajc.
\newblock {Simplified and Space-Optimal Semi-Streaming (2+epsilon)-Approximate
  Matching}.
\newblock In {\em Symposium on Simplicity in Algorithms (SOSA)}, volume~69,
  pages 13:1--13:8, 2019.

\bibitem[HHS21]{Hanauer2021}
Kathrin Hanauer, Monika Henzinger, and Christian Schulz.
\newblock {Recent Advances in Fully Dynamic Graph Algorithms}.
\newblock {\em CoRR}, 2021.

\bibitem[HK73]{hopcroft1973n}
John~E. Hopcroft and Richard~M. Karp.
\newblock {An $n^{5/2}$ Algorithm for Maximum Matchings in Bipartite Graphs}.
\newblock {\em SIAM Journal on Computing}, 2(4):225--231, 1973.

\bibitem[HS22]{huang20221}
Shang-En Huang and Hsin-Hao Su.
\newblock $(1-\epsilon)$-approximate maximum weighted matching in distributed,
  parallel, and semi-streaming settings.
\newblock {\em arXiv preprint arXiv:2212.14425}, 2022.

\bibitem[Kap13]{kapralov2013better}
Michael Kapralov.
\newblock {Better Bounds for Matchings in the Streaming Model}.
\newblock In {\em the Proceedings of the Symposium on Discrete Algorithms
  (SODA)}, pages 1679--1697, 2013.

\bibitem[Kap21]{kapralov2021space}
Michael Kapralov.
\newblock {Space Lower Bounds for Approximating Maximum Matching in the Edge
  Arrival Model}.
\newblock In {\em the Proceedings of the Symposium on Discrete Algorithms
  (SODA)}, pages 1874--1893, 2021.

\bibitem[KKS14]{kapralov2014approximating}
Michael Kapralov, Sanjeev Khanna, and Madhu Sudan.
\newblock {Approximating Matching Size from Random Streams}.
\newblock In {\em the Proceedings of the Symposium on Discrete Algorithms
  (SODA)}, pages 734--751, 2014.

\bibitem[KMM12]{konrad2012maximum}
Christian Konrad, Fr{\'e}d{\'e}ric Magniez, and Claire Mathieu.
\newblock {Maximum Matching in Semi-Streaming with Few Passes}.
\newblock In {\em Approximation, Randomization, and Combinatorial Optimization:
  Algorithms and Techniques (APPROX/RANDOM)}, pages 231--242. Springer, 2012.

\bibitem[KMNFT20]{kapralov2020space}
Michael Kapralov, Slobodan Mitrovi{\'c}, Ashkan Norouzi-Fard, and Jakab Tardos.
\newblock Space efficient approximation to maximum matching size from uniform
  edge samples.
\newblock In {\em Proceedings of the Fourteenth Annual ACM-SIAM Symposium on
  Discrete Algorithms}, pages 1753--1772. SIAM, 2020.

\bibitem[KN21]{konrad2021two}
Christian Konrad and Kheeran~K Naidu.
\newblock On two-pass streaming algorithms for maximum bipartite matching.
\newblock {\em arXiv preprint arXiv:2107.07841}, 2021.

\bibitem[Kon15]{konrad2015maximum}
Christian Konrad.
\newblock {Maximum Matching in Turnstile Streams}.
\newblock In {\em the Proceedings of the European Symposium on Algorithms
  (ESA)}, pages 840--852, 2015.

\bibitem[KS95]{Kalantari1995}
Bahman Kalantari and Ali Shokoufandeh.
\newblock {Approximation Schemes For Maximum Cardinality Matching}.
\newblock Technical report, Rutgers University, 1995.

\bibitem[KT17]{kale2017maximum}
Sagar Kale and Sumedh Tirodkar.
\newblock {Maximum Matching in Two, Three, and a Few More Passes Over Graph
  Streams}.
\newblock In {\em Approximation, Randomization, and Combinatorial Optimization.
  Algorithms and Techniques (APPROX/RANDOM)}, volume~81, pages 15:1--15:21,
  2017.

\bibitem[KVV90]{karp1990optimal}
Richard~M. Karp, Umesh~V. Vazirani, and Vijay~V. Vazirani.
\newblock {An Optimal Algorithm for On-line Bipartite Matching}.
\newblock In {\em the Proceedings of the Symposium on Theory of Computing
  (STOC)}, pages 352--358, 1990.

\bibitem[LPSP15]{lotker2015improved}
Zvi Lotker, Boaz Patt-Shamir, and Seth Pettie.
\newblock {Improved Distributed Approximate Matching}.
\newblock {\em Journal of the ACM (JACM)}, 62(5):1--17, 2015.

\bibitem[McG05]{mcgregor2005finding}
Andrew McGregor.
\newblock {Finding Graph Matchings in Data Streams}.
\newblock In {\em Approximation, Randomization and Combinatorial Optimization:
  Algorithms and Techniques (APPROX/RANDOM)}, pages 170--181. Springer, 2005.

\bibitem[McG14]{mcgregor2014graph}
Andrew McGregor.
\newblock {Graph Stream Algorithms: A Survey}.
\newblock {\em ACM SIGMOD Record}, 43(1):9--20, 2014.

\bibitem[MP80]{munro1980selection}
J.~Ian Munro and Mike~S. Paterson.
\newblock {Selection and Sorting with Limited Storage}.
\newblock {\em Theoretical Computer Science}, 12(3):315--323, 1980.

\bibitem[Mut05]{muthukrishnan2005data}
Shanmugavelayutham Muthukrishnan.
\newblock {\em Data streams: Algorithms and applications}.
\newblock Now Publishers Inc, 2005.

\bibitem[MV80]{micali1980v}
Silvio Micali and Vijay~V. Vazirani.
\newblock {An $O(\sqrt{ \vert V \vert} \vert E \vert)$ Algorithm for Finding
  Maximum Matching in General Graphs}.
\newblock In {\em the Proceedings of the Symposium on Foundations of Computer
  Science (FOCS)}, pages 17--27, 1980.

\bibitem[MY11]{mahdian2011online}
Mohammad Mahdian and Qiqi Yan.
\newblock {Online Bipartite Matching with Random Arrivals: An Approach Based on
  Strongly Factor-Revealing LPs}.
\newblock In {\em the Proceedings of Symposium on Theory of Computing (STOC)},
  pages 597--606, 2011.

\bibitem[Ona18]{onak2018round}
Krzysztof Onak.
\newblock {Round Compression for Parallel Graph Algorithms in Strongly
  Sublinear Space}.
\newblock {\em arXiv preprint arXiv:1807.08745}, 2018.

\bibitem[Pre99]{preis1999linear}
Robert Preis.
\newblock {Linear Time $1/2$-Approximation Algorithm for Maximum Weighted
  Matching in General Graphs}.
\newblock In {\em the Proceedings of the Symposium on Theoretical Aspects of
  Computer Science (STACS)}, pages 259--269, 1999.

\bibitem[PS17]{paz20172+}
Ami Paz and Gregory Schwartzman.
\newblock {A $(2+\varepsilon)$-Approximation for Maximum Weight Matching in the
  Semi-Streaming Model}.
\newblock In {\em the Proceedings of the Symposium on Discrete Algorithms
  (SODA)}, pages 2153--2161, 2017.

\bibitem[Sch03]{schrijver2003combinatorial}
Alexander Schrijver.
\newblock {\em {Combinatorial Optimization: Polyhedra and Efficiency}},
  volume~24.
\newblock Springer Science \& Business Media, 2003.

\bibitem[Sol16]{solomon2016fully}
Shay Solomon.
\newblock {Fully Dynamic Maximal Matching in Constant Update Time}.
\newblock In {\em the Proceedings of the Symposium on Foundations of Computer
  Science (FOCS)}, pages 325--334, 2016.

\bibitem[Tir18]{tirodkar2018deterministic}
Sumedh Tirodkar.
\newblock {Deterministic Algorithms for Maximum Matching on General Graphs in
  the Semi-Streaming Model}.
\newblock In {\em the Proceedings of the Conference on Foundations of Software
  Technology and Theoretical Computer Science (FSTTCS)}, 2018.

\bibitem[Zel12]{zelke2012weighted}
Mariano Zelke.
\newblock {Weighted Matching in the Semi-Streaming Model}.
\newblock {\em Algorithmica}, 62(1-2):1--20, 2012.

\end{thebibliography}

\end{document}